\documentclass[11pt]{article}%
\usepackage[latin9]{inputenc}
\usepackage{amsmath}
\usepackage{amssymb}
\usepackage[authoryear]{natbib}
\usepackage{float}
\usepackage{graphicx}
\usepackage{amsfonts}
\usepackage{amsthm}
\usepackage{enumerate}
\usepackage{epsfig}
\usepackage{ifthen}
\usepackage{latexsym}
\usepackage{syntonly}
\usepackage{rotating}
\usepackage{lscape}
\usepackage{color}
\usepackage{booktabs}
\usepackage[bottom]{footmisc}
\usepackage{epstopdf}
\usepackage{makeidx}
\usepackage{xr}%
\providecommand{\U}[1]{\protect\rule{.1in}{.1in}}
\newtheorem{theorem}{Theorem}
\newtheorem{corollary}{Corollary}

\newtheorem{lemma}{Lemma}
\newtheorem{proposition}{Proposition}

\theoremstyle{definition}

\def\baselinestretch{1.3}

\allowdisplaybreaks
\begin{document}

\title{Wealth Effect on Portfolio Allocation in Incomplete Markets\thanks{%
\protect\linespread{1.1}\protect\scriptsize We would like to thank
T.\ Berrada, J.\ Detemple, I. Karatzas, R. Koijen, H.\ Langlois, P.\ Glasserman, and
M.\ Rindisbacher for valuable comments and suggestions. We also benifited from
the comments of participants at the 2018 3rd (resp. the 2019 4th) PKU-NUS
Annual International Conference on Quantitative Finance and Economics, the
2019 Myron Scholes Financial Forum at Nanjing University, and seminar at
Columbia University. The research of Chenxu Li was supported by the Guanghua
School of Management, the Center for Statistical Science, the High-performance
Computing Platform, and the Key Laboratory of Mathematical Economics and
Quantitative Finance (Ministry of Education) at Peking University, as well as
the National Natural Science Foundation of China (Grant 71671003).%
}}
\author{Chenxu Li\thanks{%
\protect\linespread{1.1}\protect\scriptsize Address: Guanghua School of
Management, Peking University, Beijing, 100871, P. R. China. E-mail address:
\texttt{cxli@gsm.pku.edu.cn.}%
}\\{\footnotesize Guanghua School of Management}\\{\footnotesize Peking University}
\and Olivier Scaillet\thanks{%
\protect\linespread{1.1}\protect\scriptsize Address: University of
Geneva and Swiss Finance Institute, Bd du Pont d'Arve 40, CH - 1211 Geneve 4,
Suisse. E-mail address: \texttt{olivier.scaillet@unige.ch.}%
}\\{\footnotesize University of Geneva and}\\{\footnotesize Swiss Finance Institute}
\and Yiwen Shen\thanks{%
\protect\linespread{1.1}\protect\scriptsize
Address: Department of Information Systems, Business Statistics and Operations Management, The Hong Kong University of Science and Technology, Kowloon, Hong Kong SAR. E-mail
address: \texttt{yiwenshen@ust.hk. }%
}\\{\footnotesize The Hong Kong University of }\\{\footnotesize Science and Technology} }
\date{\today }
\maketitle

\begin{abstract}
We develop a novel five-component decomposition of optimal dynamic portfolio
choice, which reveals the simultaneous impacts from market incompleteness and
wealth-dependent utilities. Under the HARA utility and a nonrandom interest
rate, we can explicitly solve for the optimal policy as a combination of a
bond holding scheme and the corresponding simpler CRRA strategy. Under a
stochastic volatility model estimated on US equity data, we use closed-form
solution to demonstrate the sophisticated impacts from the wealth-dependent
utilities, including cycle-dependence and hysteresis effect in optimal
portfolio allocation, as well as a risk-return trade-off in investment performance.



\end{abstract}

\newpage

\section{Introduction}

Optimal portfolio choice is a central topic in modern financial economics,
drawing continuous attention from both academia and industry. Hedge funds,
asset management firms, and pension funds, which manage large positions of
portfolios, as well as individual investors, are confronted with this type of
decisions frequently. The optimal portfolio choice problem has also drawn
long-standing interest in academic research. The celebrated static
mean-variance framework of \cite{markowitz1952portfolio} laid a foundation.
Following the seminal work by \cite{Samuelson69} and \cite{merton69,merton71},
various studies have been developed for the optimal dynamic portfolio choice;
see the comprehensive surveys in, e.g., \cite{Detemple2014},
\cite{Brandt_2009_survey}, and \cite{wachter2010asset}, as well as the
references therein. As an optimal stochastic control problem in a
continuous-time setting, studying optimal policies usually relies on two
approaches. The first one is the well-known dynamic programming method, which
employs the highly nonlinear Hamilton-Jacobi-Bellman (HJB hereafter) equation
to characterize the optimal policy. The second one is the martingale method
pioneered and developed by, e.g., \cite{pliska86},
\cite{KaratzasLehoczkyShreve87}, \cite{cox1989optimal},
\cite{OconeKaratzas1991}, \cite{Cvitanic_Karatzas_1992_AAP},
\cite{Detemple_2003_JF}, and \cite{Koijen2014} among others.

For the purpose of understanding optimal policies from an economic
perspective, the decomposition of optimal policies into mean variance and
hedging components was initiated by \cite{merton71} and matured as a
state-of-the-art approach; see, e.g., \cite{Detemple_2003_JF},
\cite{Detemple_Rindisbacher_2005_MF,DetempleRindisbacher2010}, and
\cite{Detemple2014}.\footnote{See also \cite{Basak_Chabakuri_2010} for
establishing an analytical characterization of optimal portfolios for
mean-variance investors.} For the purpose of implementing and analyzing the
behavior of optimal portfolios, existing works largely focus on specific
affine models (see, e.g., \cite{duffiepansingleton00}) and wealth-independent
utilities, such as the basic constant relative risk aversion (CRRA) utility
and its generalization to the recursive utility by separating risk aversion
and intertemporal elasticity of substitution (see, e.g.,
\cite{epstein1989risk} and \cite{duffie1992stochastic}). While these
specifications can bring mathematical convenience, e.g., closed-form optimal
portfolio policies in some rare and limited cases\footnote{See, e.g.,
\cite{Kim96} and \cite{Wachter02} for modeling stochastic market price of risk
of the asset by using an Ornstein-Uhlenbeck model, \cite{LiouiPoncet2001} for
considering stochastic interest rates by employing a constant-parameter
instantaneous forward rate model, \cite{Pan_Liu_2003_JFE} for discussing
dynamic derivative strategies, \cite{Liu_Longstaff_Pan_2003_JF} for studying
impacts of event risk via affine stochastic volatility models with jumps,
\cite{liu2007portfolio} for taking various stochastic environments (e.g.,
stochastic volatility) into account by modeling the asset return via quadratic
affine processes, and \cite{Burraschietal2010} for characterizing hedging
components against both stochastic volatility and correlation risks under
Wishart processes.}, they limit the model capacity to capture empirically
flexible market dynamics and investor preferences, e.g., wealth-dependent
utilities. As an effective method that can be applied to flexible diffusion
models without closed-form policies, \cite{Detemple_2003_JF} develop a Monte
Carlo simulation approach based on the above-mentioned decomposition of the
optimal policy. However, this milestone of methods is by far limited to the
complete market setting.\footnote{\cite{MC_Cvitanic_2003} develop an
alternative simulation approach. Besides, other numerical methods are
proposed; see, e.g., the early attempts based on the dynamic programming
approach; see, e.g., \cite{FF_MOR_1991}, \cite{Brennan_Schwartz_Lagnado_1997},
\cite{brennan1998role}, \cite{chacko2005dynamic}, \cite{brennan2002dynamic},
and \cite{campbell2004strategic}. We refer to the recent book of
\cite{Dumas_book2017} for a survey of different numerical methods available
for optimal portfolio choice.} Now, a notable open question is to obtain an
economically insightful decomposition of optimal policies under the more
flexible and realistic setting with general incomplete market dynamics and
wealth-dependent utilities, e.g., the hyperbolic absolute risk aversion (HARA)
utility. We expect such a decomposition to reflect the fundamental impacts
from these two effects at the theoretical level, i.e., how market
incompleteness and wealth-dependent utility, as well as their interaction,
affect the optimal policy. Subsequently, the decomposition facilitates the
implementation of the optimal policy via either closed-form solutions or
potential numerical methods, and allows us to obtain new economic insights by
flexible empirical applications. These challenges are exactly our focus.

In this paper, we develop and implement a decomposition for the optimal
portfolio policy under a general class of incomplete market diffusion models
and flexible utilities over both intermediate consumption and terminal wealth.
In our decomposition, we express each component of the optimal policy as
conditional expectation of random variables with sophisticated but explicit
dynamics. In an incomplete market, investors cannot fully hedge the risk by
investing in the risky assets. As a preparation, we apply and explore the
\textquotedblleft least favorable completion\textquotedblright\ principle
developed by \cite{karatzas_martingale_1991} under general diffusion models.
It completes the market by introducing suitable fictitious assets. Then, we
establish the equivalence between the optimal policy in the completed market
and that in the original market by choosing the appropriate price of risk for
these fictitious assets. Such price of risk is endogenously determined by the
investor's utility function and investment horizon, and is thus referred to as
the investor-specific price of risk in incomplete market models. It is also
known as the \textquotedblleft shadow price\textquotedblright\ of market
incompleteness in the literature; see, e.g., \cite{DetempleRindisbacher2010}.

We begin by revealing and applying the following structure: Under general
incomplete market models with wealth-dependent utilities, the appropriate
investor-specific price of risk depends not only on the current market state,
but also implicitly on investor's wealth level. The latter dependence is
entirely absent in the market price of risk associated with the real assets.
Due to this special structure, the state price density, which plays a key role
in solving the optimal policy, also becomes wealth-dependent in general
incomplete market models. It is fundamentally different from the complete
market case, where the state price density only depends on the market state.
Thus, the structure of investor-specific price of risk introduces additional
channels for the investor's wealth level to affect the optimal policy.
Recognizing such structure is crucial for correctly establishing the
decomposition of optimal policy under general incomplete market models with
flexible utilities.

We first decompose the optimal policy into four components, all as functions
of the current market state variable and investor's wealth level. The optimal
policy includes the mean-variance component, the interest rate hedge
component, and two components for hedging the uncertainty in market and
investor-specific price of risk respectively. Thus, our results substantially
extend the existing complete-market policy decomposition with the
mean-variance, interest rate hedge, and market price of risk hedge; see, e.g.,
\cite{merton71} and \cite{Detemple_2003_JF}. As the investor-specific price of
risk depends on investor's wealth level, the component related to it needs to
hedge the uncertainty from both the market state and investor's wealth. It
highlights a key difference from the complete market case. In particular,
under the mild assumption that the investor-specific price of risk function is
differentiable in its arguments, we show the investor-specific price of risk
hedge component can be further decomposed into two parts, which hedge the
uncertainty from market state and investor's wealth respectively. It
effectively leads to a five-component decomposition of the optimal
policy,\ and highlights the fundamental differences between our decomposition
and the ones in the literature, e.g., \cite{merton71}, \cite{Detemple_2003_JF}%
, \cite{Detemple_Rindisbacher_2005_MF,DetempleRindisbacher2010}, and
\cite{Detemple2014}. To our best knowledge, we are the first to identify the
last component in the optimal policy, which hedges the uncertainty in
investor-specific price of risk due to variation in investor's wealth. We find
this component only appears in the optimal policy when the market is
incomplete and, at the same time, the utility is wealth-dependent. Thus, it
reveals the impact from the combination of market incompleteness and
wealth-dependent utility on optimal policies.

Our new decomposition, due to its structural clarity, facilitates the
implementation of optimal policy under flexible incomplete-market model with
wealth-dependent utility. As an important application, we apply our
decomposition to solve the optimal policy under the wealth-dependent HARA
utility. Compared with the simple CRRA utility that is commonly used in the
literature, the HARA utility offers more flexibility in capturing investor's
preference, and is more realistic in reflecting investment constraints such as
investment goal and subsistence level. However, it is much less studied due to
its commonly believed mathematical inconvenience; see, e.g., \cite{Kim96} for
a rare case with closed-form policy and \cite{duffie1997hedging} for
characterizing the optimal policy by viscosity solution in an incomplete
market with constant coefficients for stock and income dynamics. Our
decomposition clearly reveals how the investor's wealth level gets involved in
the optimal policy when transiting from the CRRA utility to the HARA utility.
In particular, under the special case with nonrandom but possibly time-varying
interest rate, we show that we can decompose the optimal policy under HARA
utility as a product of its counterpart under CRRA utility and a key
multiplier related to investor's wealth level and bond prices. Our
decomposition indicates that the HARA optimal policy is constructed as
follows: the investor first buys a series of zero-coupon bonds to exactly
satisfy the minimum requirements for terminal wealth and intermediate
consumptions, then allocates the remaining wealth just as an investor with
CRRA utility. Moreover, this intuitive structure suggests that the roles of
investor wealth and current market state can be explicitly separated in the
optimal policy under HARA utility. As such, our decomposition facilitates the
closed-form solution of HARA optimal policy under specific models. It also
brings potential benefits to the open problem of developing efficient
numerical approaches for incomplete market models with wealth-dependent utilities.

To demonstrate the wealth effects in optimal portfolio allocation, we analyze
the behaviors of optimal policies for investors with HARA utility under the
prototypical stochastic volatility model of \cite{heston93} (Heston-SV
hereafter). To achieve fresh empirical validity, we calibrate the model using
the daily data of SPDR S\&P 500 ETF from the recent ten years, i.e., 2010 to
2019. We employ the maximum likelihood estimation approach developed in
\cite{AitSahalia_Li_Li_2020_MMLE}, which is an efficient method for estimating
continuous-time models with latent factors. With our decomposition results for
HARA utility, we can solve the optimal policy in closed-form under the
Heston-SV model, which facilitates our subsequent economicl analysis. We show
that when switching from the simple CRRA utility to the HARA utility, the
optimal policy is impacted by not only the investor's wealth level, but also
the interest rate and investment horizon. Thus, the wealth-dependent property
of HARA utility can influence the optimal policy via multiple channels beyond
investor's current wealth. Our results provide a potential theoretical
explanation to the empirical evidence in the literature that the investment in
risky assets increases concavely in investor's financial wealth; see, e.g.,
\cite{Roussanov2010}, \cite{wachter2010household}, and \cite{calvet2014twin}.

In addition to the static analysis of the optimal policy at a fixed time
point, we analyze the significant wealth impacts from a dynamic perspective,
i.e., how the wealth-dependent property of the HARA utility interacts with the
complex market dynamics in affecting the optimal allocation strategy and
overall investment performance.\footnote{Although important, this type of
dynamic analysis is rare in the literature. \cite{MoreiraMuir2019} show that,
under stochastic volatility, ignoring the hedge component in the optimal
policy leads to a substantial utility loss. However, since
\cite{MoreiraMuir2019} employ the recursive utility of
\cite{duffie1992stochastic}, the consequent optimal policy is independent of
investor's wealth level, contrasting\ with our focus on wealth-dependence.} In
particular, we find that the wealth-dependent property of HARA utility leads
to sophisticated dependence of the optimal portfolio on the entire paths of
the asset price and volatility. As a vivid illustration, we find that the
optimal policies of HARA investors with high and low initial wealth levels
become more (resp.\ less) similar to each other in the bull (resp.\ bear)
market regimes. Next, we show the initial wealth level of HARA investors
substantially impacts the overall investment performance. It introduces a
risk-return trade-off induced by the wealth level: the HARA investor with
higher initial wealth allocates more wealth to the risky asset, leading to a
higher return but also more risk. We quantify such trade-off by simulating a
large number of paths under our estimated model, i.e., accounting for the
probabilistic distribution of possible market scenarios. With the simulated
paths and the closed-form optimal policy, we compute a priori expectations of
performance statistics including excess return mean, volatility, Sharpe ratio,
and maximum drawdown. We find that as HARA investor's initial wealth increases
by ten times from the subsistence level, the average annual excess return
increases from $3.5\%$ to $23.5\%$, while the volatility increases from
$4.2\%$ to $25.4\%$, and the maximum drawdown jumps from $8.2\%$ to $38.1\%$.
The huge differences in the investment performance exemplify the importance
and practical relevance of understanding the wealth effects in delegated
portfolio management. Moreover, the average Sharpe ratio increases from $0.74$
to $0.82$. This surge can be attributed to the additional uncertainty caused
by the market cycles, which is more pronounced for low wealth HARA investors.
The above cycle-dependence is entirely absent under the simple CRRA utility.

Then, we perform some more in-depth analysis on how the wealth effect impacts
the overall investment performance. First, we find that higher interest rate
and/or longer investment horizon increase both the average excess return and
volatility. It can be interpreted as follows. With higher interest rate or
longer horizon, investors can accumulate more wealth during their investment
horizon, leading to more allocation on the risky asset. In addition, such
effect is more significant for HARA investors with low initial wealth,
reflecting the interaction between the wealth effect and model parameters.
Finally, we reveal a novel hysteresis effect in the optimal portfolio
allocation. That is, the investment performance depends on not only the
realization of market states, but also the history of their occurrence. In
particular, when we shuffle the path of stock price and move the ``good''
years with high returns to the beginning of the investment horizon, the
average excess return of HARA investors remains almost unchanged, but their
volatility increases substantially. Such hysteresis effect essentially stems
from the cycle-dependence of optimal policy under the wealth-dependent HARA
utility, and entirely vanishes for CRRA investors. These analyses further
suggest that the wealth-dependent property can impact the optimal portfolio
allocation in sophisticated ways, which must be taken in to account in
delegated portfolio management.

The rest of this paper is organized as follows. Section \ref{section:model}
gives the model set-up and recalls the machinery of the fictitious completion
method. In Section \ref{section:structural_decomposition}, we establish the
decomposition for general incomplete market models with flexible utilities. We
reveal the fundamental impact of market incompleteness and wealth-dependent
utilities by comparing our decomposition to that under the CRRA utility and/or
complete market cases. In Section \ref{section:closed_form}, we apply our
decomposition to the wealth-dependent HARA utility and establish a closed-form
relation between optimal policies under HARA and CRRA utilities under
nonrandom interest rate. We devote Section \ref{section:wealth_impact} to the
economic analysis of the static and dynamic impacts of wealth-dependence. In
particular, we provide explanations for the cycle-dependence and hysteresis
effect of optimal policy under wealth-dependent utilities, which is unseen in
the previous literature on optimal dynamic portfolio allocation. Section
\ref{Section_conclusions_discussions} concludes and provides discussions. We
collect auxiliary derivations in the appendices. For completeness, we include
further results in the online supplementary material.

\section{Model set-up and fictitious completion method \label{section:model}}

We begin by setting up the model, utility function, and optimal dynamic
portfolio choice problem under general incomplete market models before
recalling the fictitious completion method used to get our new decomposition results.

\subsection{Model set-up}

Assume that the market consists of $m$ stocks and one savings account. The
stock price $S_{it},$ for $i=1,2,\ldots,m$, follows the generic SDE:
\begin{equation}
\frac{dS_{it}}{S_{it}}=\left(  \mu_{i}(t,Y_{t})-\delta_{i}(t,Y_{t})\right)
dt+\sigma_{i}(t,Y_{t})dW_{t}, \label{SDE_price}%
\end{equation}
where $Y_{t}$ is an $n$-dimensional state variable driven by the following
generic SDE:
\begin{equation}
dY_{t}=\alpha(t,Y_{t})dt+\beta(t,Y_{t})dW_{t}. \label{SDE_Y}%
\end{equation}
In (\ref{SDE_price}), $W_{t}$ is a standard $d-$dimensional Brownian motion;
$\mu_{i}(t,y)$ is a scalar function for modeling the mean rate of return;
$\delta_{i}(t,y)$ is a scalar function for modeling the dividend rate;
$\sigma_{i}(t,y)$ is a $d-$dimensional vector-valued function for modeling the
volatility. In (\ref{SDE_Y}), $\alpha(t,y)$ is an $n-$dimensional
vector-valued function for modeling the drift of the state variable $Y_{t};$
$\beta(t,y)$ is an $n\times d$ dimensional matrix-valued function for modeling
the diffusion of the state variable $Y_{t}.$ We assume the existence and
uniqueness of solutions to SDEs (\ref{SDE_price}) and (\ref{SDE_Y}%
).\footnote{Sufficient technical conditions include but are not limited to,
e.g., the Lipschitz condition and the polynomial growth condition on the
coefficient functions; see, e.g., Sections 5.2 and 5.3 in
\cite{Karaztas_Shreve}. In addition, we assume that the solution to SDE
(\ref{SDE_Y}) is Malliavin differentiable. A sufficient condition is that the
coefficient functions of SDE (\ref{SDE_Y}) have bounded derivatives for each
order; see, e.g., Section 2.2 of \cite{Nualart_MC_book}.} Besides, we assume
that the savings account appreciates at the instantaneous interest rate
$r_{t}=r(t,Y_{t})$ for some scalar-valued function $r(t,y).$ The state
variable $Y_{t}$ governs all the investment opportunities in the market
through the rate of return, the dividend rate, the volatility, and the
instantaneous interest rate.

We mainly focus on the incomplete market case where the number of independent
Brownian motions is strictly larger than the number of tradable risky assets,
i.e., $d>m$. In this case, we cannot fully hedge the uncertainty stemming from
the Brownian motion by investing in the risky assets. As we will show, due to
market incompleteness, the decomposition and subsequent implementation for the
optimal portfolio policy become a challenging issue; it enjoys a
high-dimensional nature with a sophisticated structure even for one-asset
cases. Denote the investor's wealth process by $X_{t}$. Then, it satisfies the
following wealth equation:
\begin{equation}
dX_{t}=(r_{t}X_{t}-c_{t})dt+X_{t}\pi_{t}^{\top}\left[  (\mu_{t}-r_{t}%
1_{m})dt+\sigma_{t}dW_{t}\right]  . \label{opt_constraint}%
\end{equation}
In (\ref{opt_constraint}), $\mu_{t}$ and $\sigma_{t}$ represent the mean rate
of return and volatility of the risky assets, which satisfy $\mu_{t}%
=\mu(t,Y_{t})$ and $\sigma_{t}=\sigma(t,Y_{t})$ where the $m$--dimensional
vector $\mu(t,y)$ and the $m\times d$--dimensional matrix $\sigma(t,y)$ are
defined by $\mu(t,y):=(\mu_{1}(t,y),\mu_{2}(t,y),\cdots, $ $\mu_{m}%
(t,y))^{\top}$ and $\sigma(t,y):=$ $(\sigma_{1}(t,y),\sigma_{2}(t,y),\cdots,$
$\sigma_{m}(t,y))^{\top}$. We assume the volatility function $\sigma(t,y)$ has
rank $m,$ i.e., its rows are linearly independent. Besides, $c_{t}$ is the
instantaneous consumption rate; $\pi_{t}$ is an $m-$dimensional vector
representing the weights of the risky assets in the portfolio; $1_{m}$ denotes
an $m-$dimensional column vector with all elements equal to $1.$ The investor
maximizes her expected utility over both intermediate consumptions and
terminal wealth by dynamically allocating her wealth among the risky assets
and the risk-free asset, subject to the non-bankruptcy condition. We can
formulate this optimization problem as
\begin{equation}
\sup_{(\pi_{t},c_{t})}E\left[  \int_{0}^{T}u(t,c_{t})dt+U(T,X_{T})\right]
,\text{ with }X_{t}\geq0\text{ for all }t\in\lbrack0,T], \label{utility}%
\end{equation}
where $u(t,\cdot)$ and $U(T,\cdot)$ are the time-additive utility functions of
the intermediate consumptions and the terminal wealth; they are allowed to be
time-varying, in order to reflect the time value, e.g., the discount effect,
and are assumed to be strictly increasing and concave with $\lim
_{x\rightarrow\infty}\partial u(t,x)/\partial x=0$ and $\lim_{x\rightarrow
\infty}\partial U(T,x)/\partial x=0$.

Under the incomplete market setting, we aim at developing a useful
decomposition and methods of implementation for the optimal policy under
general time-additive utility functions. They are sufficiently flexible for
capturing the effects of market incompleteness and wealth-dependent property
in investor's risk preference. A representative example of wealth-dependent
utilities is the HARA utility. Following the convention (see, e.g.,
\cite{carroll1996concavity}), we define it by
\begin{equation}
u(t,c)=we^{-\rho t}\frac{\left(  c-\bar{c}\right)  ^{1-\gamma}}{1-\gamma
}\text{ and }U(T,x)=(1-w)e^{-\rho T}\frac{\left(  x-\bar{x}\right)
^{1-\gamma}}{1-\gamma} \label{hara_utility}%
\end{equation}
for $c>\bar{c}$ and $x>\bar{x},$ where $\gamma>0$ is the risk aversion
coefficient; $w\in\lbrack0,1]$ is the weight for balancing the intermediate
consumption and the terminal wealth, and $\rho$ is the discount rate. The
constants $\bar{c}$ and $\bar{x}$ represent the minimum allowable levels for
the intermediate consumption and terminal wealth. The HARA utility allows for
imposing lower bound constraints on the intermediate consumption and/or
terminal wealth. This realistic feature is particularly suitable for
incorporating, e.g., portfolio insurance, investment goal constraints, and
subsistence level constraints. Although important, closed-form optimal
policies under the HARA utility are rare due to its technical difficulties;
see \cite{Kim96} for one such case with the stochastic market price of risk
modelled by an Ornstein-Uhlenbeck process. Potential numerical methods
include, e.g., the Monte Carlo simulation approaches of
\cite{Detemple_2003_JF} and \cite{MC_Cvitanic_2003}, which are, however,
developed under complete market settings.

As a simpler and special case, the HARA utility reduces to the widely used
CRRA utility when $\bar{c}$ and $\bar{x}$ are set as zero in
(\ref{hara_utility}).\footnote{The wealth independence nature of CRRA utility
brings mathematical convenience that leads to closed-form formulae of the
optimal policy or significant simplifications of the optimization problem
under some specific models; see, e.g., \cite{Wachter02}, and
\cite{liu2007portfolio} for closed-form optimal policies, as well as
\cite{Detemple_2003_JF} for a Monte Carlo simulation approach.} An alternative
generalization of the simple CRRA utility is the recursive utility that aims
at separating risk aversion from elasticity of intertemporal substitution;
see, e.g., \cite{epstein1989risk} and \cite{duffie1992stochastic}. However,
most applications with the recursive utility inherit the wealth-independent
property of the CRRA utility (see, e.g., \cite{chacko2005dynamic} and
\cite{MoreiraMuir2019}), thus cannot reveal the fundamental impact of
investor's wealth on the optimal policy. We do not intend to cover the
recursive utility in this paper, as the general time-additive utility we
consider is flexible enough for our focus on the sophisticated impacts of
market incompleteness and wealth dependence on both theoretical and practical
aspects. For ease of exposition, we abbreviate time-additive utility as
utility in what follows.

\subsection{Fictitious completion method}

To establish a novel decomposition of the optimal policy, we set as a basic
foundation the least favorable completion\footnote{As documented in the
literature (see, e.g., \cite{karatzas_martingale_1991}), we can interpret the
terminology \textquotedblleft least favorable completion\textquotedblright\ as
follows. Consider all the possible fictitious completions and their associated
optimal policies. We naturally say that a completion is more (resp.\ less)
favorable if its corresponding optimal policy results in higher (resp.\ lower)
expected utility. The completion (\ref{Pi_F_0}) below, which leads to an
optimal portfolio with zero weight on the fictitious assets, must be the least
favorable one. Indeed, in any other fictitious completion, since this
portfolio without the fictitious assets is admissible (i.e., a candidate
portfolio strategy), the optimal one must result in a higher expected utility
and thus becomes more favorable.} principle introduced in
\cite{karatzas_martingale_1991} at the purely theoretical level. For this
purpose, we begin by introducing the following notations as necessary
preparations, based on which we explore and disentangle the essential
structures of the optimal policy in the next section.

We introduce $d-m$ fictitious assets without dividend payment to complete the
market, following \cite{karatzas_martingale_1991}. Their prices $F_{it}$, for
$i=1,2,\ldots,d-m$, satisfy the following SDE:
\begin{equation}
\frac{dF_{it}}{F_{it}}=\mu_{it}^{F}dt+\sigma_{i}^{F}(t,Y_{t})dW_{t},
\label{SDE_F}%
\end{equation}
where the mean rates of returns $\mu_{it}^{F}$ are stochastic processes
adaptive to the filtration generated by the Brownian motion $W_{t}$. We can
choose the volatility function $\sigma^{F}(t,y):=(\sigma_{1}^{F}%
(t,y),\cdots,\sigma_{d-m}^{F}(t,y))^{\top}$ of the fictitious assets
arbitrarily, as long as it has rank $d-m$ and satisfies the following
orthogonal condition with respect to the volatility function $\sigma(t,y)$ of
the real risky assets $S_{t}$:
\begin{equation}
\sigma(t,y)\sigma^{F}(t,y)^{\top}\equiv0_{m\times(d-m)}. \label{ortho_cond}%
\end{equation}
It guarantees that the fictitious and real assets are driven by different
Brownian shocks, and thus leads to the success of the market completion.

Combining the $m$ real risky assets with prices $S_{t}$ in (\ref{SDE_price})
and the $d-m$ fictitious risky assets with prices $F_{t}$ in (\ref{SDE_F}), we
construct a completed market consisting of $d$ risky assets and driven by $d$
independent Brownian motions. In this completed market, we represent the
prices of the risky assets, including both the real and fictitious ones, by a
$d$--dimensional column vector $\mathcal{S}_{t}=(S_{t}^{\top},F_{t}^{\top
})^{\top}$. According to (\ref{SDE_price}) and (\ref{SDE_F}), $\mathcal{S}%
_{t}$ is driven by the SDE: $d\mathcal{S}_{t}=$diag$(\mathcal{S}_{t})\left[
\mu_{t}^{\mathcal{S}}dt+\sigma^{\mathcal{S}}(t,Y_{t})dW_{t}\right]  ,$ with
the diagonal matrix diag$(\mathcal{S}_{t})=$ diag$(S_{t},F_{t}),$ the
$d$--dimensional column vector $\mu_{t}^{\mathcal{S}}=((\mu(t,Y_{t}%
)-\delta(t,Y_{t}))^{\top},(\mu_{t}^{F})^{\top})^{\top}$, and the $d\times d$
dimensional matrix $\sigma^{\mathcal{S}}(t,Y_{t})$ $=(\sigma(t,Y_{t})^{\top
},\sigma^{F}(t,Y_{t})^{\top})^{\top}$. By linear algebra, the orthogonal
condition (\ref{ortho_cond}) implies that $\sigma^{\mathcal{S}}(t,y)$ must be
nonsingular. Thus, we are now in a complete market, where we can fully hedge
the uncertainty stemming from all Brownian motions. The completed market
allows for investing in both the real assets $S_{t}$ and the fictitious assets
$F_{t}$\emph{. }We denote by $\pi_{t}$ and $\pi_{t}^{F}$ their corresponding
weights, which are $m$ and $(d-m)$--dimensional vectors. Similar to
(\ref{utility}), we consider the utility maximization problem in this
completed market, still with the non-bankruptcy constraint $X_{t}\geq0$.

In the completed market, we define the total price of risk as $\theta
_{t}^{\mathcal{S}}:=\sigma^{\mathcal{S}}(t,Y_{t})^{-1}(\mu_{t}^{\mathcal{S}%
}-r(t,Y_{t})1_{d}).$ By the orthogonal condition (\ref{ortho_cond}), we can
decompose the total price of risk as:
\begin{equation}
\theta_{t}^{\mathcal{S}}=\theta^{h}(t,Y_{t})+\theta_{t}^{u}. \label{mpr_total}%
\end{equation}
Here, $\theta^{h}(t,Y_{t})$ and $\theta_{t}^{u}$ are the prices of risk
associated with the real and fictitious assets, respectively. They are defined
by the $d$--dimensional column vectors:%
\begin{subequations}
\begin{equation}
\theta^{h}(t,Y_{t}):=\sigma(t,Y_{t})^{+}(\mu(t,Y_{t})-r(t,Y_{t})1_{m})
\label{def_thetah}%
\end{equation}
and%
\begin{equation}
\theta_{t}^{u}:=\sigma^{F}(t,Y_{t})^{+}(\mu_{t}^{F}-r(t,Y_{t})1_{d-m}),
\label{def_thetau}%
\end{equation}
where $A^{+}:=A^{\top}(AA^{\top})^{-1}$ denotes the Moore--Penrose inverse
(see, e.g., \cite{penrose_1955}) of a general matrix $A$ with linearly
independent rows. The term $\theta^{h}(t,Y_{t})$ in (\ref{def_thetah}) is
referred to as the market price of risk, as it is fully determined by the real
assets shared by all investors in the market. The term $\theta_{t}^{u}$ in
(\ref{def_thetau}), however, is purely associated with the fictitious assets,
which are specifically introduced for solving the optimal portfolio choice
problem (\ref{utility}) in the incomplete market. As we will show momentarily,
$\theta_{t}^{u}$ is endogenously determined by the investor's utility function
and the investment horizon. Thus, following the literature, we refer to
$\theta_{t}^{u}$ as the investor-specific price of risk, since it varies from
one investor to another. Our investigation starts from expressing the
functional form of $\theta_{t}^{u}$. Then, we develop a novel decomposition of
the optimal policy based on the structure of $\theta_{t}^{u}$.

With the total market price of risk in (\ref{mpr_total}), we introduce the
state price density as%
\end{subequations}
\begin{equation}
\xi_{t}^{\mathcal{S}}:=\exp\left(  -\int_{0}^{t}r(v,Y_{v})dv-\int_{0}%
^{t}(\theta_{v}^{\mathcal{S}})^{\top}dW_{v}-\frac{1}{2}\int_{0}^{t}(\theta
_{v}^{\mathcal{S}})^{\top}\theta_{v}^{\mathcal{S}}dv\right)  .\footnote{To
guarantee the martingale property of $\xi_{t}^{\mathcal{S}}\exp(\int_{0}%
^{t}r(v,Y_{v})dv)$, we assume that the total price of risk $\theta
_{v}^{\mathcal{S}}$ satisfies the Novikov condition: $\displaystyle E\left[
\exp\left(  \frac{1}{2}\int_{0}^{T}(\theta_{v}^{\mathcal{S}})^{\top}\theta
_{v}^{\mathcal{S}}dv\right)  \right]  <\infty.$} \label{def_Xit_incomp}%
\end{equation}
For any $s\geq t\geq0,$ we define the relative state price density as
$\xi_{t,s}^{\mathcal{S}}=\xi_{s}^{\mathcal{S}}/\xi_{t}^{\mathcal{S}}.$ By
Ito's formula, it satisfies $d\xi_{t,s}^{\mathcal{S}}=-\xi_{t,s}^{\mathcal{S}%
}[r(s,Y_{s})ds+(\theta_{s}^{\mathcal{S}})^{\top}dW_{s}] $ with initial value
$\xi_{t,t}^{\mathcal{S}}=1$.\emph{ }The above dynamics of $\xi_{t,s}%
^{\mathcal{S}}$ clearly hinges on the undetermined investor-specific price of
risk $\theta_{s}^{u}$.

The martingale approach pioneered by \cite{KaratzasLehoczkyShreve87} and
\cite{cox1989optimal} starts by formulating the dynamic problem (\ref{utility}%
) with information up to time $t$ as the following equivalent static
optimization problem:
\begin{equation}
\sup_{(c_{t},X_{T})}E_{t}\left[  \int_{t}^{T}u(s,c_{s})ds+U(T,X_{T})\right]
\text{ subject to }E_{t}\left[  \int_{t}^{T}\xi_{t,s}^{\mathcal{S}}c_{s}%
ds+\xi_{t,T}^{\mathcal{S}}X_{T}\right]  \leq X_{t}, \label{Max_expectation}%
\end{equation}
where, throughout the paper, $E_{t}$ denotes the expectation condition on the
information up to time $t$ and $X_{t}$ is the wealth level assuming that the
investor always follows the optimal policy. Then, following the standard
method of Lagrangian multiplier, we can represent the optimal intermediate
consumption and terminal wealth as $c_{t}=I^{u}\left(  t,\lambda_{t}^{\ast
}\right)  $ and $X_{T}=I^{U}\left(  T,\lambda_{T}^{\ast}\right)  $, with
$I^{u}(t,\cdot)$ and $I^{U}(t,\cdot)$ being the inverse marginal utility
functions of $u(t,\cdot)$ and $U(t,\cdot)$, i.e., the functions satisfying
$\partial u/\partial x(t,I^{u}(t,y))=y$ and $\partial U/\partial
x(t,I^{U}(t,y))=y$. The quantity $\lambda_{t}^{\ast}$ denotes the Lagrangian
multiplier for the wealth constraint in (\ref{Max_expectation}). It is
uniquely characterized by%
\begin{equation}
X_{t}=E_{t}[\mathcal{G}_{t,T}(\lambda_{t}^{\ast})], \label{XtGt}%
\end{equation}
where $\mathcal{G}_{t,T}(\lambda_{t}^{\ast})$ is defined as $\mathcal{G}%
_{t,T}(\lambda_{t}^{\ast}):=\Gamma_{t,T}^{U}(\lambda_{t}^{\ast})+\int_{t}%
^{T}\Gamma_{t,s}^{u}(\lambda_{t}^{\ast})ds.$ Here, $\Gamma_{t,T}^{U}%
(\lambda_{t}^{\ast})$ and $\Gamma_{t,s}^{u}(\lambda_{t}^{\ast})$ are given by%
\begin{subequations}
\begin{equation}
\Gamma_{t,T}^{U}(\lambda_{t}^{\ast})=\xi_{t,T}^{\mathcal{S}}I^{U}\left(
T,\lambda_{t}^{\ast}\xi_{t,T}^{\mathcal{S}}\right)  \text{ and }\Gamma
_{t,s}^{u}(\lambda_{t}^{\ast})=\xi_{t,s}^{\mathcal{S}}I^{u}\left(
s,\lambda_{t}^{\ast}\xi_{t,s}^{\mathcal{S}}\right)  . \label{Gamma_u1}%
\end{equation}
By (\ref{XtGt}), we can determine the multiplier $\lambda_{t}^{\ast}$ with
information up to time $t$. Besides, we introduce the following quantities for
expressing the optimal policy in the next section:%
\begin{equation}
\Upsilon_{t,T}^{U}(\lambda_{t}^{\ast})=\lambda_{t}^{\ast}\left(  \xi
_{t,T}^{\mathcal{S}}\right)  ^{2}\frac{\partial I^{U}}{\partial y}\left(
T,\lambda_{t}^{\ast}\xi_{t,T}^{\mathcal{S}}\right)  \text{ and }\Upsilon
_{t,s}^{u}(\lambda_{t}^{\ast})=\lambda_{t}^{\ast}\left(  \xi_{t,s}%
^{\mathcal{S}}\right)  ^{2}\frac{\partial I^{u}}{\partial y}\left(
s,\lambda_{t}^{\ast}\xi_{t,s}^{\mathcal{S}}\right)  . \label{Gamma_u2}%
\end{equation}

Consequently, we can express the optimal policy $(\pi_{t},\pi_{t}^{F})$ for
the completed market via the martingale representation theorem (see, e.g.,
Section 3.4 in \cite{Karaztas_Shreve}). With the Clark-Ocone formula (see
\cite{OconeKaratzas1991}), we can further represent this optimal policy in the
form of conditional expectations of suitable random variables. Under a general
and flexible complete-market diffusion model, \cite{Detemple_2003_JF} propose
an explicit conditional expectation form of the optimal policy, and develop a
Monte Carlo simulation method for its implementation; see also, e.g.,
\cite{DetempleRindisbacher2010} along this line of contributions and
\cite{Detemple2014} for a comprehensive survey of the related developments. We
aim at explicitly decomposing the optimal policy for the incomplete market
case and reveal the fundamental difference that arises from market
incompleteness and wealth-dependent utilities, and more importantly, the
particular case where these two situations co-exist.

By the least favorable completion principle proposed in
\cite{karatzas_martingale_1991}, the optimal policy $\pi_{t}$ for the real
assets in the completed market coincides with its counterpart in the original
incomplete market, as long as we properly choose the investor-specific price
of risk $\theta_{v}^{u}$ such that the optimal weights for the fictitious
assets are always identically zero, i.e.,%
\end{subequations}
\begin{equation}
\pi_{v}^{F}\equiv0_{d-m},\text{ for any }0\leq v\leq T. \label{Pi_F_0}%
\end{equation}
Given an arbitrary choice of the volatility function $\sigma^{F}(v,y),$ the
least favorable constraint (\ref{Pi_F_0}) and the orthogonal condition
(\ref{ortho_cond}) determine the desired $\theta_{v}^{u}$ for $0\leq v\leq T$.
Then, the corresponding optimal policy $\pi_{t}$ of the real assets for the
completed market is also optimal for the original incomplete market. In
particular, the desired $\theta_{v}^{u}$ satisfying (\ref{Pi_F_0}) and the
resulting optimal policy $\pi_{t}$ are independent of the specific choice of
$\sigma^{F}(v,y).$

\section{A novel decomposition of optimal dynamic portfolio
choice\label{section:structural_decomposition}}

In this section, we deliver our new decomposition results based on the
previous setup. We substantially develop economic insights regarding the
structure of the optimal policy under the incomplete market setting ($d>m$)
for models with flexible dynamics (\ref{SDE_price}) -- (\ref{SDE_Y}) and
flexible utility functions.

\subsection{Decomposing optimal policy\label{Section:five_comp_decomposition}}

We start by applying the following lemma stating the functional representation
of the unknown investor-specific price of risk $\theta_{v}^{u}$ that satisfies
the least favorable completion constraint (\ref{Pi_F_0}) in general incomplete
market models with flexible utilities.

\begin{lemma}
\label{lemma:thetau_represent} \bigskip The investor-specific price of risk
$\theta_{v}^{u}$ can be expressed as
\begin{equation}
\theta_{v}^{u}=\theta^{u}\left(  v,Y_{v},\lambda_{v}^{\ast};T\right)  ,
\label{thetau_representation_v}%
\end{equation}
for some function $\theta^{u}\left(  v,y,\lambda;T\right)  $ endogenously
determined by the investor's utility function and investment horizon; it
depends on the time $v$, market state $Y_{v}$, as well as investor's wealth
level $X_{v}$ via the multiplier $\lambda_{v}^{\ast}$.
\end{lemma}

\begin{proof}
See Section S.1 in the online supplementary material, where we verify such a result by combining the fictitious completion
approach in \cite{karatzas_martingale_1991} and the minimax local martingale
approach in \cite{HePearson91}.
\end{proof}

Representation (\ref{thetau_representation_v}) reveals the structure of the
investor-specific price of risk $\theta^{u}_{v} = \theta^{u}\left(
v,Y_{v},\lambda_{v}^{\ast};T\right)  $, which is strikingly different from
that of the market price of risk $\theta^{h}\left(  v,Y_{v}\right)  $ defined
in (\ref{def_thetah}) for real assets. It leads to fundamental differences
between incomplete and complete market cases. First, $\theta^{u}\left(
v,Y_{v},\lambda_{v}^{\ast};T\right)  $ depends on the constraint multiplier
$\lambda_{v}^{\ast}$, which solves $X_{v}=E_{v}[\mathcal{G}_{v,T}(\lambda
_{v}^{\ast})]$ by (\ref{XtGt}). Thus, $\theta^{u}\left(  v,Y_{v},\lambda
_{v}^{\ast};T\right)  $ implicitly depends on the current wealth level $X_{v}%
$, and as a consequence, also depends on the entire path of previous market
dynamics up to time $v$. Second, $\theta^{u}\left(  v,Y_{v},\lambda_{v}^{\ast
};T\right)  $ also depends on the investment horizon $T$, which is
economically meaningful and technically important at the level of
implementation. However, neither of these two types of dependence exists in
the market price of risk $\theta^{h}\left(  v,Y_{v}\right)  $\emph{.
}Accordingly, we get the clear insight that, when completing the market
following the least favorable principle (\ref{Pi_F_0}), the introduced
fictitious assets ought to depend on both the current wealth level and the
investment horizon of the specific investor.

The above analysis based on (\ref{thetau_representation_v}) plainly reveals
the economic nature of the investor-specific price of risk $\theta_{v}^{u}$,
which is also known as the shadow price for market incompleteness. However,
the importance of structure (\ref{thetau_representation_v}) goes far beyond
this point. As we will show in Theorem \ref{thm_representation_new} below, the
structure of the investor-specific price of risk is crucial for correctly
establishing the novel decomposition of the optimal policy. It reveals the
fundamental impacts from market incompleteness and wealth-dependent utilities.
In particular, it shows that market incompleteness introduces new channels for
the investors' wealth to impact the optimal policy. Moreover, the structure
(\ref{thetau_representation_v}) facilitates the implementation of the optimal
policy by closed-form solutions or potential numerical methods.

\begin{theorem}
\label{thm_representation_new}Under the incomplete market model
(\ref{SDE_price}) -- (\ref{SDE_Y}), the optimal policy $\pi_{t}$ for the real
assets with prices $S_{t}$ admits the following decomposition:
\begin{equation}
\pi_{t}=\pi^{mv}(t,X_{t},Y_{t})+\pi^{r}(t,X_{t},Y_{t})+\pi^{\theta}%
(t,X_{t},Y_{t}). \label{thm_decomp1}%
\end{equation}
The terms $\pi^{mv}(t,X_{t},Y_{t})$, $\pi^{r}(t,X_{t},Y_{t})$, and
$\pi^{\theta}(t,X_{t},Y_{t})$ denote the mean-variance component, the interest
rate hedge component, and the price of risk hedge component. The price of risk
hedge component $\pi^{\theta}(t,X_{t},Y_{t})$ is further decomposed as
\begin{equation}
\pi^{\theta}(t,X_{t},Y_{t})=\pi^{h}(t,X_{t},Y_{t})+\pi^{u}(t,X_{t},Y_{t}),
\label{thm_decomp2}%
\end{equation}
where the first two components $\pi^{h}(t,X_{t},Y_{t})\ $and $\pi^{u}%
(t,X_{t},Y_{t})$ hedge the uncertainties in market and investor-specific price
risk. The components can be expressed as conditional expectations on random
variables with explicit dynamics\footnote{In line with the time--$t$
formulation of the optimization problem (\ref{Max_expectation}), we express
these components by the time--$t$ state variable $Y_{t}$ and the current
wealth $X_{t}$, rather than the initial wealth $X_{0}$ as in most of the
existing literature, e.g., \cite{Detemple_2003_JF}. We can see that, by
solving constraint (\ref{XtGt}), $\lambda_{t}^{\ast}$ is a function of $t$,
$X_{t}$, and $Y_{t}.$}:
\begin{subequations}
\begin{align}
\pi^{mv}(t,X_{t},Y_{t})  &  =-(\sigma(t,Y_{t})^{+})^{\top}\theta^{h}%
(t,Y_{t})E_{t}[\mathcal{Q}_{t,T}\left(  \lambda_{t}^{\ast}\right)
]/X_{t},\text{ }\label{pimv_thm1}\\
\pi^{r}(t,X_{t},Y_{t})  &  =-(\sigma(t,Y_{t})^{+})^{\top}E_{t}[\mathcal{H}%
_{t,T}^{r}\left(  \lambda_{t}^{\ast}\right)  ]/X_{t},\text{ } \label{pir_thm1}%
\\
\pi^{\theta}(t,X_{t},Y_{t})  &  =-(\sigma(t,Y_{t})^{+})^{\top}E_{t}%
[\mathcal{H}_{t,T}^{\theta}\left(  \lambda_{t}^{\ast}\right)  ]/X_{t},
\label{pitheta_thm1}%
\end{align}
and
\end{subequations}
\begin{subequations}
\begin{align}
\pi^{h}(t,X_{t},Y_{t})  &  =-(\sigma(t,Y_{t})^{+})^{\top}E_{t}[\mathcal{H}%
_{t,T}^{h}]/X_{t},\label{pit_thm1}\\
\pi^{u}(t,X_{t},Y_{t})  &  =-(\sigma(t,Y_{t})^{+})^{\top}E_{t}[\mathcal{H}%
_{t,T}^{u}\left(  \lambda_{t}^{\ast}\right)  ]/X_{t}. \label{pit_thm2}%
\end{align}
Hereof, $\lambda_{t}^{\ast}$ is the multiplier uniquely determined by
(\ref{XtGt}), i.e., $X_{t}=E_{t}[\mathcal{G}_{t,T}(\lambda_{t}^{\ast})]$. It
depends on $X_{t}$ and satisfies the relation $\lambda_{t}^{\ast}=\lambda
_{0}^{\ast}\xi_{t}^{\mathcal{S}}.$ The expressions for $\mathcal{Q}%
_{t,T}\left(  \lambda_{t}^{\ast}\right)  $, $\mathcal{H}_{t,T}^{r}\left(
\lambda_{t}^{\ast}\right)  $, $\mathcal{H}_{t,T}^{\theta}\left(  \lambda
_{t}^{\ast}\right)  $, $\mathcal{H}_{t,T}^{h}$, and $\mathcal{H}_{t,T}%
^{u}\left(  \lambda_{t}^{\ast}\right)  $ are explicitly given in Proposition
\ref{corollary_expression} below. The optimal intermediate consumption $c_{t}$
and terminal wealth $X_{T}$ are given by $c_{t}=I^{u}\left(  t,\lambda
_{t}^{\ast}\right)  $ and $X_{T}=I^{U}\left(  T,\lambda_{T}^{\ast}\right)  .$
\end{subequations}
\end{theorem}

\begin{proof}
See Section \ref{Proof_piu_decomp} in the online supplementary material.
\end{proof}

In Proposition \ref{corollary_expression} below, we provide the explicit
expressions based on standard SDEs for the variables $\mathcal{Q}_{t,T}\left(
\lambda_{t}^{\ast}\right)  $, $\mathcal{H}_{t,T}^{r}\left(  \lambda_{t}^{\ast
}\right)  $, $\mathcal{H}_{t,T}^{\theta}\left(  \lambda_{t}^{\ast}\right)  $,
$\mathcal{H}_{t,T}^{h},$ and $\mathcal{H}_{t,T}^{u}\left(  \lambda_{t}^{\ast
}\right)  $ involved in Theorem \ref{thm_representation_new}. For this
purpose, we apply the representation of the individual-specific price of risk
$\theta_{v}^{u}$ in (\ref{thetau_representation_v}) to introduce the following
$\lambda_{t}^{\ast}-$parameterized version of $\theta_{s}^{u}:$%
\begin{equation}
\theta_{s}^{u}(\lambda_{t}^{\ast})=\theta^{u}\left(  s,Y_{s},\lambda_{s}%
^{\ast};T\right)  =\theta^{u}\left(  s,Y_{s},\lambda_{t}^{\ast}\xi
_{t,s}^{\mathcal{S}}(\lambda_{t}^{\ast});T\right)  ,
\label{thetau_representation_s}%
\end{equation}
where the second equality follows from the definition $\xi_{t,s}^{\mathcal{S}%
}=\xi_{s}^{\mathcal{S}}/\xi_{t}^{\mathcal{S}}$ and the relation $\lambda
_{t}^{\ast}=\lambda_{0}^{\ast}\xi_{t}^{\mathcal{S}}$, i.e., $\lambda_{s}%
^{\ast}=\lambda_{0}^{\ast}\xi_{s}^{\mathcal{S}}=\lambda_{t}^{\ast}\xi
_{t,s}^{\mathcal{S}};$ $\xi_{t,s}^{\mathcal{S}}(\lambda_{t}^{\ast})$ in
(\ref{thetau_representation_s}) denotes the $\lambda_{t}^{\ast}-$parameterized
version of the state price density, which evolves according to%
\begin{equation}
d\xi_{t,s}^{\mathcal{S}}(\lambda_{t}^{\ast})=-\xi_{t,s}^{\mathcal{S}}%
(\lambda_{t}^{\ast})[r(s,Y_{s})ds+\theta_{s}^{\mathcal{S}}(\lambda_{t}^{\ast
})^{\top}dW_{s}], \label{thm1_SDE_xi_incomp_explicit}%
\end{equation}
with%
\begin{equation}
\theta_{s}^{\mathcal{S}}(\lambda_{t}^{\ast})=\theta_{s}^{h}(t,Y_{t}%
)+\theta^{u}\left(  s,Y_{s},\lambda_{t}^{\ast}\xi_{t,s}^{\mathcal{S}}%
(\lambda_{t}^{\ast});T\right)  . \label{theta_S_lambda}%
\end{equation}

We introduce these $\lambda_{t}^{\ast}-$dependent versions
(\ref{thetau_representation_s})--(\ref{theta_S_lambda}) to highlight the
impact from the investor's wealth, as $\lambda_{t}^{\ast}$ depends on the
current wealth level $X_{t}$ via (\ref{XtGt}), i.e., $X_{t}=E_{t}%
[\mathcal{G}_{t,T}(\lambda_{t}^{\ast})]$. In addition, we see that
$\lambda_{t}^{\ast}$ can be fully determined by the information at time $t$.
Thus, these $\lambda_{t}^{\ast}$--dependent versions clearly reveal the
temporal structure of the optimal policy by isolating the information
available at time $t$.

\begin{proposition}
\label{corollary_expression}The quantities $\mathcal{Q}_{t,T}\left(
\lambda_{t}^{\ast}\right)  $, $\mathcal{H}_{t,T}^{r}(\lambda_{t}^{\ast})$, and
$\mathcal{H}_{t,T}^{\theta}(\lambda_{t}^{\ast})$ in (\ref{pimv_thm1}) --
(\ref{pitheta_thm1}) are given by%
\begin{subequations}
\begin{align}
&  \mathcal{Q}_{t,T}(\lambda_{t}^{\ast})=\Upsilon_{t,T}^{U}(\lambda_{t}^{\ast
})+\int_{t}^{T}\Upsilon_{t,s}^{u}(\lambda_{t}^{\ast})ds,\label{coro_Q}\\
&  \mathcal{H}_{t,T}^{r}(\lambda_{t}^{\ast})=(\Gamma_{t,T}^{U}(\lambda
_{t}^{\ast})+\Upsilon_{t,T}^{U}(\lambda_{t}^{\ast}))H_{t,T}^{r}+\int_{t}%
^{T}(\Gamma_{t,s}^{u}(\lambda_{t}^{\ast})+\Upsilon_{t,s}^{u}(\lambda_{t}%
^{\ast}))H_{t,s}^{r}ds,\label{coro_Hr}\\
&  \mathcal{H}_{t,T}^{\theta}(\lambda_{t}^{\ast})=(\Gamma_{t,T}^{U}%
(\lambda_{t}^{\ast})+\Upsilon_{t,T}^{U}(\lambda_{t}^{\ast}))H_{t,T}^{\theta
}(\lambda_{t}^{\ast})+\int_{t}^{T}(\Gamma_{t,s}^{u}(\lambda_{t}^{\ast
})+\Upsilon_{t,s}^{u}(\lambda_{t}^{\ast}))H_{t,s}^{\theta}(\lambda_{t}^{\ast
})ds, \label{coro_Ht}%
\end{align}
where $\Gamma_{t,T}^{U}(\lambda_{t}^{\ast})$, $\Gamma_{t,s}^{u}(\lambda
_{t}^{\ast}),$ $\Upsilon_{t,T}^{U}(\lambda_{t}^{\ast})$, and $\Upsilon
_{t,s}^{u}(\lambda_{t}^{\ast})$ are defined in (\ref{Gamma_u1}) and
(\ref{Gamma_u2}) except for replacing the relative state price density
$\xi_{t,s}^{\mathcal{S}}$ by the $\lambda_{t}^{\ast}-$dependent version
$\xi_{t,s}^{\mathcal{S}}(\lambda_{t}^{\ast})$. $H_{t,s}^{r}$ in (\ref{coro_Hr}%
) is a $d$--dimensional vector-valued processes evolving according to SDEs:%
\end{subequations}
\begin{equation}
dH_{t,s}^{r}=\left(  \mathcal{D}_{t}Y_{s}\right)  \nabla r(s,Y_{s})ds,
\label{thm1_SDE_Hr}%
\end{equation}
for $t\leq s\leq T,$ with initial value $H_{t,t}^{r}=0_{d}.$ Here and
throughout this paper, $\nabla$ denotes the gradient of functions with respect
to the arguments in the place of $Y_{s}$\footnote{For an $m-$dimensional
vector-valued function $f(t,y)=(f_{1}(t,y),f_{2}(t,y),\cdots,f_{m}(t,y))$, its
gradient is an $n\times m$ matrix with each element given by $\left[  \nabla
f(t,y)\right]  _{ij}=\partial f_{j}/\partial y_{i}(t,y),$ for $i=1,2,\ldots,n$
and $j=1,2,\ldots,m.$}, and $\mathcal{D}_{t}$ denotes the Malliavin derivative
at time $t$. Specifically, $\mathcal{D}_{t}Y_{s}$ is a $d\times n$ matrix with
$\mathcal{D}_{t}Y_{s}=((\mathcal{D}_{1t}Y_{s})^{\top},(\mathcal{D}_{2t}%
Y_{s})^{\top},\cdots,(\mathcal{D}_{dt}Y_{s})^{\top})^{\top},$ where each
$\mathcal{D}_{it}Y_{s}$ is an $n$--dimensional column vector satisfying the
SDE:
\begin{equation}
d\mathcal{D}_{it}Y_{s}=\left(  \nabla\alpha(s,Y_{s})\right)  ^{\top
}\mathcal{D}_{it}Y_{s}ds+\sum_{j=1}^{d}\left(  \nabla\beta_{j}(s,Y_{s}%
)\right)  ^{\top}\mathcal{D}_{it}Y_{s}dW_{js},\text{ }\lim_{s\rightarrow
t}\mathcal{D}_{it}Y_{s}=\beta_{i}(t,Y_{t}), \label{thm1_SDE_DYM}%
\end{equation}
Besides, $H_{t,s}^{\theta}(\lambda_{t}^{\ast})$ in (\ref{coro_Ht}) follows the
SDE\footnote{Here, $\mathcal{D}_{t}\theta_{s}^{u}(\lambda_{t}^{\ast})$ is a
$d\times d$ matrix with $\mathcal{D}_{t}\theta_{s}^{u}(\lambda_{t}^{\ast
})=((\mathcal{D}_{1t}\theta_{s}^{u}(\lambda_{t}^{\ast}))^{\top},(\mathcal{D}%
_{2t}\theta_{s}^{u}(\lambda_{t}^{\ast}))^{\top},\cdots,(\mathcal{D}_{dt}%
\theta_{s}^{u}(\lambda_{t}^{\ast}))^{\top})^{\top},$ where each $\mathcal{D}%
_{it}\theta_{s}^{u}(\lambda_{t}^{\ast})$ is a $d$--dimensional column
vector.}
\begin{equation}
dH_{t,s}^{\theta}(\lambda_{t}^{\ast})=\big[ \left(  \mathcal{D}_{t}%
Y_{s}\right)  \nabla\theta^{h}(s,Y_{s})+\mathcal{D}_{t}\theta_{s}^{u}%
(\lambda_{t}^{\ast})\big] (\theta_{s}^{\mathcal{S}}(\lambda_{t}^{\ast
})ds+dW_{s}), \label{thm1_SDE_Htheta}%
\end{equation}
with initial value $H_{t,t}^{\theta}(\lambda_{t}^{\ast})=0_{d}$. Finally, the
terms $\mathcal{H}_{t,T}^{h}$ and $\mathcal{H}_{t,T}^{u}(\lambda_{t}^{\ast})$
in (\ref{pit_thm1})--(\ref{pit_thm2}) are defined in the same way as that for
$\mathcal{H}_{t,T}^{\theta}(\lambda_{t}^{\ast})$ in (\ref{coro_Ht}) except for
replacing $H_{t,s}^{\theta}(\lambda_{t}^{\ast})$ by $H_{t,s}^{h}$ and
$H_{t,s}^{u}(\lambda_{t}^{\ast})$ for $t\leq s\leq T$, which are both
$d$--dimensional vector-valued processes evolving according to SDEs:
\begin{subequations}
\begin{align}
dH_{t,s}^{h}  &  =\left(  \mathcal{D}_{t}Y_{s}\right)  \nabla\theta
^{h}(s,Y_{s})(\theta^{h}(s,Y_{s})ds+dW_{s}),\label{dH_thetah}\\
dH_{t,s}^{u}(\lambda_{t}^{\ast})  &  =\mathcal{D}_{t}\theta_{s}^{u}%
(\lambda_{t}^{\ast})(\theta_{s}^{u}(\lambda_{t}^{\ast})ds+dW_{s}),
\label{dH_thetau1}%
\end{align}
with initial values $H_{t,t}^{h}=H_{t,t}^{u}(\lambda_{t}^{\ast})=0_{d}$.
\end{subequations}
\end{proposition}

\begin{proof}
See Section S.1 in the online supplementary material.
\end{proof}

The explicit structure of the optimal policy, with four components given in
(\ref{thm_decomp1}) and (\ref{thm_decomp2}) extends the usual complete-market
decomposition with only three components; see, e.g., \cite{merton71} and
\cite{Detemple_2003_JF}. We now analyze the economic implication of each
component. The component $\pi^{mv}(t,X_{t},Y_{t})$ in (\ref{pimv_thm1}) is the
mean-variance component, as reflected by the market price of risk $\theta
^{h}(t,Y_{t})$ defined in (\ref{def_thetah}) -- a mean-variance trade-off for
the risky assets. The second component $\pi^{r}(t,X_{t},Y_{t})$ in
(\ref{pir_thm1}) is for hedging the uncertainty in the interest rate, as seen
from the gradient of interest rate $\nabla r$ in (\ref{thm1_SDE_Hr}). Finally,
the component $\pi^{\theta}(t,X_{t},Y_{t})$ given by (\ref{pitheta_thm1})
hedges the uncertainty in the price of risk. By (\ref{thm_decomp2}), it can be
further decomposed into two components $\pi^{h}(t,X_{t},Y_{t})$ and $\pi
^{u}(t,X_{t},Y_{t})$. The two components hedge the uncertainty in market price
of risk $\theta^{h}$ and investor-specific price of risk $\theta^{u}$, as seen
by the gradient $\nabla\theta^{h}$ in (\ref{dH_thetah}) and the Malliavin term
$\mathcal{D}_{t}\theta_{s}^{u}(\lambda_{t}^{\ast})$ in (\ref{dH_thetau1}). As
we will discuss in Section \ref{section: complete_representation}, the last
component $\pi^{u}(t,X_{t},Y_{t})$ is absent under complete market models.

In particular, as a natural analogue to a classical derivative, we can
intuitively understand the Malliavin derivative $\mathcal{D}_{t}$ as the
sensitivity to the underlying Brownian motion $W_{t};$ see Appendix D of
\cite{Detemple_2003_JF} for an accessible survey of Malliavin calculus in
finance.\footnote{We can view Malliavin calculus as the stochastic calculus of
variation in the space of sample paths. Malliavin calculus has proven its
important role in financial economics through its merit in solving portfolio
choice problems, see, e.g., \cite{OconeKaratzas1991}, \cite{Detemple_2003_JF},
\cite{Detemple_Rindisbacher_2005_MF}, and \cite{DetempleRindisbacher2010}.
See, for example, \cite{Nualart_MC_book} for a book-length discussion of the
theory of Malliavin calculus.} Indeed, throughout this paper, the Malliavin
derivative on the state variable $\mathcal{D}_{t}Y_{s}$ can be viewed as a
standard multidimensional diffusion process with dynamics given explicitly in
(\ref{thm1_SDE_DYM}). On the other hand, the Malliavin derivative on the
investor-specific price of risk $\mathcal{D}_{t}\theta_{s}^{u}(\lambda
_{t}^{\ast})$ clearly depends on the structure of $\theta_{s}^{u}(\lambda
_{t}^{\ast})$ given in (\ref{thetau_representation_s}). Specifically, the
underlying Brownian motion $W_{t}$ impacts the value of $\theta_{s}%
^{u}(\lambda_{t}^{\ast})=\theta^{u}\left(  s,Y_{s},\lambda_{s}^{\ast
};T\right)  $ via both the state variable $Y_{s}$ and the multiplier
$\lambda_{s}^{\ast}$. This structure introduces additional variation in the
investor-specific price of risk $\theta^{u}(s,Y_{s},\lambda_{s}^{\ast};T)$,
which is absent in the market price of risk $\theta^{h}(s,Y_{s})$.

Our decompositions in Theorem \ref{thm_representation_new} and Proposition
\ref{corollary_expression} provide an explicit representation of the optimal
policy in general incomplete market models. In addition, they clearly reveal
how investor's wealth $X_{t}$ impacts the optimal policy, i.e., the channels
for the wealth-dependent effect. It can be seen by checking how the time-$t$
multiplier $\lambda_{t}^{\ast}$ gets involved in the optimal policy, as it is
determined by $X_{t}$ according to (\ref{XtGt}). First, the time-$t$
multiplier $\lambda_{t}^{\ast}$ directly appears in the optimal policies
(\ref{pimv_thm1}) -- (\ref{pitheta_thm1}) through the functions $\Gamma
^{\iota}_{t,T}(\lambda_{t}^{\ast})$ and $\Upsilon^{\iota}_{t,T}(\lambda
_{t}^{\ast})$ for $\iota\in\{u,U\}$, which are given by (\ref{Gamma_u1}) and
(\ref{Gamma_u2}). Second, as shown by (\ref{thm1_SDE_xi_incomp_explicit}) and
(\ref{thm1_SDE_Htheta}), the multiplier $\lambda_{t}^{\ast}$ affects the
dynamics of the building blocks $\xi^{\mathcal{S}}_{t,s}(\lambda_{t}^{\ast})$
and $H^{\theta}_{t,s}(\lambda_{t}^{\ast})$. It introduces an implicit impact
of the wealth level on the optimal policy. By (\ref{theta_S_lambda}), such
implicit impact essentially stems from the special structure of the
investor-specific price of risk $\theta^{u}_{s}(\lambda_{t}^{\ast})$, which
only appears in incomplete market models. Thus, the market incompleteness
leads to additional channels for the wealth-dependent effect, i.e., via the
structure of the investor-specific price of risk.

Now, we still face one remaining difficulty in solving for the optimal policy
-- the investor-specific price of risk $\theta^{u}\left(  v,Y_{v},\lambda
_{v}^{\ast};T\right)  $ in (\ref{thetau_representation_v}) is still
undetermined. For this task, we establish a novel equation for characterizing
$\theta^{u}\left(  v,Y_{v},\lambda_{v}^{\ast};T\right)  $ in Theorem
\ref{thm_integral_equation} below. To obtain $\theta^{u}\left(  v,Y_{v}%
,\lambda_{v}^{\ast};T\right)  $, we first express the optimal policy for the
fictitious assets $\pi_{v}^{F}$ using representations similar to those in
Theorem \ref{thm_representation_new}. Then, we combine the least favorable
completion principle (\ref{Pi_F_0}) and the orthogonal condition
(\ref{ortho_cond}). The derivation of the proper $\theta^{u}\left(
v,Y_{v},\lambda_{v}^{\ast};T\right)  $ again shows that it is essential to
understand the correct structure of investor-specific price of risk in
(\ref{thetau_representation_v}) and establish the explicit decomposition of
optimal policy as in Theorem \ref{thm_representation_new}.

\begin{theorem}
\label{thm_integral_equation}The proper investor-specific price of risk\emph{
}$\theta_{v}^{u}=\theta^{u}\left(  v,Y_{v},\lambda_{v}^{\ast};T\right)  $ for
the least favorable completion satisfies the following $d-$dimensional
equation:
\begin{equation}
\theta^{u}\left(  v,Y_{v},\lambda_{v}^{\ast};T\right)  =\frac{\sigma
(v,Y_{v})^{+}\sigma(v,Y_{v})-I_{d}}{E[\mathcal{Q}_{v,T}(\lambda_{v}^{\ast
})|Y_{v},\lambda_{v}^{\ast}]}\times(E[\mathcal{H}_{v,T}^{r}(\lambda_{v}^{\ast
})|Y_{v},\lambda_{v}^{\ast}]+E[\mathcal{H}_{v,T}^{\theta}(\lambda_{v}^{\ast
})|Y_{v},\lambda_{v}^{\ast}]), \label{integral_theta}%
\end{equation}
for $0\leq v\leq T$. Here $I_{d}$ denotes the $d$--dimensional identity
matrix; $\sigma(v,Y_{v})^{+}$ is given by $\sigma(v,Y_{v})^{+}=\sigma
(v,Y_{v})^{\top}(\sigma(v,Y_{v})\sigma(v,Y_{v})^{\top})^{-1}$.
\end{theorem}

\begin{proof}
See Section S.2 in the online supplementary material.
\end{proof}

Theorem \ref{thm_integral_equation} shows that the investor-specific price of
risk $\theta^{u}\left(  v,Y_{v},\lambda_{v}^{\ast};T\right)  $ is determined
by a complex multidimensional equation system, which consists of equation
(\ref{integral_theta}) and the SDEs of $Y_{s},$ $\xi_{v,s}^{\mathcal{S}%
}(\lambda_{v}^{\ast}),$ $H_{v,s}^{r},$ $H_{v,s}^{\theta}(\lambda_{v}^{\ast}),$
and $\mathcal{D}_{iv}Y_{s}$ given in (\ref{SDE_Y}),
(\ref{thm1_SDE_xi_incomp_explicit}), (\ref{thm1_SDE_Hr}),
(\ref{thm1_SDE_Htheta}), and (\ref{thm1_SDE_DYM}). In the literature, other
types of differential equations are employed for characterizing optimal
portfolios in incomplete markets. For example, \cite{HePearson91} indirectly
relate the optimal policy and its equivalent local martingale measure to the
solution of a quasi-linear PDE. \cite{Detemple_Rindisbacher_2005_MF} derive
and solve a forward-backward SDE governing the shadow price under a model of
partially hedgeable Gaussian interest rate and CRRA utility (see also
\cite{DetempleRindisbacher2010} for generalizing such results). The
contribution of our characterization (\ref{integral_theta}) lies in that it
reveals the explicit structure of the investor-specific price of risk using
the representation (\ref{thetau_representation_v}). In particular, it clearly
shows the wealth level $X_{v}$ impacts the value of $\theta^{u}\left(
v,Y_{v},\lambda_{v}^{\ast};T\right)  $ via the multiplier $\lambda_{v}^{\ast}%
$, as shown by the SDEs of $\xi_{v,s}^{\mathcal{S}}(\lambda_{v}^{\ast})$ and
$H_{v,s}^{\theta}(\lambda_{v}^{\ast})$ in (\ref{thm1_SDE_xi_incomp_explicit})
and (\ref{thm1_SDE_Htheta}).

In addition, equation (\ref{integral_theta}) implies its terminal condition:
\begin{equation}
\theta^{u}\left(  T,Y_{T},\lambda_{T}^{\ast};T\right)  \equiv0_{d},
\label{integral_terminal}%
\end{equation}
which corresponds to the investor-specific price of risk with the investment
horizon shrinking to zero. It suggests that, under a vanishing investment
horizon, the returns of the fictitious assets $\mu_{t}^{F}$ employed in the
least favorable completion converge to the risk-free return $r(t,Y_{t})$. It
follows from $\theta^{u}\left(  v,Y_{v},\lambda_{v}^{\ast};T\right)
=\sigma^{F}(v,Y_{v})^{+}(\mu_{v}^{F}-r(v,Y_{v})1_{d-m})\ $ by
(\ref{def_thetau}) and (\ref{thetau_representation_v}). The terminal condition
(\ref{integral_terminal}) plays an important role in potential numerical
methods for solving the investor-specific price of risk $\theta^{u}\left(
v,Y_{v},\lambda_{v}^{\ast};T\right)  $. For example, in a simulation approach,
it can serve as the boundary condition for backward induction. It verifies
that the investor-specific price of risk indeed endogenously depends on the
investment horizon of each investor.

\subsection{Revealing wealth-related component under differentiability}

\label{section:differentiability}

In this section, we reveal a novel structural impact of the wealth-dependent
property on the optimal policy, under the additional assumption that the
investor-specific price of risk function $\theta^{u}(v,y,\lambda;T)$ is
differentiable in its arguments. As we will show in Section
\ref{section:hn_closedform}, this natural assumption holds for the classic
Heston-SV model.

\begin{proposition}
\label{prop:piu_decomp}\bigskip Under the assumption that the
investor-specific price of risk function $\theta^{u}(v,y,\lambda;T)$ is
differentiable in its arguments, we can decompose the component $\pi
^{u}(t,X_{t},Y_{t})$ as
\begin{equation}
\pi^{u}(t,X_{t},Y_{t})=\pi^{u,Y}(t,X_{t},Y_{t})+\pi^{u,\lambda}(t,X_{t}%
,Y_{t}), \label{piu_decomp}%
\end{equation}
where%
\begin{subequations}
\begin{align}
\pi^{u,Y}(t,X_{t},Y_{t})  &  =-(\sigma(t,Y_{t})^{+})^{\top}E_{t}%
[\mathcal{H}_{t,T}^{u,Y}(\lambda_{t}^{\ast})]/X_{t},\label{pi_uY}\\
\pi^{u,\lambda}(t,X_{t},Y_{t})  &  =-(\sigma(t,Y_{t})^{+})^{\top}%
E_{t}[\mathcal{H}_{t,T}^{u,\lambda}\left(  \lambda_{t}^{\ast}\right)  ]/X_{t}.
\label{pi_uLm}%
\end{align}
The terms $\mathcal{H}_{t,T}^{u,Y}(\lambda_{t}^{\ast})$ and $\mathcal{H}%
_{t,T}^{u,\lambda}(\lambda_{t}^{\ast})$ in (\ref{pi_uY})--(\ref{pi_uLm}) are
defined in the same way as that for $\mathcal{H}_{t,T}^{\theta}(\lambda
_{t}^{\ast})$ in (\ref{coro_Ht}) except for replacing $H_{t,s}^{\theta
}(\lambda_{t}^{\ast})$ by $H_{t,s}^{u,Y}(\lambda_{t}^{\ast})$ and
$H_{t,s}^{u,\lambda}(\lambda_{t}^{\ast})$ for $t\leq s\leq T$, which follow
SDEs:%
\end{subequations}
\begin{subequations}
\begin{align}
dH_{t,s}^{u,Y}(\lambda_{t}^{\ast})  &  =\left(  \mathcal{D}_{t}Y_{s}\right)
\nabla\theta^{u}(s,Y_{s},\lambda_{t}^{\ast}\xi_{t,s}^{\mathcal{S}}(\lambda
_{t}^{\ast});T)(\theta_{s}^{u}(\lambda_{t}^{\ast})ds+dW_{s}%
),\label{dH_thetauY}\\
dH_{t,s}^{u,\lambda}(\lambda_{t}^{\ast})  &  =-\lambda_{t}^{\ast}\xi
_{t,s}^{\mathcal{S}}(\lambda_{t}^{\ast})\big(\theta_{t}^{\mathcal{S}}%
(\lambda_{t}^{\ast})+H_{t,s}^{r}+H_{t,s}^{\theta}(\lambda_{t}^{\ast
})\big)\nonumber\\
&  \cdot\partial\theta^{u}/\partial\lambda(s,Y_{s},\lambda_{t}^{\ast}\xi
_{t,s}^{\mathcal{S}}(\lambda_{t}^{\ast});T)(\theta_{s}^{u}(\lambda_{t}^{\ast
})ds+dW_{s}). \label{dH_thetauLm}%
\end{align}

\end{subequations}
\end{proposition}

\begin{proof}
See Section \ref{Proof_piu_decomp} in the online supplementary material.
\end{proof}

As one of our contributions, we show in (\ref{piu_decomp}) that we can further
decompose the investor-specific price of risk hedge component $\pi^{u}%
(t,X_{t},Y_{t})$ into two interpretable parts under the differentiability
assumption. By combining (\ref{piu_decomp}) with (\ref{thm_decomp1}) --
(\ref{thm_decomp2}) in Theorem \ref{thm_representation_new}, we obtain a final
decomposition of the optimal policy with five interpretable components.
Specifically, the first term $\pi^{u,Y}(t,X_{t},Y_{t})$ hedges the fluctuation
in the investor-specific price of risk that arises from the state variable
$Y_{t}$, as explicitly reflected by the gradient $\nabla\theta^{u}%
(s,Y_{s},\lambda_{t}^{\ast}\xi_{t,s}^{\mathcal{S}}(\lambda_{t}^{\ast});T)$ in
(\ref{dH_thetau1}). Its structure resembles that of $\pi^{h}(t,X_{t},Y_{t})$,
except for replacing $\theta^{h}$ by $\theta^{u}$ in (\ref{dH_thetah}). The
second term $\pi^{u,\lambda}(t,X_{t},Y_{t})$, however, is introduced via a
fundamentally different channel: it essentially hedges the uncertainty in the
investor-specific price of risk due to the variation in investor's wealth
level. To see this, we note that $\pi^{u,\lambda}(t,X_{t},Y_{t})$ hinges on
the partial derivative $\partial\theta^{u}/\partial\lambda(s,Y_{s},\lambda
_{t}^{\ast}\xi_{t,s}^{\mathcal{S}}(\lambda_{t}^{\ast});T)$ in
(\ref{dH_thetauLm}). It captures the sensitivity\ of $\theta^{u}$ with respect
to the multiplier $\lambda_{s}^{\ast}=\lambda_{t}^{\ast}\xi_{t,s}%
^{\mathcal{S}}(\lambda_{t}^{\ast})$, which is directly related to investor's
wealth level $X_{s}$ via $X_{s}=E_{s}[G_{s,T}(\lambda_{s}^{\ast})]$ by
(\ref{XtGt}).

By the above analysis, we see that investors need to hedge the uncertainty in
investor-specific price of risk from both the state variable and investor's
wealth level. The two channels of uncertainty stems from the structure of
investor-specific price of risk in (\ref{thetau_representation_v}) --- it
depends on both the state variable $Y_{s}$ and investor's wealth level $X_{s}$
via the multiplier $\lambda_{s}^{\ast}$. For the market price of risk
$\theta^{h}(s,Y_{s})$, however, only the uncertainty from the state variable
needs to be hedged via the component $\pi^{h}(t,X_{t},Y_{t})$. In addition, as
we will show in the next two sections, the term $\pi^{u,\lambda}(t,X_{t}%
,Y_{t})$ can only possibly appear when the market is incomplete and, at the
same time, the utility is wealth-dependent. That is, it vanishes under the
CRRA utility, and also under complete market models. Thus, it highlights the
impact on the optimal policy coming from the interaction of market
incompleteness and wealth-dependent utilities, i.e., the proper way for the
investor to complete the market depends on her current wealth level. To our
best knowledge, we are the first to show this additional term in optimal
dynamic portfolio allocation.

\subsection{Decomposing optimal policy under the CRRA utility}

\label{section:crra_policy}

To reveal the structural impact of wealth-dependent utility on the optimal
policy, we now investigate how the optimal policy degenerates under the
wealth-independent CRRA utility, i.e., with $\bar{c}=0$ and $\bar{x}=0$ in
(\ref{hara_utility}). The results are obtained by applying our general
decomposition in Theorems \ref{thm_representation_new} and
\ref{thm_integral_equation}, and then simplifying the results using the
special structure of the CRRA utility. We provide the main results here and
discuss the key insights. We refer to \ref{section:impact_utility} for a
further detailed discussion.

First, we can show that the investor-specific price of risk $\theta^{u}\left(
v,Y_{v},\lambda_{v}^{\ast};T\right)  $ introduced in
(\ref{thetau_representation_v}) is independent of the multiplier $\lambda
_{v}^{\ast}$ and thus can be reduced to $\theta^{u}\left(  v,Y_{v};T\right)  $
under the CRRA utility. Thus, it is no longer affected by investor's wealth
level $X_{v}$ via the equation $X_{v}=E_{v}[\mathcal{G}_{v,T}(\lambda
_{v}^{\ast})]$. It means that under incomplete market models, the proper way
for a CRRA investor to complete the market only depends on the market state
and its horizon. Consequently, the relative state price density $\xi
_{t,s}^{\mathcal{S}}(\lambda_{t}^{\ast})$ and Malliavin term $H_{t,s}^{\theta
}(\lambda_{t}^{\ast})$, as given in general by dynamics
(\ref{thm1_SDE_xi_incomp_explicit}) and (\ref{thm1_SDE_Htheta}), become
independent of the multiplier $\lambda_{t}^{\ast}.$ Thus, they can be written
as $\xi_{t,s}^{\mathcal{S}}$ and $H_{t,s}^{\theta}$ with dynamics%
\begin{subequations}%
\begin{equation}
d\xi_{t,s}^{\mathcal{S}}=-\xi_{t,s}^{\mathcal{S}}[r(s,Y_{s})ds+(\theta
^{h}(s,Y_{s})+\theta^{u}(s,Y_{s};T))^{\top}dW_{s}], \label{SDE_xi_crra}%
\end{equation}
and%
\begin{equation}
dH_{t,s}^{\theta}=\left(  \mathcal{D}_{t}Y_{s}\right)  (\nabla\theta
^{h}(s,Y_{s})+\nabla\theta^{u}(s,Y_{s};T))[\theta^{\mathcal{S}}(s,Y_{s}%
;T)ds+dW_{s}], \label{SDE_Htheta_CRRA}%
\end{equation}%
\end{subequations}%
where $\theta^{\mathcal{S}}(s,Y_{s};T)=\theta^{h}(s,Y_{s})+\theta^{u}%
(s,Y_{s};T)$.

By the above analysis, we see the wealth-independent property of the CRRA
utility is embodied by multiple building blocks of the optimal policy: the
investor-specific price of risk $\theta^{u}\left(  v,Y_{v};T\right)  $, the
relative state price density $\xi_{t,s}^{\mathcal{S}}$, and the Malliavin term
$H_{t,s}^{\theta}$ are all independent of investor's current wealth level, in
contrast to their counterparts under general utilities. With these
preparations, we provide the decomposition of optimal policy under the CRRA
utility in the following proposition.

\begin{proposition}
\label{Corollary_CRRA copy}Under the incomplete market model (\ref{SDE_price})
-- (\ref{SDE_Y}) and the CRRA utility function given in (\ref{hara_utility})
with $\overline{x}=\overline{c}=0$, the investor-specific price of risk\emph{
}$\theta^{u}\left(  v,Y_{v},\lambda_{v}^{\ast};T\right)  $ does not depend on
the multiplier $\lambda_{v}^{\ast}$, and can be simplified as $\theta
^{u}\left(  v,Y_{v};T\right)  $. Consequently, the relative state price
density $\xi_{t,s}^{\mathcal{S}}(\lambda_{t}^{\ast})$ and Malliavin term
$H_{t,s}^{\theta}(\lambda_{t}^{\ast})$ become independent of $\lambda
_{t}^{\ast}$, with dynamics given in (\ref{SDE_xi_crra}) --
(\ref{SDE_Htheta_CRRA}). The optimal policy is wealth-independent, and allows
the decomposition:
\begin{equation}
\pi_{t}=\pi^{mv}(t,Y_{t})+\pi^{r}(t,Y_{t})+\pi^{\theta}(t,Y_{t})
\label{crra_decompose}%
\end{equation}
and $\pi^{\theta}(t,Y_{t})=\pi^{h}(t,Y_{t})+\pi^{u}(t,Y_{t}),$ where the
mean-variance component $\pi^{mv}(t,Y_{t})$ can be explicitly solved as%
\begin{equation}
\pi^{mv}(t,Y_{t})=\frac{1}{\gamma}(\sigma(t,Y_{t})^{+})^{\top}\theta
^{h}(t,Y_{t})=\frac{1}{\gamma}(\sigma(t,Y_{t})\sigma(t,Y_{t})^{\top}%
)^{-1}\left(  \mu(t,Y_{t})-r(t,Y_{t})1_{m}\right)  ; \label{pmv_crra_explicit}%
\end{equation}
the other components\ and the equation characterizing $\theta^{u}\left(
v,Y_{v};T\right)  $ are explicitly given in the appendix. Under the
differentiability assumption for $\theta^{u}\left(  v,y;T\right)  $, we have
$\pi^{u}(t,Y_{t})=\pi^{u,Y}(t,Y_{t})$ in (\ref{piu_decomp}), i.e., the last
component $\pi^{u,\lambda}(t,X_{t},Y_{t})$ vanishes under the CRRA utiltiy.
\end{proposition}

\begin{proof}
See Section \ref{proof_corrolary_crra} in the online supplementary material.
\end{proof}

Comparing the decomposition results in Proposition \ref{Corollary_CRRA copy}
with those in Theorem \ref{thm_representation_new}, we find that the optimal
policy under the CRRA utility differs from that under general utilities in the
following aspects. First, the optimal policy under the CRRA utility is
independent of investor's wealth level $X_{t}$, reflecting the
wealth-independent property of the CRRA utility. It is guaranteed as follows.
First, the time--$t$ multiplier $\lambda_{t}^{\ast}$ does not appear in the
optimal policy (\ref{crra_decompose}) under the CRRA utility. Second, with the
simplified structure of $\theta^{u}(s,Y_{s};T)$, the building blocks
$\xi_{t,s}^{\mathcal{S}}$ and $H_{t,s}^{\theta}$ in (\ref{SDE_xi_crra}) --
(\ref{SDE_Htheta_CRRA}) are also independent of $\lambda_{t}^{\ast}$. It
ensures that the wealth level $X_{t}$ is not implicitly involved in the
optimal policy.

Next, the mean-variance component $\pi^{mv}(t,Y_{t})$ can be explicitly solved
as (\ref{pmv_crra_explicit}) under the CRRA utility. It is given as the
product of the inverse covariance matrix $(\sigma(t,Y_{t})\sigma
(t,Y_{t})^{\top})^{-1}$ and the excess return $\mu(t,Y_{t})-r(t,Y_{t})1_{m}$,
and then further divided by the investor's risk aversion level $\gamma$.
Moreover, the mean-variance component is \textquotedblleft
myopic\textquotedblright\ under the CRRA utility, as it is independent of the
investor's horizon. Finally, the last component $\pi^{u,\lambda}(t,X_{t}%
,Y_{t})$ in (\ref{piu_decomp}) vanishes in the CRRA policy, as $\theta
^{u}(s,Y_{s};T)$ no longer depends on the multiplier $\lambda_{s}^{\ast}$. It
suggests investors do not need to hedge the uncertainty in investor-specific
price of risk due to variation in their wealth level. This again reflects the
impact of the wealth-independent property of the CRRA utility.

\subsection{Discussing the complete market
case\label{section: complete_representation}}

We now briefly describe the decomposition of optimal policy under complete
market models. By comparing it with its counterpart under incomplete market
models, we can reveal the fundamental impact of market incompleteness on
optimal portfolio allocation. The complete market model is set under the
general model (\ref{SDE_price}) and (\ref{SDE_Y}) with the number of risky
assets being equal to the number of driving Brownian motions, i.e., $m=d$. In
addition, we assume that the resulting square matrix $\displaystyle\sigma
(t,y)$ is non-singular. Thus, we can fully hedge the risk by investing in the
risky assets, and do not need the completion procedure with fictitious assets.
In this case, the market price of risk in (\ref{def_thetah}) is defined as
$\theta(t,Y_{t}):=\sigma(t,Y_{t})^{-1}\left(  \mu(t,Y_{t})-r(t,Y_{t}%
)1_{m}\right)  ,$ while the investor-specific price of risk $\theta_{t}^{u}$
vanishes as the investor does not need to complete the market by fictitious
assets. Then, the decomposition of the optimal policy follows as a special
case of the incomplete market results in Theorem \ref{thm_representation_new},
via replacing the components by their counterparts in the complete market.

The market completeness introduces the following two simplifications of
optimal policy decomposition. First, as the undetermined investor-specific
price of risk $\theta_{s}^{u}$ is not involved, the dynamics of the relative
state price density $\xi_{t,s}^{\mathcal{S}}$ and Malliavin term
$H_{t,s}^{\theta}$, as given in (\ref{thm1_SDE_xi_incomp_explicit}) and
(\ref{thm1_SDE_Htheta}) for general incomplete markets, are now explicitly
given. Furthermore, they no longer depend on the multiplier $\lambda_{t}%
^{\ast}$, thus are independent of investor's wealth level. Second, it is
straightforward to verify that the last component, i.e., $\pi^{u}%
(t,X_{t},Y_{t})$, in the optimal policy decomposition (\ref{thm_decomp1})
vanishes under the complete market model, as the investors do not need to
hedge the uncertainty in the investor-specific price of risk. We omit further
details, and refer to, e.g., \cite{Detemple_2003_JF}.

\section{Closed-form optimal policies for wealth-dependent HARA
utility\label{section:closed_form}}

In the previous section, we have revealed the fundamental impact of
wealth-dependent utility on optimal policy in general incomplete market
models. As discussed in Section \ref{Section:five_comp_decomposition}, such
impact essentially stems from the special structure of the investor-specific
price of risk, which implicitly depends on the investor's wealth level $X_{v}$
via the multiplier $\lambda_{v}^{\ast}$. In this and subsequent sections, we
show that the wealth impact on investor's portfolio allocation goes far beyond
this. In particular, we study the optimal policy under the HARA utility in
(\ref{hara_utility}), which is a representative wealth-dependent utility.
First, under nonrandom but possibly time-varying interest rate, we show that
the optimal policy under the HARA utility can be decomposed into a bond
holding scheme and the corresponding CRRA portfolio. Next, we solve the
optimal policies in closed-form under the Heston-SV model for investors with
HARA utility. We further calibrate the model using the S\&P 500 ETF data in
the recent ten years, which is then used in Section
\ref{section:numerical_all} to illustrate the static and dynamic impacts of
wealth-dependence and get further economic insights.

\subsection{Connection between optimal policies under HARA and CRRA utilities}

We first apply our decomposition results in Theorems
\ref{thm_representation_new} and \ref{thm_integral_equation} to explicitly
characterize the optimal policy under HARA utility. To save space, we provide
and discuss the results in Corollary \ref{Corollary_HARA} of
\ref{Appendix_policy_HARA}. Comparing with the CRRA policy in Section
\ref{section:crra_policy}, our decomposition clearly reveals the channels for
investor's wealth to affect the optimal policy under the HARA utility. These
channels are essentially introduced by the wealth constraints $\bar{c}$ and
$\bar{x}$ in (\ref{hara_utility}) for HARA investors.

Next, we study the optimal policy under HARA utility with nonrandom but
possibly time-varying interest rate. In this special case, we show the
investor-specific price of risk under HARA utility is indeed identical to that
under the corresponding CRRA utility. Furthermore, the optimal policies under
HARA and CRRA utilities are connected to each other by a simple multiplier
related to current wealth level and bond prices. This new relationship sheds
light on the construction of the optimal policy under the HARA utility. The
results are summarized in Proposition \ref{prop:hara_opt} below.

\begin{proposition}
\label{prop:hara_opt}With nonrandom but possibly time-varying interest rate
$r_{t}$, the investor-specific price of risk $\theta_{v}^{u}$ under HARA
utility (\ref{hara_utility}) coincides with its counterpart under CRRA
utility. It does not depend on the multiplier $\lambda_{v}^{\ast}$ and allows
the representation $\theta_{v}^{u}=\theta^{u}\left(  v,Y_{v};T\right)  .$ The
optimal policy under HARA utility satisfies the following simple ratio
relationship with its counterpart under CRRA utility:
\begin{equation}
\pi_{H}^{mv}(t,X_{t},Y_{t})=\pi_{C}^{mv}(t,Y_{t})\frac{\bar{X}_{t}}{X_{t}%
}\text{ and }\pi_{H}^{\theta}(t,X_{t},Y_{t})=\pi_{C}^{\theta}(t,Y_{t}%
)\frac{\bar{X}_{t}}{X_{t}}, \label{H_C_ratio}%
\end{equation}
as well as $\pi_{H}^{r}(t,X_{t},Y_{t})=\pi_{C}^{r}(t,Y_{t})\bar{X}_{t}%
/X_{t}\equiv0_{d}$ due to the deterministic nature of interest rate. Here, the
subscripts $H$ and $C$ represent for the HARA and CRRA utilities. Besides,
$\bar{X}_{t}$ in (\ref{H_C_ratio}) is given by%
\begin{subequations}
\begin{equation}
\bar{X}_{t}=X_{t}-\overline{x}B_{t,T}-\overline{c}\int_{t}^{T}B_{t,s}ds,\text{
for }w>0, \label{xtH_1}%
\end{equation}
and
\begin{equation}
\bar{X}_{t}=X_{t}-\overline{x}B_{t,T},\text{ for }w=0, \label{xtH_2}%
\end{equation}
where $B_{t,s}:=\exp(-\int_{t}^{s}r_{v}dv)$ is the price at time $t$ for a
zero-coupon bond with face value one that matures at time $s$.
\end{subequations}
\end{proposition}

\begin{proof}
See \ref{Appendix_proof_HARA_prop}.
\end{proof}

We now discuss the main findings in Proposition \ref{prop:hara_opt}. First,
the investor-specific prices of risk under the two utilities agree with each
other. That is, when there is no uncertainty in the interest rate, the
investor completes the market in\ exactly the same way under the two
utilities, and the impact of the current wealth level entirely vanishes in the
investor-specific price of risk for HARA investors. Second, the ratio
relationship (\ref{H_C_ratio}) connects the optimal policy under HARA utility
to its counterpart under the CRRA utility, which is much easier to solve in
closed form (or implement via, e.g., simulation methods) due to its
wealth-independent nature. The relationship (\ref{H_C_ratio}) provides a
convenient way to compute the optimal policy under HARA utility. Although we
assume a deterministic interest rate in Proposition \ref{prop:hara_opt}, no
assumptions are imposed on the state variable. Thus, the relationship
(\ref{H_C_ratio}) can be applied to various models with sophisticated state
variables and complex dynamics. In Section \ref{section:hn_closedform}, we
explicitly illustrate such an application using the Heston-SV model used in
\cite{liu2007portfolio}. Extension of Proposition \ref{prop:hara_opt} to the
case with random interest rates can be regarded as an open research topic, for
which the change of numeraire techniques in \cite{DetempleRindisbacher2010}
may render a useful tool.

The relationship (\ref{H_C_ratio}) allows for the following intuitive economic
interpretation. With a deterministic interest rate, $B_{t,s}$ represents the
time--$t$ price of a zero-coupon bond with face value one maturing at time
$s$. Thus, $\bar{X}_{t}$ given in (\ref{xtH_1}), i.e., $\bar{X}_{t}%
=X_{t}-\overline{x}B_{t,T}-\overline{c}\int_{t}^{T}B_{t,s}ds$ is the remaining
wealth after the investor buys $\overline{x}$ zero-coupon bonds maturing at
$T$ and a continuum of $\overline{c}\,ds$ zero-coupon bonds maturing at $s$
for all $s\in\lbrack t,T]$. The continuous payments from this bonds holding
position exactly render the minimum terminal wealth $\overline{x}$ and
intermediate consumption $\overline{c}$ required by the HARA utility
(\ref{hara_utility}), i.e., $\overline{x}$ at time $T$ and $\overline{c}\,ds$
at each $s\in\lbrack t,T]$. After purchasing the bonds, the HARA investor
allocates the remaining wealth following the optimal policy under the CRRA
utility, i.e., $\pi_{C}^{mv}(t,Y_{t})\bar{X}_{t}$ and $\pi_{C}^{\theta
}(t,Y_{t})\bar{X}_{t}$ for the mean-variance and price of risk hedge
components. It leads to the optimal policy under HARA utility given in
(\ref{H_C_ratio}). To summarize these insights, the HARA investor first buys a
series of zero-coupon bonds to satisfy the minimum requirements for terminal
wealth and intermediate consumptions over the entire investment horizon, then
allocates her remaining wealth just as a pure CRRA investor.\footnote{A
similar intuition is developed in \cite{detemple1992optimal}. They show that
under a complete market with deterministic coefficients, investors with habit
formation will first invest in a perfectly safe portfolio that finances habit
consumption and then invest as a standard CRRA investor.} This simple but
important relationship demonstrates again the application potential of our
decomposition results for general incomplete market models in Theorems
\ref{thm_representation_new} and \ref{thm_integral_equation}. Such a structure
explains how HARA investors maximize their expected utility while fulfilling
the minimum requirements for terminal wealth and intermediate consumptions.

From an economic aspect, we can view relationship (\ref{H_C_ratio}) as a
decomposition of the HARA policy that separates the roles of the state
variable and the investor's wealth level: the state variable $Y_{t}$ impacts
the optimal policy only via the CRRA policies $\pi_{C}^{mv}(t,Y_{t})$ and
$\pi_{C}^{\theta}(t,Y_{t})$, while the current wealth level $X_{t}$ impacts
the optimal policy only via the ratio $\bar{X}_{t}/X_{t}$. Moreover, the
decomposition (\ref{H_C_ratio}) implies the following behavior of the optimal
policy. First, the HARA investors allocate more on risky assets as their
wealth level increases, since the ratio $\bar{X}_{t}/X_{t}$ monotonically
increases in the current wealth level $X_{t}$. In the limit as $X_{t}%
\ $approaches infinity, the optimal policies $\pi_{H}^{mv}(t,X_{t},Y_{t})$ and
$\pi_{H}^{\theta}(t,X_{t},Y_{t})$ converge to their CRRA counterparts, as the
multiplier $\bar{X}_{t}/X_{t}$ converges to one. These findings reconcile the
analysis from the Arrow-Pratt relative risk aversion coefficient defined by
$\gamma^{U}(x):=-(\partial U(t,x)/\partial x)^{-1}x\partial U^{2}%
(t,x)/\partial x^{2}.$ Under the HARA utility for terminal wealth, it is given
by $\gamma^{U}(X_{t})=\gamma X_{t}/(X_{t}-\overline{x})$, which decreases in
$X_{t}$ and converges to its CRRA counterpart $\gamma$ as $X_{t}$ approaches
infinity. Thus, the decrease in risk aversion degree motivates HARA investors
to invest more in risky assets as their wealth level increases.

\subsection{Closed-form HARA policy under Heston-SV model
\label{section:hn_closedform}}

In what follows, we solve the optimal policy in closed form for HARA investors
under the Heston-SV model, applying the closed-form relationship of
Proposition \ref{prop:hara_opt}. We begin by setting up the model. The asset
price $S_{t}$ follows
\begin{subequations}
\begin{equation}
dS_{t}/S_{t}=(r+\lambda V_{t})dt+\sqrt{V_{t}}dW_{1t}, \label{Heston_St}%
\end{equation}
and the variance $V_{t}$ follows
\begin{equation}
dV_{t}=\kappa(\theta-V_{t})dt+\sigma\sqrt{V_{t}}(\rho dW_{1t}+\sqrt{1-\rho
^{2}}dW_{2t}), \label{Heston_Vt}%
\end{equation}
where $W_{1t}$ and $W_{2t}$ are two independent standard one-dimensional
Brownian motions. Here, the parameter $r$ denotes the risk-free interest rate;
the parameter $\lambda$ controls the market price of risk; the positive
parameters $\kappa,$ $\theta,$ and $\sigma$ represent the rate of
mean-reversion, the long-run mean, and the proportional volatility of the
variance process $V_{t}$. We assume the Feller's condition holds:\ $2\kappa
\theta>\sigma^{2}$. The leverage effect parameter $\rho\in\lbrack-1,1]$
measures the instantaneous correlation between the asset return and the change
in its variance.

The Heston-SV model belongs to the class of affine models
\citep{duffiepansingleton00}. This model and its generalizations can match a
number of stylized facts in stock return and variance dynamics, and are thus
widely applied in the empirical literature (see, e.g.,
\cite{poteshman2001underreaction} and \cite{pan2002jump}). For example, when
$\rho<0$, the Heston-SV model captures the \textquotedblleft leverage
effect\textquotedblright, i.e., the stock return and variance are negatively
correlated. Besides, the market price of risk (Sharpe ratio) of the risky
asset, $\lambda\sqrt{V_{t}}$, is proportional to $\sqrt{V_{t}}$ and thus
increases in the volatility. As noted in \cite{liu2007portfolio}, this feature
of the Heston-SV model is empirically supported by \cite{campbell1999force}.
In what follows, we solve the optimal policy under the Heston-SV model for
investors with HARA utility over the terminal wealth, i.e.,
\end{subequations}
\begin{equation}
U(T,x)=\left(  x-\overline{x}\right)  ^{1-\gamma}/(1-\gamma
)\ \label{HARA_utility}%
\end{equation}
for some risk aversion coefficient $\gamma>1$. Note that in the case of
$\overline{x}=0$, it reduces to the CRRA utility, i.e.,
\begin{equation}
U(T,x)=x^{1-\gamma}/(1-\gamma). \label{CRRA_utility}%
\end{equation}

While the optimal policy for CRRA investors under the Heston-SV model is
available in closed form (see \cite{liu2007portfolio}), its counterpart for
HARA investors, to our best knowledge, is absent in the literature. It is
probably due to the common consensus that HARA utility usually causes
mathematical inconvenience. We circumvent this challenge by applying
Proposition \ref{prop:hara_opt} on decomposing the optimal policy for HARA
investors, and deduce the next corollary.

\begin{corollary}
\label{corollary:Heston}Under the Heston-SV model in (\ref{Heston_St}) --
(\ref{Heston_Vt}) and the HARA utility (\ref{HARA_utility}) over terminal
wealth, the optimal policy $\pi_{H}(t,X_{t},V_{t})$ can be solved in
closed-form as
\begin{equation}
\pi_{H}(t,X_{t},V_{t})=\frac{\bar{X}_{t}}{X_{t}}\pi_{C}(t,V_{t}),
\label{HN_hara}%
\end{equation}
where $\bar{X}_{t}$ follows
\begin{equation}
\bar{X}_{t}=X_{t}-\overline{x}\exp(-r(T-t)), \label{Xbar_Heston}%
\end{equation}
and $\pi_{C}(t,V_{t})$ is the optimal policy under the corresponding CRRA
utility (\ref{CRRA_utility}). It admits the following closed-form
decomposition%
\begin{equation}
\pi_{C}(t,V_{t})=\pi_{C}^{mv}(t,V_{t})+\pi_{C}^{\theta}(t,V_{t}),
\label{piHN_crra}%
\end{equation}
where the mean-variance component and the price of risk hedge component are
given by $\pi_{C}^{mv}(t,V_{t})=\lambda/\gamma$ and $\pi_{C}^{\theta}%
(t,V_{t})=-\rho\sigma\delta\phi(T-t)$. The function $\phi(\tau)$ is defined
by
\[
\phi(\tau)=\frac{2\left[  \exp(\varsigma\tau)-1\right]  }{(\tilde{\kappa
}+\varsigma)\left[  \exp(\varsigma\tau)-1\right]  +2\varsigma},
\]
with $\tilde{\kappa}=\kappa-(1-\gamma)\lambda\rho\sigma/\gamma,$
$\delta=-\left(  1-\gamma\right)  \lambda^{2}/(2\gamma^{2})$, and
$\varsigma=\sqrt{\tilde{\kappa}^{2}+2\delta\sigma^{2}\left(  \rho^{2}%
+\gamma(1-\rho^{2})\right)  }$. Finally, the investor-specific price of risk,
i.e., the shadow price of market incompleteness, can be solved in closed form
as $\theta^{u}(t,V_{t};T)=\left(  0,\theta_{2}^{u}(t,V_{t};T)\right)  ^{\top}%
$, where
\begin{equation}
\theta_{2}^{u}(t,V_{t};T)=-\frac{1}{\rho}\gamma\sqrt{1-\rho^{2}}\sqrt{V_{t}%
}\pi_{C}^{\theta}(t,V_{t})=\gamma\sqrt{1-\rho^{2}}\sigma\delta\phi
(T-t)\sqrt{V_{t}}. \label{thetau_explicit}%
\end{equation}

\end{corollary}

\begin{proof}
See \ref{Appendix_proof_Corollary_Heston}.
\end{proof}

The above explicit results illustrate again how our decomposition can be
potentially applied to obtain closed-form optimal policies under sophisticated
settings. They will be further applied in Section \ref{section:numerical_all}
for conducting economically insightful comparative studies. Closed-form
solution (\ref{thetau_explicit}) for the investor-specific price of risk
$\theta^{u}$ provides a concrete example for solving the complex equation
system developed in Theorem \ref{thm_integral_equation}, demonstrating the
existence, uniqueness, and differentiability of the solution for that
particular model. Besides, straightforward calculation based on
(\ref{def_thetau}) renders $\theta_{2}^{u}(t,V_{t};T)=(\mu_{t}^{F}%
-r)/\sigma_{2}^{F}(t,V_{t})$, i.e., $\theta_{2}^{u}(t,V_{t};T)$ also equals
the Sharpe ratio of the fictitious asset. Similar to the market price of risk
$\lambda\sqrt{V_{t}}$, the investor-specific price of risk increases in
variance via the square-root term $\sqrt{V_{t}}$. However, it also
additionally depends on the remaining investment horizon $T-t$, as revealed in
representation (\ref{thetau_representation_v}).

\subsection{Empirical model estimation\label{Section_estimation}}

To achieve empirical validity up to date in the subsequent analysis, we
estimate the parameters of Heston-SV model based on the daily data of SPDR
S\&P 500 ETF from the recent ten years, i.e., 2010/1/3 to 2019/12/31. We
employ the maximum likelihood estimation approach developed by
\cite{AitSahalia_Li_Li_2020_MMLE}, which is an efficient method for estimating
continuous-time models with latent factors. The annualized parameter estimates
are obtained as
\begin{equation}
\kappa=12.5850,\text{ }\rho=-0.8141,\text{ }\lambda=6.6992,\text{ }%
\theta=0.0193,\text{ }r=0.0051,\text{ and }\sigma=0.5385.
\label{Heston_parameter}%
\end{equation}
The large value of $\kappa$ indicates that the variance process is highly
mean-reverting. The large negative value of $\rho$ implies a strong leverage
effect, i.e., the price changes of the risky asset and its variance process
are negatively correlated. Moreover, the market is characterized by a low
long-run volatility $\sqrt{\theta}\approx0.1389$, a low risk-free rate $r$,
and a high market price of risk coefficient $\lambda$. These features reflect
the behavior of the US financial market between 2010 and 2020, which is the
longest bull market in the history, associated with remarkably low volatility
and interest rate, as well as significantly high returns (e.g., the S\&P 500
Index surges by $185\%$ during this period).

\section{Static and dynamic impacts of
wealth-dependence\label{section:numerical_all}}

\label{section:wealth_impact}

In this section, we further reveal the impact of wealth-dependent utility on
optimal portfolio allocation using the representative Heston-SV model with
HARA utility. First, we investigate static impacts from wealth-dependent
utility. Second, we examine dynamic aspects, namely cycle-dependence of
optimal policy. Then, we quantify the dynamic impacts of the wealth-dependent
utility on investment performance for delegated portfolio management, i.e.,
how the investment performance varies with the investor's initial wealth level
and its interaction with other model parameters. It induces a risk-return
trade-off that is unseen under wealth-independent utilities. Next, we
demonstrate a novel hysteresis effect in the portfolio allocation, which stems
from cycle-dependence of optimal policy under the wealth-dependent HARA
utility. We show that the optimal policy and investment performance depend not
only on the composition of market states, but also on their sequence of
occurence. Thanks to the closed-form optimal policy under the estimated
Heston-SV model of Section \ref{section:hn_closedform}, we can use path
simulations to quantify hysteresis effect by taking different market scenarios
into account, instead of using just one stock price path. Our analysis
provides novel valuable understandings for delegated portfolio management.

\subsection{Static impacts from wealth-dependence\label{Section_static_impact}%
}

We now apply the closed-form formulae (\ref{HN_hara}) -- (\ref{piHN_crra})
under the Heston-SV model to explicitly examine how the optimal policy
$\pi_{H}(t,X_{t},V_{t})$ is impacted by the wealth-dependent property of the
HARA utility. In particular, we examine static impacts via three channels: the
wealth level $X_{t}$, the interest rate $r$, and the investment horizon $T-t$.
For the subsequent impact analysis, it helps that the corresponding optimal
policy $\pi_{C}(t,V_{t})$ under CRRA utility is independent of the wealth
level $X_{t}$ and the interest rate $r$. In the three panels of Figure
\ref{fig:HN_policy_vary}, we plot how the optimal policies vary with the
current wealth level $X_{t}$ (left), interest rate $r$ (middle), and
investment horizon $T-t$ (right). The red square and blue circle curves denote
the policy under CRRA and HARA utilities. \ 

By (\ref{HN_hara}), we see that the optimal policy $\pi_{H}(t,X_{t},V_{t})$
under HARA utility is impacted by the wealth level $X_{t}$ via the ratio
$\bar{X}_{t}/X_{t}$. By (\ref{Xbar_Heston}), we can express this ratio as
\begin{equation}
\frac{\bar{X}_{t}}{X_{t}}=1-\left(  \frac{X_{t}}{\overline{x}}\right)
^{-1}\exp(-r(T-t)). \label{XbarX_formula}%
\end{equation}
Here, $X_{t}/\overline{x}$ measures the current wealth level $X_{t}$ relative
to the minimum requirement $\overline{x}$. This formula explicitly shows that
$\bar{X}_{t}/X_{t}$ increases concavely in $X_{t}$ with $\lim_{X_{t}%
\rightarrow\infty}\bar{X}_{t}/X_{t}=1$. The left panel of Figure
\ref{fig:HN_policy_vary} shows that the impact of wealth level can be
substantial. As $X_{t}$ increases from $2\overline{x}$ to $10\overline{x}$,
the allocation on the risky asset from HARA investors increases by
approximately $70\%$ from $97\%$ to $167\%$. But such an impact diminishes as
the wealth level becomes higher, i.e., the minimum requirement constraint
becomes less binding. Such a concave increase in the allocation for risky
asset with respect to investor's wealth is empirically observed in househould
finance literatue; see, e.g., \cite{Roussanov2010},
\cite{wachter2010household}, and \cite{calvet2014twin}.%

\begin{figure}[H]%
%

\begin{center}%
\begin{tabular}
[c]{c}%
{\includegraphics[
trim=0.795944in 0.000000in 0.596958in 0.000000in,
height=2.2972in,
width=6.6608in
]%
{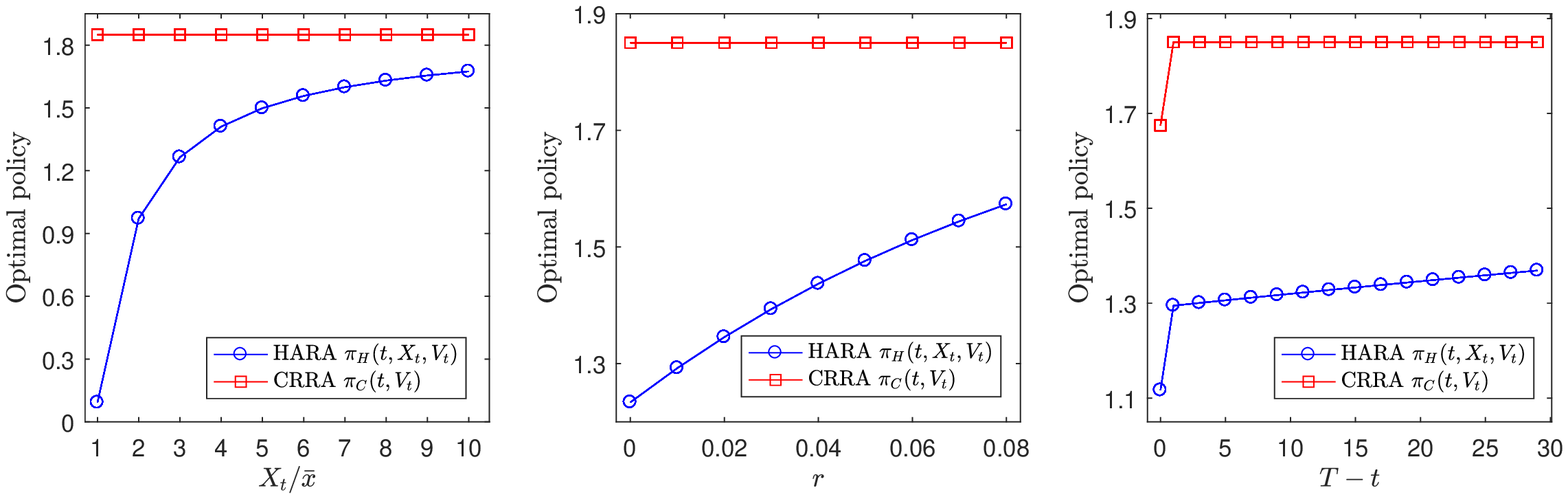}%
}
\end{tabular}%
\caption
{Behaviors of optimal policy in the Heston-SV model under the HARA and CRRA utilities.}%
\label{fig:HN_policy_vary}%
\end{center}%
%

\renewcommand{\baselinestretch}{1.0}%
\begin{small}%
The three panels plot the optimal policies in the Heston-SV model under the
HARA (blue circle) and CRRA (red square) utilities at different wealth level
(left), interest rate (middle), and investment horizon (right). These figures
are generated using our estimated parameters in (\ref{Heston_parameter}) with
the risk aversion parameter $\gamma=4$. Besides, we set $T-t=10$ years in the
left and middle panels, and $X_{t}/\overline{x}=3$ in the middle and right
panels.%
\end{small}%
\noindent
\end{figure}%

We next analyze how the optimal policy changes with the interest rate $r$.
With a higher interest rate, the risk-free asset delivers a higher return.
Thus, we might expect that investors should increase their allocation on the
risk-free asset, and accordingly decrease their allocation on the risky asset.
However, thanks to our closed-form optimal policy, we can show that the
investors with HARA utility will actually behave in the opposite way, i.e.,
they will increase their allocation on the risky asset under a higher interest
rate. Indeed, by the closed-form formulae (\ref{HN_hara}) and
(\ref{Xbar_Heston}), we find that the optimal policy $\pi_{H}(t,X_{t},V_{t})$
depends on the interest rate $r$ via the ratio $\bar{X}_{t}/X_{t}$, which
increases concavely in $r$ according to (\ref{XbarX_formula}). We can
interpret this behavior according to the optimal investment strategy described
after Proposition \ref{prop:hara_opt}. While a higher risk-free rate does not
impact the optimal policy under CRRA utility, it lowers the bond price
$B_{t,T}=\exp(-r(T-t))$. So, it costs less for HARA investors to ensure their
minimum requirement $\overline{x}$ by investing in bonds, and thus increases
their remaining wealth $\bar{X}_{t}$ for investing in the risky asset
following the CRRA policy. The above analysis is demonstrated by the middle
panel of Figure \ref{fig:HN_policy_vary}, where we observe that the optimal
policy $\pi_{H}(t,X_{t},V_{t})$ increases concavely with $r$. The magnitude of
such an impact is indeed sizeable: the optimal policy $\pi_{H}(t,X_{t},V_{t})$
increases by approximately $30\%$ as the interest rate increases from $0$ to
$0.08$.

Finally, by the closed-form formulae (\ref{piHN_crra}), the investment horizon
$T-t$ affects the optimal policy $\pi_{C}(t,V_{t})$ for CRRA investors only
through the price of risk hedge component $\pi_{C}^{\theta}(t,V_{t}%
)=-\rho\sigma\delta\phi(T-t)$. When $\rho<0$, this component increases
monotonically in $T-t$. However, for HARA investors, in addition to the
dependence through the corresponding CRRA policy as mentioned above, their
optimal policy is impacted by $T-t$ also via the ratio $X_{t}/\overline{x}$
according to (\ref{HN_hara}) and (\ref{XbarX_formula}), which increases
concavely in $T-t$. The right panel of Figure \ref{fig:HN_policy_vary} shows
that both the optimal policies $\pi_{C}(t,V_{t})$ and $\pi_{H}(t,X_{t},V_{t})$
increase with the investment horizon $T-t$. For the CRRA case, the sharp
increase due to the price of risk hedge component mainly occurs when the
investment horizon is short; under longer investment horizons, the optimal
policy becomes almost insensitive to $T-t$. For the HARA case, however, beside
the similar sharp increase for short horizons, $\pi_{H}^{\theta}(t,X_{t}%
,V_{t})$ keeps increasing in $T-t$ even for longer investment horizons. We can
interpret these behaviors as follows. Under CRRA utility, a longer investment
horizon increases the uncertainty in the price of risk. It leads to a larger
hedging demand for risk-averse investors, and results in a larger hedge
component in the optimal policy. On the other hand, the increase in $\pi
_{H}^{\theta}(t,X_{t},V_{t})$ under the HARA utility is generated by a
combination of two effects: first, the decrease of the bond price
$B_{t,T}=\exp(-r(T-t))$ and thus the increase of the remaining wealth $\bar
{X}_{t}$ in (\ref{Xbar_Heston}), and second, the increase of the corresponding
CRRA hedging demand as discussed above. The first effect is significant even
for long investment horizons, leading to a more lasting impact from the
investment horizon.

The above comparative analysis illustrates how we can apply our theoretical
decompositions to understand the behavior of the optimal policy under
incomplete market models with wealth-dependent utilities. In particular, we
have explicitly shown that the wealth-dependent\ property of the HARA utility
should not be taken only \textquotedblleft literally\textquotedblright.
Precisely speaking, by Proposition \ref{prop:hara_opt}, the optimal policy of
HARA investors is determined by the remaining wealth $\bar{X}_{t}$, which
depends on the current wealth level $X_{t}$, interest rate $r$, and investment
horizon $T-t$ according to (\ref{Xbar_Heston}). It gives further understanding
of what is induced by the wealth-dependent property of the HARA utility.

\subsection{Cycle-dependence of optimal policy\label{sec:cycle_dependence}}

In addition to the above static analysis of optimal policies at an arbitrary
fixed time $t,$ we now employ the Heston-SV model as representative model to
analyze the significant wealth impact of HARA utility from a dynamic
perspective, i.e., how the wealth-dependent property of HARA utility interacts
with the complex market dynamics in determining the optimal allocation
strategy and overall investment performance. In the literature, the vast
majority of studies on optimal dynamic portfolio allocation focuses on solving
optimal policies and then conducting static impact analysis from model
specifications under wealth-independent utilities; see e.g., \cite{Wachter02},
\cite{chacko2005dynamic}, and \cite{liu2007portfolio}, among others. As an
important issue in our setting, the dynamic impacts on investment performance
over different asset allocation strategies and/or investor profiles remain to
be investigated. Here, we aim to shed light on this under an environment with
stochastic volatility and wealth-dependent utility. A recent study on the
dynamic impacts from different policies can be found in \cite{MoreiraMuir2019}%
. It relies on a stochastic volatility model and shows that ignoring the hedge
component in the optimal policy would lead to a substantial utility loss.
However, \cite{MoreiraMuir2019} assume that the investor shares the recursive
utility of \cite{duffie1992stochastic}, which generalizes the
wealth-independent CRRA utility. Then, the derived optimal policy is
independent of the investor's wealth level.

In contrast, the wealth-dependent property of HARA utility leads to a
sophisticated dependence of the optimal portfolio policy on the entire path of
asset price and its volatility. This feature is totally absent under the
simple CRRA utility. The ratio $\bar{X}_{t}/X_{t}$ given in
(\ref{XbarX_formula}) is determined by the investor's wealth level $X_{t}$,
which evolves according to the following explicit dynamics%
\begin{equation}
dX_{t}=rX_{t}dt+\left[  X_{t}-\overline{x}\exp(-r(T-t))\right]  [\lambda
/\gamma-\rho\sigma\delta\phi(T-t)][\lambda V_{t}dt+\sqrt{V_{t}}dW_{1t}].
\label{wealth_equation_HARA}%
\end{equation}
It is obtained by plugging in the HARA policy (\ref{HN_hara}), specifying
$\mu_{t}=r+\lambda V_{t}$, $\sigma_{t}=\sqrt{V_{t}}$, as well as $r_{t}=r$ in
(\ref{opt_constraint}) according to the specification of Heston-SV model in
(\ref{Heston_St}), and setting $c_{t}\equiv0$ since intermediate consumption
is not included in investor's utility. The wealth dynamics
(\ref{wealth_equation_HARA}) obviously depends on the entire path of the stock
price $S_{t}$ and variance $V_{t},$ due to their common Brownian shock
$W_{1t}$. Thus, the optimal policy of HARA investors is also influenced by the
historical path of market dynamics through the ratio $\bar{X}_{t}/X_{t}$ in
(\ref{HN_hara}). On the contrary, under CRRA utility, we see from
(\ref{piHN_crra}) that the optimal policy $\pi_{C}(t,V_{t})$ is simply
independent of both current wealth level $X_{t}$ and variance $V_{t}$, and
only depends on the remaining investment horizon $T-t$. Thus, the optimal
policy of CRRA investor is entirely deterministic and does not depend on the
realization of the model dynamics (\ref{Heston_St}) -- (\ref{Heston_Vt}). This
stark contrast has important implications on delegated portfolio management.
The delegated portfolio manager for HARA investors must adjust the optimal
policy in a stochastic way according to the realization of market dynamics,
while, in contrast, portfolio managers for CRRA investors can simply determine
the optimal policy for the entire investment horizon by (\ref{piHN_crra}) at
the beginning of the period.

Let us illustrate the dynamic impact of the HARA utility, as discussed above
in theory, by examining how the initial wealth level of HARA investors affects
their optimal portfolio allocation over the entire investment horizon.
Consider a market where the stock price $S_{t}$ and its variance $V_{t}$
follow the Heston-SV model with parameters given in (\ref{Heston_parameter}).
Without loss of generality, the initial price and variance are set as
$S_{0}=100$ and $V_{0}=\theta=0.0195$. We consider two investors with HARA
utilities over terminal wealth for an investment horizon of $T=10$ years.
Their risk aversion coefficient and minimum requirement for terminal wealth
are set to $\gamma=4$ and $\overline{x}=10^{6}$, i.e., one million. The two
investors only differ in their initial wealth levels: the high-wealth investor
has an initial wealth of $X_{0}^{H}=5\times10^{6}$, while the low-wealth
investor has an initial wealth of $X_{0}^{L}=2\times10^{6}$. Thus, their
ratios of initial wealth over the minimum requirement are equal to $X_{0}%
^{H}/\overline{x}=5$ and $X_{0}^{L}/\overline{x}=2$. Denote the optimal
policies of the two investors by $\pi_{t}^{H}$ and $\pi_{t}^{L}$. Then, the
ratio $\pi_{t}^{H}/\pi_{t}^{L}$ measures how the optimal policy of the
high-wealth investor differs from that of the low-wealth investor. We simulate
market scenarios and the corresponding dynamics of optimal policies $\pi
_{t}^{H}$ and $\pi_{t}^{L}$ for the entire investment period. We conduct the
simulations using a standard Euler scheme on the Heston-SV model
(\ref{Heston_St}) -- (\ref{Heston_Vt}). Along the simulated path, we evaluate
the optimal policies $\pi_{t}^{H}$ and $\pi_{t}^{L}$ via (\ref{HN_hara}), and
the investor's wealth evolves according to equation
(\ref{wealth_equation_HARA}).

Figure \ref{fig: HN_simu_path} shows a representative market scenario, i.e., a
simulated path of the policy ratio $\pi_{t}^{H}/\pi_{t}^{L}$ (in red
dashdotted with the right $y$-axis in both panels), and the corresponding
paths of stock price $S_{t}$ (in blue solid with the left $y$-axis in the
upper panel) and its realized variance $RV_{t}$ (in black solid with the left
$y$-axis in the lower panel). Given the simulated path of the spot variance
$V_{t}$, we approximate the realized variance by averaging the daily spot
variance over the past one-month time window, i.e., $RV_{t}=\Sigma_{i=0}%
^{n-1}V_{t-i\Delta}/n$ with $\Delta=1/252$ and $n=22$. We classify the market
regimes following the method proposed by \cite{lunde2004duration}: a
transition from a bear to bull (resp. bull to bear) market is triggered when
the stock price increases by $20\%$ (resp. drops by $15\%$) from its lowest
(resp. highest) level in the current bear (resp. bull) market. {We represent
the bear markets by the shaded areas in Figure \ref{fig: HN_simu_path}. Not
surprisingly, the realized variance spikes during bear markets due to the
leverage effect, i.e., the changes of stock prices and variance are negatively
correlated.}

Recall that the optimal policy under CRRA utility, as given in
(\ref{piHN_crra}), is independent of the variance level $V_{t}$, stock price
$S_{t}$, and wealth level $X_{t}$. Thus, we simply have $\pi_{t}^{H}/\pi
_{t}^{L}=1$ for all $t$ under CRRA utility, i.e., the optimal policies from
the high-wealth and low-wealth investors always coincide. However, under HARA
utility, we observe that the ratio $\pi_{t}^{H}/\pi_{t}^{L}$ varies markedly
during the investment period. The relative difference in $\pi_{t}^{H}$ and
$\pi_{t}^{L}$ can be as large as $50\%$, i.e., the ratio $\pi_{t}^{H}/\pi
_{t}^{L}$ can exceed $1.5$. In particular, it is negatively correlated with
the stock price, but positively correlated with the realized volatility. It is
not an incidental result out of a specific path. With $10^{4}$ trials of
simulations, we find that the average correlation between $\pi_{t}^{H}/\pi
_{t}^{L}$ and $S_{t}$, resp.\ $RV_{t}$, is approximately $-0.90$,
resp.\ $0.10$; both are statistically significant at $0.1\%$ level. We can
interpret such a pattern as follows. Under our estimated Heston-SV model, both
HARA investors hold positive positions in the risky asset, so their wealth
levels decrease when the stock price drops. Moreover, according to our
closed-form formulae (\ref{HN_hara}) and (\ref{XbarX_formula}), the impact of
the wealth level on the optimal policy is more significant when the wealth
level is low, i.e., the minimum constraint is more binding, as shown by the
left panel of Figure \ref{fig:HN_policy_vary}. Thus, when stock price drops
(and realized variance increases from the leverage effect), we expect to see a
larger difference in the optimal policies of high- and low-wealth HARA
investors, resulting in a higher ratio $\pi_{t}^{H}/\pi_{t}^{L}$. In contrast,
in the bull market when stock price increases, both investors are wealthier
and behave more similarly, it leads to a smaller ratio $\pi_{t}^{H}/\pi
_{t}^{L}$.%

\begin{figure}[H]%
%

\begin{center}%
\begin{tabular}
[c]{c}%
{\includegraphics[
trim=0.422756in 0.282893in 0.000000in 0.284227in,
height=4.0568in,
width=5.5478in
]%
{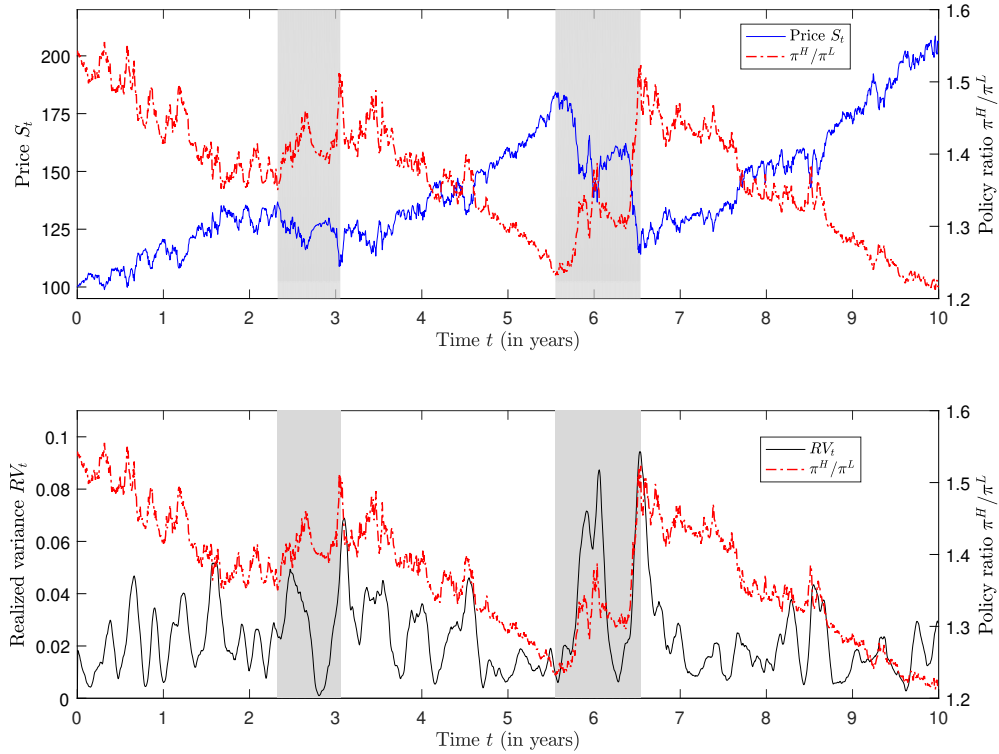}%
}
\end{tabular}%
\caption{A simulated path of stock price $S_{t}%
$, its realized variance $RV_{t}$, and optimal policy ratio $\pi_{t}^{H}%
/\pi_{t}^{L}$.}%
\label{fig: HN_simu_path}%
\end{center}%
%

\renewcommand{\baselinestretch}{1.0}%
\begin{small}%
This figure plots a simulated path of stock price $S_{t}$ in blue solid with
left $y$-axis in the upper panel, realized variance $RV_{t}$ in black solid
with left $y$-axis in the lower panel, and that of the corresponding optimal
policy ratio between the high-- and low--wealth investors under HARA utility,
i.e., $\pi_{t}^{H}/\pi_{t}^{L},$ in red dashdotted with right $y$-axis in both
two panels. The ratios are calculated by the closed-form formulae
(\ref{HN_hara}) and (\ref{piHN_crra}). The shaded areas depict the bear market
regimes. The parameters for the Heston SV model are chosen as our estimated
ones in (\ref{Heston_parameter}).%
\end{small}%
\noindent
\end{figure}%

This example illustrates the impact of HARA utility from a dynamic
perspective. As opposed to being fully determined by their different initial
wealth levels, the investor's optimal investment decisions indeed depend on
the historical path of market performance. In particular, cycles matter for
HARA investors. As we can see from Figure \ref{fig: HN_simu_path}, the optimal
policies of high- and low-wealth HARA investors tend to diverge, as shown by a
larger ratio $\pi_{t}^{H}/\pi_{t}^{L}$, during the bear markets (marked by the
shaded areas), where stock price drops and volatility spikes. However, this
important market cycle dependence is totally absent under the
wealth-independent CRRA utility. Such a dependence is important for delegated
portfolio management. In particular, investment advices can be considerably
erroneous and lead to completely suboptimal strategies if we ignore the
wealth-dependence property of investor's HARA utility.

\subsection{Impacts on investment performance for delegated portfolio
management\label{sec:impact_performance}}

The previous analysis documents that the initial wealth of HARA investor
induces a substantial impact on their portfolio allocation. In particular, the
HARA investor with a higher initial wealth will allocate more of her wealth to
the risky asset, as the ratio $\pi_{t}^{H}/\pi_{t}^{L}$ is greater than $1$
over the investment horizon. Thus, it is intuitive that a high-wealth HARA
investor enjoys higher returns but is exposed to more risk. However, precisely
quantifying this intuition on investment performance remains an open problem
with practical relevance for delegated portfolio management. For example, how
do the return mean and Sharpe ratio from the optimal investment strategy vary
for HARA investors with different initial wealth levels? While the impact
analysis as shown in Figure \ref{fig: HN_simu_path} partly sheds light on such
issues, it cannot give reliable answers owing to the limitation of capturing
only one specific market scenario among others. This type of limitation is
common in backtesting investment strategies. In contrast, thanks to our
closed-form optimal policy (\ref{HN_hara}), we can precisely and
quantitatively analyze the over-all impacts on investment performance by
taking various market scenarios into account under our estimated Heston-SV
model. Specifically, for the purpose of covering different market scenarios,
we simulate a large number of paths for the investment problem for HARA
investors with different levels of initial wealth.\footnote{Without the
closed-form optimal policy (\ref{HN_hara}), this task would become
computationally intensive due to the nested simulations for evaluating the
optimal policies during the investment period.}

For each simulated path, we compute the excess return mean, volatility, and
Sharpe ratio over the entire investment horizon $[0,T]$ using the\emph{
}annualized daily excess returns of the wealth, i.e., $R_{i}=\ln(X_{i\Delta
}/X_{(i-1)\Delta})/\Delta-r$ with $\Delta=1/252$ and $i=1,2,...,N,$ where $N$
denotes the total number of trading days in the entire investment horizon,
i.e., $T=N\Delta$. The excess return mean and volatility are calculated as the
sample mean and standard deviation of daily excess returns, i.e., $\bar{R} =
\sum_{i=1}^{N} R_{i}/N$ and $SD = \sqrt{\sum_{i=1}^{N}(R_{i}-\bar{R}%
)^{2}/(N-1)}$. Then, the Sharpe ratio is simply computed as $SR=\bar{R}/SD$.
Besides, we also compute the maximum drawdown of wealth, that measures the
risk of extreme losses in the investment strategy. It is calculated as
$MD=\max_{0\leq n\leq N}(1-M_{n}^{\text{min}}/M_{n}^{\text{max}})$,
where$\ M_{n}^{\text{max}}$ and $M_{n}^{\text{min}}$ denote the running
maximum and minimum of investor's wealth until day $n$, i.e., $M_{n}%
^{\text{max}}=\max_{0\leq k\leq n}X_{k\Delta}$ and $M_{n}^{\text{min}}%
=\min_{0\leq k\leq n}X_{k\Delta}.$ Then, to evaluate the over-all performance
across different market scenarios, we compute the average of these
time--series performance statistics across all the simulated paths. Those
averages estimate the expectations $E[\bar{R}],$ $E[SD],$ $E[SR],$ and $E[MD]$
over different realizations of market dynamics. In this implementation, we
employ a large sample of $10^{4}$ simulated paths for each initial wealth
level. The left panel of Figure \ref{fig:HN_performance} plots the estimates
for average excess return means (red square) and volatilities (blue circle)
for different initial wealth levels; the middle and right panel plot the
estimates for average Sharpe ratio and maximum drawdown. The dashed flat lines
in each panel show the corresponding performance statistics of a CRRA
investor, whose policy is independent of the initial wealth level.%

\begin{figure}[H]%
%

\begin{center}%
\begin{tabular}
[c]{c}%
{\includegraphics[
trim=0.914232in 0.029009in 0.925660in 0.218484in,
height=2.1352in,
width=6.5017in
]%
{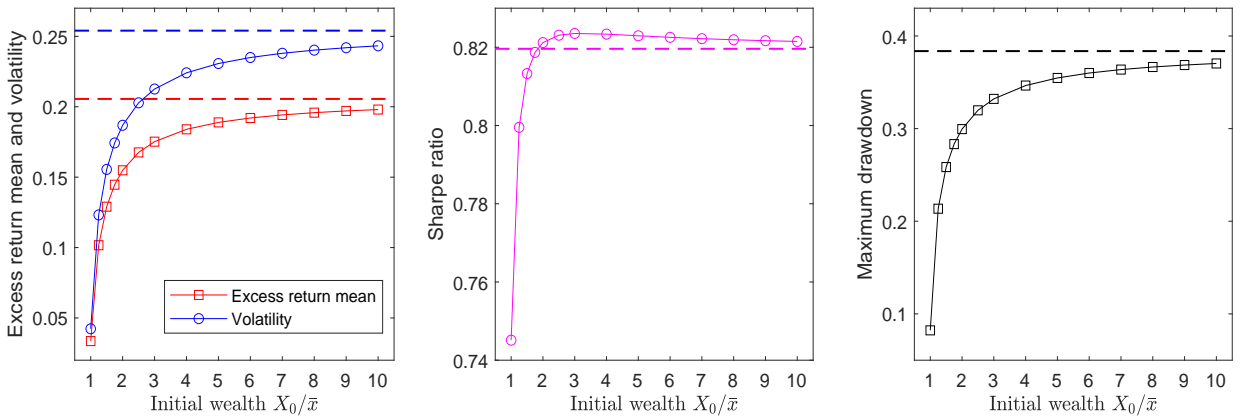}%
}
\end{tabular}%
\caption
{Performance statistics of dynamic optimal portfolio allocation with different initial wealth levels under the HARA utility.}%
\label{fig:HN_performance}%
\end{center}%
%

\renewcommand{\baselinestretch}{1.0}%
\begin{small}%
This figure the estimates of average excess return mean and volatility (left),
Sharpe ratio (middle), and maximum drawdown (right) for the investment problem
under HARA utility and our empirically estimated Heston-SV model given in
(\ref{Heston_parameter}), given that investors dynamically follow the optimal
policy. The corresponding levels under CRRA utility are represented by the
dashed lines. The averages are computed over $10^{4}$ simulated paths. We set
the risk aversion and investment horizon as $\gamma=4$ and $T-t=10$ years.%
\end{small}%
\noindent
\end{figure}%

Figure \ref{fig:HN_performance} reveals that the initial wealth level greatly
influences the investment performance as expected from the results of the
previous section. As shown in Figure \ref{fig: HN_simu_path}, HARA investors
with a higher initial wealth invest more in the risky asset, i.e., $\pi
_{t}^{H}/\pi_{t}^{L}>1$. In Figure \ref{fig:HN_performance}, we see that, on
average, it leads to a higher mean and Sharpe ratio, but also a larger
volatility and maximum drawdown. From this perspective, we can interpret the
impact on investment performance as a risk-return trade-off induced by
different initial wealth levels of HARA investors. Such a trade-off is
sizeable: as the initial wealth $X_{0}$ increases from $\overline{x}$ to
$10\overline{x}$, the average excess return mean increases from $3.5\%$ to
$23.5\%$, while the average excess return volatility increases from $4.2\%$ to
$25.4\%$ and the average maximum drawdown of wealth increases from $8.2\%$ to
$38.1\%$. Thus, a high-wealth HARA investor with initial wealth $X_{0}%
=10\overline{x}$ expects to earn an annual excess return of more than $20\%$,
but also to experience a more than $38\%$ loss of the wealth due to potential
market crashes. On the other hand, a HARA investor with relatively low initial
wealth $X_{0}=\overline{x}$ will allocate only a tiny portion of her wealth to
the risky asset, because, by the investment strategy described after
Proposition \ref{prop:hara_opt} for HARA investors, most of her wealth will be
allocated to the risk-free asset in order to ensure that the minimum
requirement for terminal wealth can always be satisfied. Such a conservative
strategy leads to a low average annual excess return mean of $3.5\%$, but also
a low average volatility ($4.2\%$) and maximum drawdown ($8.2\%$). The
above-mentioned sharp contrast between the investment performance of high- and
low-wealth HARA investors shows again that it is crucial to understand the
wealth-dependence property of the utility in delegated portfolio management.
Besides, the impact of the initial wealth is more significant at low wealth
levels, and quickly decays as $X_{0}/\overline{x}$ further increases.

Let us now provide an analysis for the Sharpe Ratio. In the above, we know
that HARA investors with higher initial wealth will allocate more to the risky
asset, and it naturally leads to higher average excess return mean and
volatility, as well as larger maximum drawdown. However, the increase of
average Sharpe ratio with the initial wealth of HARA investors, as shown in
the middle panel of Figure \ref{fig:HN_performance}, appears to be puzzling at
first sight. In the following, we show that it can be attributed to the
cycle-dependence of optimal policy, which essentially originates from the
wealth-dependent property of investor's utility.

We begin by introducing the following notations to ease exposition. Since the
daily time increment $\Delta=1/252$ is small, we can approximate the
annualized daily excess return as $R_{i}=\ln(X_{i\Delta}/X_{(i-1)\Delta
})/\Delta-r\approx\pi_{i\Delta}[\lambda V_{i\Delta}+\sqrt{V_{i\Delta}}Z_{i}],
$ according to the wealth dynamics (\ref{opt_constraint}) under the Heston-SV
model, where $Z_{i}=(W_{(i+1)\Delta}-W_{i\Delta})/\Delta$ is a normal random
variable with zero mean and variance $1/{\Delta}$. Thus, the excess return
$R_{i}$ is approximately linear in the optimal policy $\pi_{i\Delta}$. As we
mentioned, the HARA investors with higher initial wealth will allocate more
wealth on the risky asset, leading to larger $\pi_{i\Delta}$. To eliminate the
impact from the overall levels of $\pi_{i\Delta}$, we normalize the excess
returns $R_{i}$ by the average optimal policy over the entire investment
horizon, i.e., $\bar{\pi}=\sum_{i=1}^{N}\pi_{i\Delta}/N.$ That is, we define
the scaled excess return $R_{i}^{\prime}$ as $R_{i}^{\prime}:=R_{i}/\bar{\pi
}\approx\pi_{i\Delta}^{\prime}[\lambda V_{i\Delta}+\sqrt{V_{i\Delta}}Z_{i}], $
where $\pi_{i\Delta}^{\prime} :=\pi_{i\Delta}/\bar{\pi}$ denotes the scale
optimal policy. We then compute the mean and volatility of the scaled excess
return $R_{i}^{\prime}$ as $\bar{R}^{\prime}=\sum_{i=1}^{N}R_{i}^{\prime}/N$
and $SD^{\prime}=\sqrt{\sum_{i=1}^{N}(R_{i}^{\prime}-\bar{R}^{\prime}%
)^{2}/(N-1)}$. For each simulated path, it is easy to verify the Shape ratio
coincides with the one calculated based on the scaled excess returns, i.e.,
$SR= \bar{R}/SD =\bar{R}^{\prime}/SD^{\prime} $ since we have $\bar{R}%
^{\prime}=\bar{R}/\bar{\pi}$ and $SD^{\prime}=SD/\bar{\pi}.$ Thus, while the
average optimal policy $\bar{\pi}$ significantly impacts both the return mean
$\bar{R}$ and volatility $SD$, it does not directly affect the Sharpe ratio.

Those computations show that the average optimal policy $\bar{\pi}$, which
increases in the initial wealth level of HARA investors, cannot explain the
variation in Sharpe ratio observed in the middle panel of Figure
\ref{fig:HN_performance}. Instead, we can explain the increasing pattern of
Sharpe ratio by the cycle-dependence property of optimal policy under HARA
utility. It can be interpreted as follows. As we analyzed in Section
\ref{sec:cycle_dependence}, the optimal policy under HARA utility depends on
the historical path of market performance. Moreover, due to the wealth
constraint, such cycle-dependence is more pronounced for low-wealth HARA
investors, i.e., the portfolio allocation of low-wealth HARA investors is more
sensitive to the path of market performance. This cycle-dependence introduces
additional\textit{ }uncertainty in the scaled optimal policy $\pi_{i\Delta
}^{\prime}$ for HARA investors with low initial wealth, even after we
normalize it with the average level $\bar{\pi}$. It leads to more volatility
in the scaled excess return $R_{i}^{\prime}$ and thus a lower Sharpe ratio for
low-wealth HARA investors. On the other hand, the optimal policy of
high-wealth HARA investors is less sensitive to the market performance; it
reduces the uncertainty in their scaled optimal policy $\pi_{i\Delta}^{\prime
}$ and thus produces a higher Sharpe ratio.

We use simulations to illustrate the additional uncertainty in the scaled
optimal policy $\pi_{i\Delta}^{\prime}$ for low-wealth HARA investors. Figure
\ref{fig:SR2} exhibits the representative quantiles of the scaled optimal
policy $\pi_{i\Delta}^{\prime}$ for HARA investors with initial wealth
$X_{0}/\overline{x}=1$ (top) and $X_{0}/\overline{x}=4$ (bottom). For ease of
comparison, we use the same vertical axis in the two panels. We show the
quantiles at the beginning of each quarter in the investment horizon, which
are computed based on $10^{4}$ simulation trials. The drop near the end is due
to the decrease in the price of risk hedge component when investment horizon
shrinks to zero, as analyzed in Section \ref{Section_static_impact}.
Comparisons between the two panels, i.e., the low- vs. high-wealth investors,
clearly support our interpretations. The distributions of $\pi_{i\Delta
}^{\prime}$ spread out over much wider ranges for the low-wealth investor,
suggesting more uncertainty in her optimal policy. It explains the increasing
pattern of the expected Sharpe ratio in the middle panel of Figure
\ref{fig:HN_performance}.%

\begin{figure}[H]%
%

\begin{center}%
\begin{tabular}
[c]{c}%
{\includegraphics[
trim=0.496437in 0.241692in 0.467130in 0.230674in,
height=4.6185in,
width=5.0635in
]%
{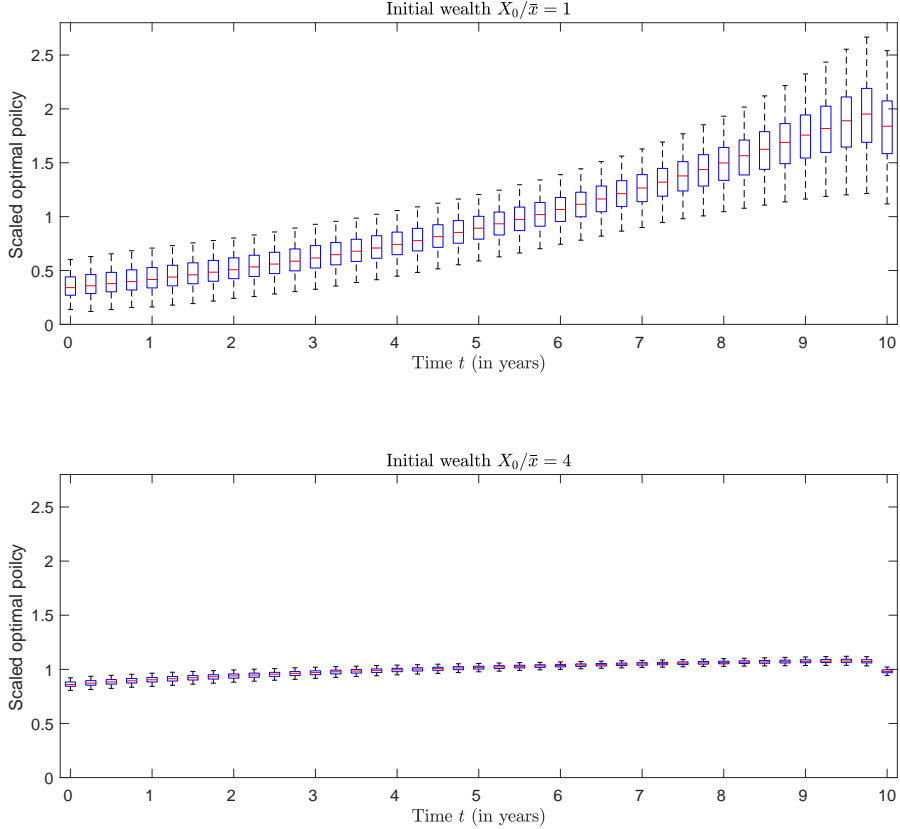}%
}
\end{tabular}%
\caption{Representative quantiles of the scaled optimal policy $\pi_{i\Delta
}^{\prime}$  for HARA investors}%
\label{fig:SR2}%
\end{center}%
%

\renewcommand{\baselinestretch}{1.0}%
\begin{small}%
This figure summarizes the distribution of the scaled optimal policy
$\pi_{i\Delta}^{\prime}$ at every quarter for HARA investors with initial
wealth $X_{0}/\overline{x}=1$ (top) and $X_{0}/\overline{x}=4$ (bottom), based
on the simulation trials under the estimated Heston-SV model given in
(\ref{Heston_parameter}). Besides, we set the risk aversion and investment
horizon as $\gamma=4$ and $T-t=10$ years. The box plots gather the 2.5th and
97.5th percentiles, the first and third quartiles, as well as the median, by
using the short bars at the ends of two whiskers, the upper and lower edges of
the blue box, as well as the red bar inside the box.
\end{small}%
\noindent
\end{figure}%

\subsection{Interaction between wealth level and model parameters}

In Section \ref{Section_static_impact}, we have shown that the
wealth-dependent property of HARA utility can affect the optimal policy via
other channels besides the current wealth level, i.e., the interest rate $r$
and investment horizon $T$. In this section, we study such effects from a
dynamic perspective, i.e., the impact of interest rate and investment horizon
on the investment performance for HARA investors. Moreover, we analyze how
these impacts vary for HARA investors with different initial wealth levels.
Such variation reflects the dynamic interaction between model parameters and
investor's wealth level.

\bigskip\bigskip%
\begin{figure}[H]%
%

\begin{center}%
\begin{tabular}
[c]{c}%
{\includegraphics[
trim=0.919478in 0.000000in 0.920577in 0.000000in,
height=2.1352in,
width=6.3884in
]%
{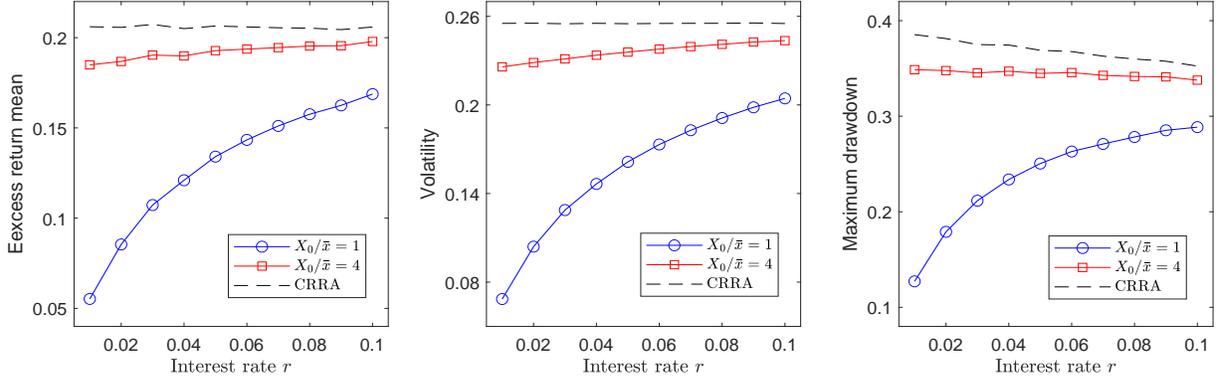}%
}
\end{tabular}%
\caption
{Performance statistics of dynamic optimal portfolio allocation with different interest rate under the HARA utility.}%
\label{fig:HN_performance_interest_rate}%
\end{center}%
%

\renewcommand{\baselinestretch}{1.0}%
\begin{small}%
This figure plots the average excess return mean (left), volatility (middle),
and maximum drawdown (right) at different interest rate $r$ for the investment
problem under HARA utility and our empirically estimated Heston-SV model given
in (\ref{Heston_parameter}). We consider three investors: low-wealth HARA
investor with $X_{0}/\overline{x}=1$ (blue circle), high-wealth HARA investor
with $X_{0}/\overline{x}=4$ (red square), and a CRRA investor (black dashed).
The averages are computed over $10^{4}$ simulated paths. We set the risk
aversion and investment horizon as $\gamma=4$ and $T-t=10$ years.%
\end{small}%
\noindent
\end{figure}%

In Figure \ref{fig:HN_performance_interest_rate}, we plot the average excess
return mean (left panel), volatility (middle panel), and maximum drawdown
(right panel) at different levels of interest rate $r$. The other model
parameters are given in (\ref{Heston_parameter}). We set the investment
horizon $T=10$ years and the risk aversion level $\gamma=4$. Same as other
studies, we estimate the investment performance by $10^{4}$ simulation trials.
In each panel, we show the corresponding statistics for three investors:
low-wealth HARA investor with $X_{0}/\overline{x}=1$ (blue circle),
high-wealth HARA investor with $X_{0}/\overline{x}=4$ (red square), and a CRRA
investor (black dashed). The optimal policy of the CRRA investor corresponds
to the limit case of $X_{0}/\overline{x}\rightarrow\infty$ for a HARA
investor, as suggested by the relationship (\ref{HN_hara}).

We now examine how the investment performance changes with interest rate $r$
for the three investors. By the left and middle panels of Figure
\ref{fig:HN_performance_interest_rate}, we see the average excess return mean
and volatility increase in $r$ for the two HARA investors, but remain almost
unchanged for the CRRA investor. Moreover, the increase is much more
significant for the low-wealth HARA investor than that for the high-wealth
one. It can be interpreted as follows. First, by the discussion in Section
\ref{sec:impact_performance}, a higher interest rate decreases the prices of
bonds held by HARA investors. Second, with a higher interest rate, the
investor's wealth increases faster as her savings account enjoys a higher
return. Both effects lead to more allocation on the risky asset for the HARA
investors, which produces higher average excess return mean and volatility.
Moreover, by Figure \ref{fig:HN_policy_vary}, the optimal policy of the
low-wealth HARA investor is more sensitive to the wealth level, as her wealth
constraint is more binding. Thus, the impact of $r$ is larger for the
low-wealth HARA investor than that for the high-wealth one. On the other hand,
the optimal policy of the CRRA investor is independent of both the wealth
level and interest rate. So the corresponding average excess return mean and
volatility are almost unaffected by $r$.

In addition, by the right panel of Figure
\ref{fig:HN_performance_interest_rate}, we find that the average maximum
drawdown can react differently to interest rate $r$ for the three investors.
In particular, the average maximum drawdown decreases slightly in $r$ for the
CRRA and high-wealth HARA investors, but increases significantly for the
low-wealth HARA investor. We can explain such different responses by the dual
role played by interest rate in affecting the maximum drawdown. On one hand, a
higher $r$ increases the return of investor's portfolio,\footnote{To see this,
note that the instantaneous expected return of portfolio is given by
$r+\lambda\pi_{t}V_{t}$.} leading to a smaller maximum drawdown. On the other
hand, as $r$ increases, the wealth of HARA investors accumulates at a faster
speed. Thus, they would invest more aggressively on the risky asset. It
introduces more risk and thus increases the maximum drawdown, as discussed in
Section \ref{sec:impact_performance}. For the low-wealth HARA investor, the
second effect dominates as her optimal policy is more sensitive to the wealth
level. It produces the increasing pattern of average maximum drawdown in the
interest rate. However, for the CRRA and high-wealth HARA investors, the first
effect outweighs the second as their optimal policy is less impacted by the
wealth level. It explains the mild decreasing pattern of average maximum
drawdown in the right panel of Figure \ref{fig:HN_performance_interest_rate}.

Next, Figure \ref{fig:HN_performance_horizon} shows how the investment horizon
impacts the investment performance for the three investors. The average excess
return mean, volatility, and maximum drawdown are plotted in the left, middle,
and right panel. The model parameters are set as (\ref{Heston_parameter}) and
the risk aversion is $\gamma=4$. By the left and middle panels, we see the
average excess return mean and volatility increase in the investment horizon
$T$ for the two HARA investors, but stay almost unchanged for the CRRA
investor.\footnote{The very small increase in the average excess return mean
for the CRRA investor is due to the increase in the price of risk hedge
component when investment horizon becomes longer; see the right panel of
Figure \ref{fig:HN_policy_vary}.} In addition, the magnitudes are much larger
for the low-wealth HARA investor. We can explain these patterns by the wealth
effect as discussed for the impact of interest rate. With longer investment
horizons, investors can accumulate more wealth as time goes by, leading to
higher wealth levels. Besides, by the analysis in Secion
\ref{Section_static_impact}, a longer investment horizon makes the bond prices
lower. Both effects increase the allocation on the risky asset for the HARA
investors, which translates to higher average excess return mean and
volatility. Again, these effects are more significant for the low-wealth HARA
investor, leading to a larger increase in the average excess return mean and volatility.%

\begin{figure}[H]%
%

\begin{center}%
\begin{tabular}
[c]{c}%
{\includegraphics[
trim=0.919066in 0.000000in 0.920285in 0.000000in,
height=2.1943in,
width=6.3885in
]%
{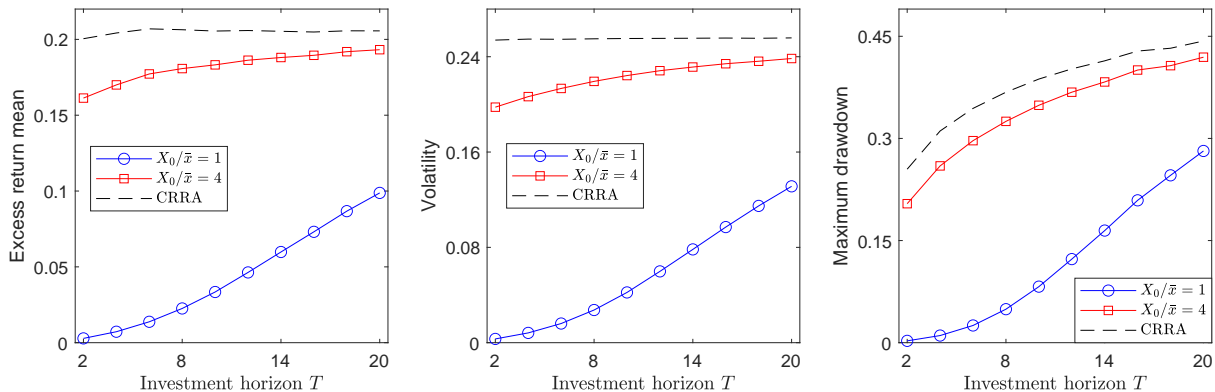}%
}
\end{tabular}%
\caption
{Performance statistics of dynamic optimal portfolio allocation with different investment horizon under the HARA utility.}%
\label{fig:HN_performance_horizon}%
\end{center}%
%

\renewcommand{\baselinestretch}{1.0}%
\begin{small}%
This figure plots the average excess return mean (left), volatility (middle),
and maximum drawdown (right) at different investment horizon $T$ for the
investment problem under HARA utility and our empirically estimated Heston-SV
model given in (\ref{Heston_parameter}). We consider three investors:
low-wealth HARA investor with $X_{0}/\overline{x}=1$ (blue circle),
high-wealth HARA investor with $X_{0}/\overline{x}=4$ (red square), and a CRRA
investor (black dashed). The averages are computed over $10^{4}$ simulated
paths. We set the risk aversion and investment horizon as $\gamma=4$ and
$T-t=10$ years.%
\end{small}%
\noindent
\end{figure}%

Finally, the right panel of Figure \ref{fig:HN_performance_horizon} shows that
the average maximum drawdown increases in the investment horizon $T$ for all
the three investors, and the magnitude of increase is the largest for the
low-wealth HARA investor. Indeed, increasing $T$ leads to a larger average
maximum drawdown for two reasons. First, the maximum drawdown naturally
increases in $T$ as investors are likely to experience bigger loss when they
have a longer investment horizon. Second, as analyzed before, a longer $T$
increases the allocation on the risky asset from HARA investors. It translates
to more uncertainty that also contributes to a larger maximum drawdown. The
second effect is absent for the CRRA investor, but more significant for the
low-wealth HARA investor. Combining these two effects explain the patterns in
the right panel of Figure \ref{fig:HN_performance_horizon} for the average
maximum drawdown.

\subsection{Hysteresis effect in optimal portfolio allocation}

In this section, we demonstrate a novel additional impact from the
wealth-dependent utility: the hysteresis effect in optimal portfolio
allocation. That is, the optimal policy and investment performance depend not
only on the composition of market states (e.g., bear and bull markets) as
exemplified in Figure \ref{fig: HN_simu_path}, but also the sequence of their
occurrences. This effect is fully absent for wealth-independent utilities, and
has not yet been investigated in the literature.

To illustrate this hysteresis effect, we perform the following experiments.
For a given market path of stock price $S_{t}$ and variance $V_{t}$, we
shuffle it by moving the \textquotedblleft good\textquotedblright\ periods to
the beginning of the investment horizon. It is done as follows. Denote the
increments in the Brownian motions by $\Delta_{W,i}^{(k)}=W_{k,i\Delta
}-W_{k,(i-1)\Delta}$, for $i=1,2,...,T/\Delta$ and $k\in\{1,2\}$. The
increments $\{\Delta_{W,i}^{(k)}\}$ can be viewed as the realization of random
market states. It determines the path of the two Brownian motions, and thus
the stock price and variance processes via (\ref{Heston_St}) --
(\ref{Heston_Vt}). For each year, we compute the annual return of the stock as
$\ln(S_{t}/S_{t-1})$. We select the three years with the highest annual
returns. Then, we move the increments $\Delta_{W,i}^{(k)}$ in these three
years to the beginning of the investment horizon for both the stock price and
its volatility process. That is, we rearrange the sequence of $\{\Delta
_{W,i}^{(k)}\}$, $k\in\{1,2\}$, such that for $j=1,2,3$, the increments in the
year with the $j$-th highest annual return is now moved to the $j$-th year of
the new path. The sequence in the other years remain unchanged. With the
shuffled increments, we construct the new stock price and variance processes
$S_{t}^{\prime}$ and $V_{t}^{\prime}$ following (\ref{Heston_St}) --
(\ref{Heston_Vt}). Thus, the shuffling does not alter the values of the
increments in the two processes, but only changes the sequence of their
occurrences by moving the three years with the best stock performance to the
beginning of the horizon.

Let us analyze how the shuffling of path impacts the optimal policy and
investment performance. In Figure \ref{fig:HN_hysteresis_example}, we plot a
representative path of stock price (upper panel) and its corresponding optimal
policy (lower panel). The paths before (resp.\ after) the shuffling are
represented by the blue solid (resp.\ red dashdotted) lines. The model
parameters are set as (\ref{Heston_parameter}) with the risk aversion level
$\gamma=4$. By the upper panel, we see that the stock price after shuffling
increases faster at the start of the investment horizon. It is not surprising
as our shuffling is to move the \textquotedblleft good\textquotedblright%
\ periods ahead. On the other hand, the terminal price $S_{T}^{\prime}$
remains almost the same after the shuffling since the values of the increments
are not changed. The rearrangement of the sequence of market states has
notable impact on the optimal policy. By the lower panel, we see that the
optimal policy from the new path, i.e., with \textquotedblleft
good\textquotedblright\ years ahead, is generally larger than that from the
original path, and the difference is especially significant in the first
several years. The explanation is the wealth effect of HARA utility. In the
shuffled path with \textquotedblleft good\textquotedblright\ periods ahead,
investors will accumulate more wealth at the start of their horizon compared
with the original path. It leads to larger optimal policy under the HARA
utility, as we show in Section \ref{Section_static_impact}.%

\begin{figure}[H]%
%

\begin{center}%
\begin{tabular}
[c]{c}%
{\includegraphics[
trim=0.325055in 0.000000in 0.326447in 0.000000in,
height=4.5587in,
width=5.6181in
]%
{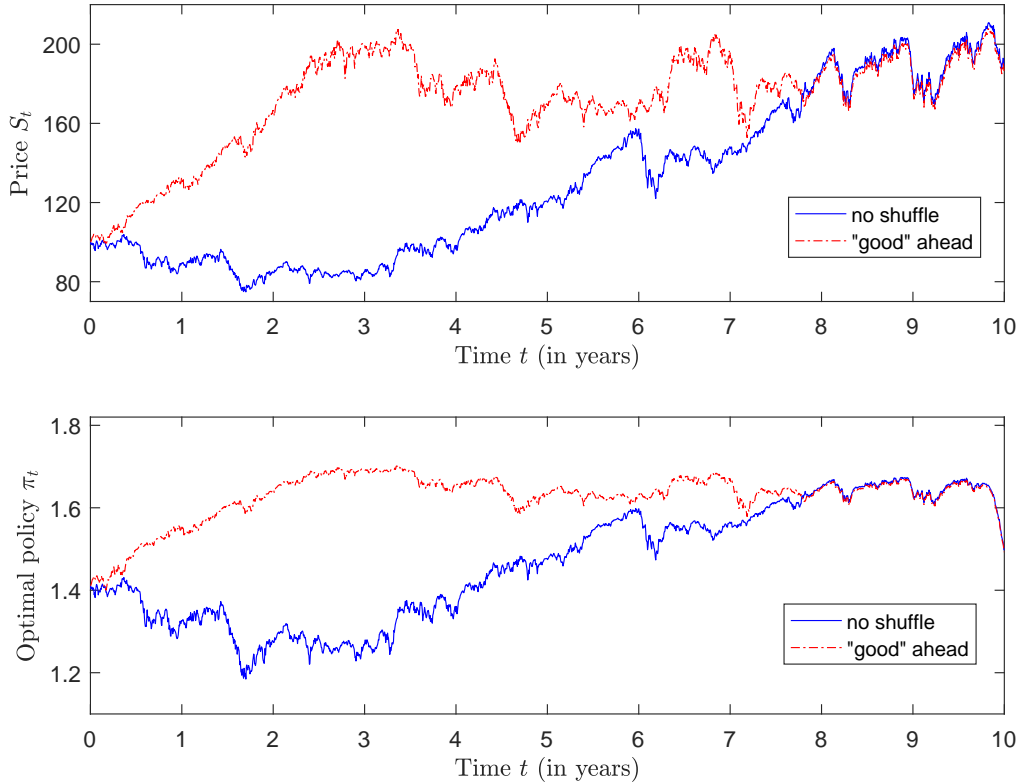}%
}
\end{tabular}%
\caption{Simulated paths of price $S_t$ and optimal policy $\pi
_t$ with and without shuffling.}%
\label{fig:HN_hysteresis_example}%
\end{center}%
%

\renewcommand{\baselinestretch}{1.0}%
\begin{small}%
This figure plots the simulated paths of stock price $S_{t}$ (upper panel) and
optimal policy $\pi_{t}$ (lower panel) under the HARA utility and our
empirically estimated Heston-SV model given in (\ref{Heston_parameter}). The
risk aversion and investment horizon are set as $\gamma=4$ and $T-t=10$ years.
We show the original paths (blue solid) and the shuffled paths (red
dashdotted) with \textquotedblleft good\textquotedblright\ periods moved
ahead.
\end{small}%
\noindent
\end{figure}%

The comparison in Figure \ref{fig:HN_hysteresis_example} highlights the
hysteresis effect in asset allocation under the HARA utility, i.e., the
sequence of market states matters. This effect essentially stems from the
wealth-dependent property of the HARA utility. The hysteresis effect can also
affect the investment performance. For example, with more allocation on the
risky asset, the HARA investor is expected to bear more risk along the
shuffled path in Figure \ref{fig:HN_hysteresis_example}. In the following, we
quantify such impact using simulations with a large number of trials. In
Figure \ref{fig:HN_hysteresis_performance}, we plot the average excess return
mean (left panel) and volatility (right panel) for HARA investors with
different initial wealth levels under three scenarios. The black dashed line
shows the results from original paths without any shuffling, i.e., same as the
results in the left panel of Figure \ref{fig:HN_hysteresis_performance}. The
blue circle (resp.\ red square) line reports the corresponding results after
we shuffle each path by moving the increments in the three years with the
highest (resp.\ lowest) annual stock returns to the beginning of the horizon.
The points on the far right in each panel represent the estimates for a CRRA
investor. It coincides with the limit case as the initial wealth level goes to
infinity. Same as the other studies, each average point in the figure is
estimated by $10^{4}$ simulation trials.%

\begin{figure}[H]%
%

\begin{center}%
\begin{tabular}
[c]{c}%
{\includegraphics[
trim=0.685643in 0.000000in 0.922312in 0.000000in,
height=2.763in,
width=6.0025in
]%
{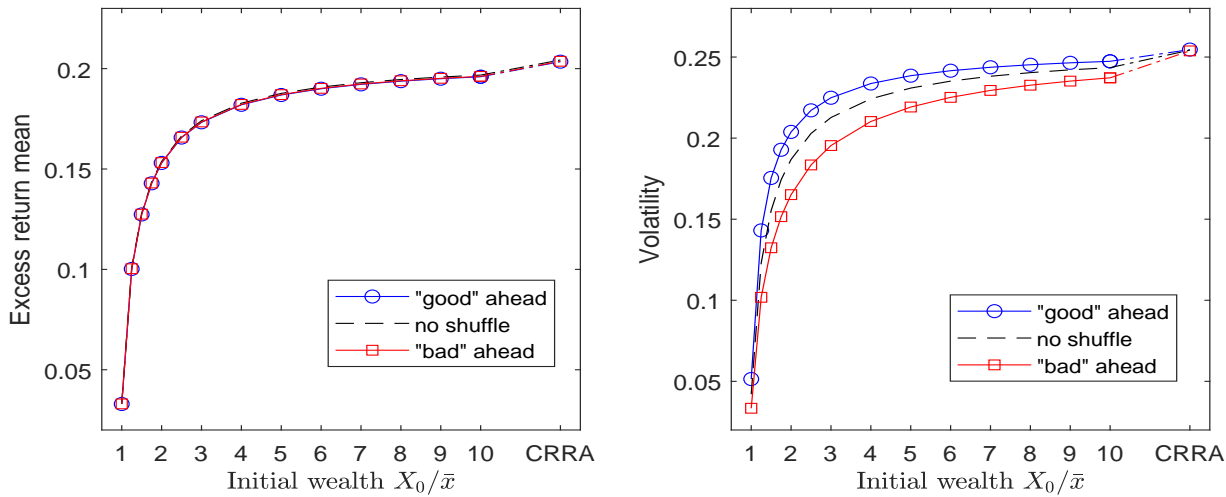}%
}
\end{tabular}%
\caption
{Performance statistics of dynamic optimal portfolio allocation with different initial wealth levels under the HARA utility.}%
\label{fig:HN_hysteresis_performance}%
\end{center}%
\vspace{-5pt}
\renewcommand{\baselinestretch}{1.0}%
\begin{small}%
This figure plots the average excess return mean (left) and volatility (right)
for the investment problem under HARA utility and our empirically estimated
Heston-SV model in (\ref{Heston_parameter}). The corresponding levels under
CRRA utility are represented by point at the right ends. Three scenarios are
considered: \textquotedblleft good\textquotedblright\ periods ahead (blue
circle), \textquotedblleft bad\textquotedblright\ periods ahead (red square),
and without shuffling (black dashed). The averages are computed by $10^{4}$
simulated paths. We set the risk aversion and investment horizon as $\gamma=4$
and $T-t=10$ years.%
\end{small}%
\noindent
\end{figure}%

By the left panel of Figure \ref{fig:HN_hysteresis_performance}, we find that
the average excess return means are almost identical in the three scenarios,
i.e., unaffected by the shuffling. The reason for this seemingly surprising
result is the following. We take the scenario with \textquotedblleft
good\textquotedblright\ periods ahead for illustration (blue circle line). In
this scenario, HARA investors tend to allocate more on the risky asset as
shown by the lower panel of Figure \ref{fig:HN_hysteresis_example}, leading to
higher excess return means. However, in the shuffled path with
\textquotedblleft good\textquotedblright\ periods ahead, investors are more
likely to experience negative shocks in the subsequent years of their
investment horizon. When the negative shocks happen, the larger weights on the
risky asset magnify the loss and decrease the excess return mean. These two
effects offset each other and the average excess return mean remains almost
unchanged. A similar discussion also applies to the scenario where we move the
\textquotedblleft bad\textquotedblright\ periods ahead. Thus, while the
sequence of market states matters for the optimal policy, it almost does not
impact the average excess return mean.

On the contrary, by the right panel of Figure
\ref{fig:HN_hysteresis_performance}, we see that the sequence of market states
does impact the average excess return volatility. In particular, for a given
initial wealth level, the scenario with \textquotedblleft
good\textquotedblright\ (resp. ``bad'') periods ahead leads to the largest
(resp.\ smallest) volatility. Such variation highlights the hysteresis effect
in optimal portfolio allocation under the HARA utility and its impact on
investment performance. This pattern can be interpreted following our
discussion for Figure \ref{fig:HN_hysteresis_example}. When the
\textquotedblleft good\textquotedblright\ (resp.\ \textquotedblleft
bad\textquotedblright) periods are moved ahead, investors will accumulate more
(resp.\ less) wealth at the beginning of their horizon; it leads to more
(resp.\ less) allocation on the risky asset under the HARA utility, and
produces a larger (resp.\ smaller) volatility. On the other hand, the average
volatilities are almost the same under the three scenarios for a CRRA
investor, as shown by the points on the right end of the panel. That is, the
hysteresis effect entirely vanishes under the CRRA utility. It is because the
optimal policy under the CRRA utility is independent of the market path.

\section{Conclusions and discussions\label{Section_conclusions_discussions}}

This paper establishes and implements a novel decomposition of the optimal
policy under general incomplete-market diffusion models with flexible
wealth-dependent utilities. The decomposition, as an indispensable
generalization of the existing complete-market policy decomposition in, e.g.,
\cite{merton71} and \cite{Detemple_2003_JF}, contains four components: the
mean-variance component, the interest rate hedge component, and two components
for hedging uncertainties in market and investor-specific price of risk. The
structural clarity of our decomposition reveals how investor's wealth level
impacts the optimal policy under incomplete market models. It facilitates the
implementation of optimal policy via closed-form solutions or potential
numerical approaches.

We apply our decomposition under the wealth-dependent HARA utility. We show
that with nonrandom interest rate, the optimal policy under HARA utility can
be decomposed as a bond holding scheme and the corresponding CRRA strategy. It
further leads to closed-form solution under specific models. To demonstrate
the wealth effects in optimal portfolio allocation, we analyze the behavior of
optimal policy for HARA investors in a typical incomplete market featuring
stochastic volatility. With the model parameters calibrated from US data, we
find the wealth-dependent property leads to sophisticated cycle-dependence for
optimal policy, as well as a substantial risk-return trade-off and a novel
hysteresis effect in investment performance.

We can adapt or generalize our decomposition for optimal portfolio policies to
other settings, e.g., the forward measure based representation considered in
\cite{DetempleRindisbacher2010}. Moreover, it is interesting, among other
possible extensions, to consider other (exotic) types of market
incompleteness, e.g., the short-selling constraint or the \textquotedblleft
rectangular\textquotedblright\ constraint considered in
\cite{Cvitanic_Karatzas_1992_AAP} and/or \cite{Detemple_Rindisbacher_2005_MF},
as well as the presence of jumps considered in, e.g.,
\cite{Ait_Sahalia_2009_AAP} and \cite{JinZhang12}. Another important direction
is to develop decomposition results for the optimal policy under incomplete
market models and general recursive utilities, e.g., those generalizing the
wealth-independent recursive utilities employed in the literature and thus
leading to wealth-dependent optimal policies.

\newpage
\bibliographystyle{elsevier}
\bibliography{APEC_bib}

\newpage%

%

\setcounter{section}{0}
\setcounter{subsection}{0}
\setcounter{equation}{0}
\renewcommand{\thesection}{Appendix \Alph{section}}
\renewcommand\thesubsection{\thesection.\arabic{subsection}}
\renewcommand{\theequation}{\Alph{section}.\arabic{equation}}%

\section{Degeneration of decomposition under wealth-independent
utility\label{section:impact_utility}}

As briefly discussed in Section \ref{section:crra_policy}, we further document
in Corollaries \ref{Corollary_CRRA} and \ref{Corollary_CRRA_pitheta} below on
how our general policy decomposition in Theorems \ref{thm_representation_new}
and \ref{thm_integral_equation} degenerates under CRRA utility. The comparison
demonstrates the structural impact of wealth-dependent utility on optimal
policy in general models. We first provide the explicit result for optimal
policy under CRRA utility, then discuss in detail its fundamental difference
with the optimal policy under general wealth-dependent utilities.

\begin{corollary}
\label{Corollary_CRRA} Under CRRA utility, the interest rate hedge and price
of risk hedge components are given by
\begin{equation}
\pi^{r}(t,Y_{t})=-(\sigma(t,Y_{t})^{+})^{\top}\frac{E_{t}[\mathcal{\tilde{H}%
}_{t,T}^{r}]}{E_{t}[\mathcal{\tilde{G}}_{t,T}]}\text{ and }\pi^{\theta
}(t,Y_{t})=-(\sigma(t,Y_{t})^{+})^{\top}\frac{E_{t}[\mathcal{\tilde{H}}%
_{t,T}^{\theta}]}{E_{t}[\mathcal{\tilde{G}}_{t,T}]}. \label{pi_CRRA}%
\end{equation}
The functions $\mathcal{\tilde{H}}_{t,T}^{\theta}$, $\mathcal{\tilde{H}}%
_{t,T}^{r}$, and $\mathcal{\tilde{G}}_{t,T}$ are defined as
\begin{subequations}
\begin{align}
\tilde{\mathcal{H}}_{t,T}^{\theta}  &  :=\left(  1-\frac{1}{\gamma}\right)
\left[  (1-w)^{\frac{1}{\gamma}}e^{-\frac{\rho T}{\gamma}}\left(  \xi
_{t,T}^{\mathcal{S}}\right)  ^{1-\frac{1}{\gamma}}H_{t,T}^{\theta}+w^{\frac
{1}{\gamma}}\int_{t}^{T}e^{-\frac{\rho s}{\gamma}}\left(  \xi_{t,s}%
^{\mathcal{S}}\right)  ^{1-\frac{1}{\gamma}}H_{t,s}^{\theta}ds\right]
,\label{crra_Ht}\\
\tilde{\mathcal{H}}_{t,T}^{r}  &  :=\left(  1-\frac{1}{\gamma}\right)  \left[
(1-w)^{\frac{1}{\gamma}}e^{-\frac{\rho T}{\gamma}}\left(  \xi_{t,T}%
^{\mathcal{S}}\right)  ^{1-\frac{1}{\gamma}}H_{t,T}^{r}+w^{\frac{1}{\gamma}%
}\int_{t}^{T}e^{-\frac{\rho s}{\gamma}}\left(  \xi_{t,s}^{\mathcal{S}}\right)
^{1-\frac{1}{\gamma}}H_{t,s}^{r}ds\right]  , \label{crra_Hr}%
\end{align}
and%
\begin{equation}
\tilde{\mathcal{G}}_{t,T}:=(1-w)^{\frac{1}{\gamma}}e^{-\frac{\rho T}{\gamma}%
}\left(  \xi_{t,T}^{\mathcal{S}}\right)  ^{1-\frac{1}{\gamma}}+w^{\frac
{1}{\gamma}}\int_{t}^{T}e^{-\frac{\rho s}{\gamma}}\left(  \xi_{t,s}%
^{\mathcal{S}}\right)  ^{1-\frac{1}{\gamma}}ds, \label{crra_G}%
\end{equation}
with $\xi_{t,s}^{\mathcal{S}},$ $H_{t,s}^{\theta},$ and $H_{t,s}^{r}$ evolving
according to (\ref{SDE_xi_crra}), (\ref{SDE_Htheta_CRRA}), and
(\ref{thm1_SDE_Hr}). The investor-specific price of risk $\theta^{u}\left(
v,Y_{v};T\right)  $ satisfies the following $d-$dimensional equation
\end{subequations}
\begin{equation}
\theta^{u}\left(  v,Y_{v};T\right)  =\frac{\sigma(v,Y_{v})^{+}\sigma
(v,Y_{v})-I_{d}}{E[\tilde{\mathcal{Q}}_{v,T}|Y_{v}]}\times(E[\mathcal{\tilde
{H}}_{v,T}^{r}|Y_{v}]+E[\mathcal{\tilde{H}}_{v,T}^{\theta}|Y_{v}]),
\label{integral_crra}%
\end{equation}
where $\tilde{\mathcal{Q}}_{v,T}=-\tilde{\mathcal{G}}_{v,T}/\gamma$, as
defined by (\ref{crra_G}), i.e.,%
\begin{equation}
\tilde{\mathcal{Q}}_{v,T}:=-\frac{1}{\gamma}\left[  (1-w)^{\frac{1}{\gamma}%
}e^{-\frac{\rho T}{\gamma}}\left(  \xi_{v,T}^{\mathcal{S}}\right)
^{1-\frac{1}{\gamma}}+w^{\frac{1}{\gamma}}\int_{v}^{T}e^{-\frac{\rho s}%
{\gamma}}\left(  \xi_{v,s}^{\mathcal{S}}\right)  ^{1-\frac{1}{\gamma}%
}ds\right]  . \label{crra_Q}%
\end{equation}
The investor-specific price of risk $\theta^{u}\left(  v,Y_{v};T\right)  $ is
determined by a multidimensional equation system consisting of equation
(\ref{integral_crra}), as well as the SDEs of $Y_{s},$ $\xi_{t,s}%
^{\mathcal{S}},$ $H_{t,s}^{r},$ $H_{t,s}^{\theta},$ and $\mathcal{D}_{it}%
Y_{s}$ given in (\ref{SDE_Y}), (\ref{SDE_xi_crra}), (\ref{thm1_SDE_Hr}),
(\ref{SDE_Htheta_CRRA}), and (\ref{thm1_SDE_DYM}), which are all independent
of the multiplier $\lambda_{t}^{\ast}.$
\end{corollary}

\begin{proof}
See Section S.3 in the online supplementary material.
\end{proof}

Though as a special case, the explicit decomposition in Corollary
\ref{Corollary_CRRA} is new relative to the existing analysis on the structure
of the optimal policy under CRRA utility by means of, e.g., HJB equations
(see, e.g., \cite{liu2007portfolio}). By comparing the decomposition in
Corollary \ref{Corollary_CRRA} and that in Theorems
\ref{thm_representation_new} and \ref{thm_integral_equation}, we explicitly
observe how the structure of the optimal policy under wealth-independent CRRA
utility differs from that under general wealth-dependent utilities, as well as
how the specific structure of CRRA utility allows for significant
simplification of the decomposition. This comparative study demonstrates again
the importance of our explicit decomposition results in Theorems
\ref{thm_representation_new} and \ref{thm_integral_equation}.

First, the building blocks employed in these two decompositions obviously have
different structures. Since investor-specific price of risk does not depend on
the multiplier $\lambda_{s}^{\ast}$ for the CRRA case, it takes the form
$\theta_{s}^{u}=\theta^{u}(s,Y_{s};T)$. By the analysis similar to those
immediately prior to Theorem \ref{thm_representation_new}, $\theta_{s}^{u}$
here is consequently independent of the wealth level $X_{s}.$ Thus, the market
completion under CRRA utility enjoys a simpler mechanism. We now compare the
dynamics of $\xi_{t,s}^{\mathcal{S}}(\lambda_{t}^{\ast})$ in
(\ref{thm1_SDE_xi_incomp_explicit}) with that of $\xi_{t,s}^{\mathcal{S}}$ in
(\ref{SDE_xi_crra}), as well as the dynamics of $H_{t,s}^{\theta}(\lambda
_{t}^{\ast})$ in (\ref{thm1_SDE_Htheta}) with that of $H_{t,s}^{\theta}$ in
(\ref{SDE_Htheta_CRRA}). Owing to the absence of multiplier $\lambda_{v}%
^{\ast}$ from $\theta^{u}\left(  v,Y_{v},\lambda_{v}^{\ast};T\right)  ,$
dynamics (\ref{SDE_xi_crra}) and (\ref{SDE_Htheta_CRRA}) are obviously simpler
than (\ref{thm1_SDE_xi_incomp_explicit}) and (\ref{thm1_SDE_Htheta}). In
particular, they are all independent of multiplier $\lambda_{t}^{\ast}.$

Next, our decomposition results illustrate how the current wealth level
impacts the optimal policy under general wealth-dependent utilities, but not
under CRRA utility.\footnote{This wealth-independent property of CRRA utility
reconciles the conclusions in \cite{Detemple_2003_JF} and
\cite{OconeKaratzas1991} for complete market models.} By Theorems
\ref{thm_representation_new} and \ref{thm_integral_equation} for general
wealth-dependent utilities, the current wealth level $X_{t}$ impacts the
optimal policy through two channels. First, it directly appears in the optimal
policy as the denominator in (\ref{pimv_thm1}) -- (\ref{pitheta_thm1}).
Second, due to the wealth equation $X_{t}=E_{t}[\mathcal{G}_{t,T}(\lambda
_{t}^{\ast})]$, $X_{t}$ is implicitly involved in the optimal policy through
the time--$t$ multiplier $\lambda_{t}^{\ast}$ in the functions $\mathcal{Q}%
_{t,T}\left(  \lambda_{t}^{\ast}\right)  $, $\mathcal{H}_{t,T}^{r}(\lambda
_{t}^{\ast})$, and $\mathcal{H}_{t,T}^{\theta}(\lambda_{t}^{\ast})$ with the
building blocks $\xi_{t,s}^{\mathcal{S}}(\lambda_{t}^{\ast})$ and
$H_{t,s}^{\theta}(\lambda_{t}^{\ast})$. However, both channels are absent
under CRRA utility, thanks to its special structure. As shown in
(\ref{pmv_crra_explicit}) and (\ref{pi_CRRA}), both $X_{t}$ and $\lambda
_{t}^{\ast}$ vanish in the components of the optimal policy. Furthermore, by
(\ref{SDE_xi_crra}) and (\ref{SDE_Htheta_CRRA}), the building blocks
$\xi_{t,s}^{\mathcal{S}}$ and $H_{t,s}^{\theta}$ are also independent of
$\lambda_{t}^{\ast}$. Such an independence guarantees that $X_{t}$ is not
implicitly involved in the optimal policy through $\lambda_{t}^{\ast}$ as in
the case with general wealth-dependent utilities. In essence, it is again
because the investor-specific price of risk $\theta^{u}(s,Y_{s};T)$ does not
depend on $\lambda_{s}^{\ast}$ under CRRA utility.

In addition, the decomposition of the price of risk hedge component, as given
in (\ref{thm_decomp2}) for general cases, can be simplified under CRRA
utility, as shown in Corollary \ref{Corollary_CRRA_pitheta} below.

\begin{corollary}
\label{Corollary_CRRA_pitheta}Under the incomplete market model
(\ref{SDE_price}) -- (\ref{SDE_Y}) and the CRRA utility function given in
(\ref{hara_utility}) with $\overline{x}=\overline{c}=0$, the price of risk
hedge component $\pi^{\theta}(t,Y_{t})$ in (\ref{crra_decompose}) can be
further decomposed as
\begin{equation}
\pi^{\theta}(t,Y_{t})=\pi^{h,Y}(t,Y_{t})+\pi^{u,Y}(t,Y_{t});
\label{crra_decompose2}%
\end{equation}
where the first and second components hedge the uncertainty in market and
investor-specific price of risk, due to variation of the state variable
$Y_{t}$. They are given by \
\[
\pi^{h,Y}(t,Y_{t})=-(\sigma(t,Y_{t})^{+})^{\top}\frac{E_{t}[\mathcal{\tilde
{H}}_{t,T}^{h,Y}]}{E_{t}[\mathcal{\tilde{G}}_{t,T}]}\text{ and }\pi
^{u,Y}(t,Y_{t})=-(\sigma(t,Y_{t})^{+})^{\top}\frac{E_{t}[\mathcal{\tilde{H}%
}_{t,T}^{u,Y}]}{E_{t}[\mathcal{\tilde{G}}_{t,T}]}.
\]
Here, $\mathcal{\tilde{H}}_{t,T}^{h,Y}$ and $\mathcal{\tilde{H}}_{t,T}^{u,Y}$
are defined in the same way as $\tilde{\mathcal{H}}_{t,T}^{\theta}$ in
(\ref{crra_Ht}), except for replacing $H_{t,s}^{\theta}$ for $t\leq s\leq T$
by $H_{t,s}^{h,Y}$ and $H_{t,s}^{u,Y}$, which follow dynamics in
(\ref{dH_thetah}) and
\begin{equation}
dH_{t,s}^{u,Y}=\mathcal{D}_{t}\theta^{u}(s,Y_{s};T)(\theta^{u}(s,Y_{s}%
;T)ds+dW_{s}). \label{dH_crra_u1}%
\end{equation}
Under the assumption that $\theta^{u}(v,y;T)$ is differentiable in its
argument, the above dynamics can be further expressed as
\begin{equation}
dH_{t,s}^{u,Y}=\left(  \mathcal{D}_{t}Y_{s}\right)  \nabla\theta^{u}%
(s,Y_{s};T)(\theta^{u}(s,Y_{s};T)ds+dW_{s}). \label{dH_crra_u1_exp}%
\end{equation}
\begin{proof}
The proof follows directly by combining the representation of the
investor-specific price of risk $\theta^{u}(s,Y_{s};T)$ under CRRA utility and
decomposition (\ref{thm_decomp2}).
\end{proof}

\end{corollary}

That is, the component $\pi^{u,\lambda}(t,X_{t},Y_{t})$ in (\ref{piu_decomp})
of Proposition \ref{prop:piu_decomp} vanishes under CRRA utility. As discussed
in Section \ref{section:crra_policy}, this is because the investor-specific
price of risk $\theta^{u}(s,Y_{s};T)$ does not depend on the wealth level via
the multiplier $\lambda_{s}^{\ast}$ under CRRA utility.

\section{Decomposition of optimal policy under HARA
utility\label{Appendix_policy_HARA}}

In this section, we establish the decomposition of optimal policy under
general incomplete market model (\ref{SDE_price}) -- (\ref{SDE_Y}) with the
wealth-dependent HARA utility (\ref{hara_utility}). The result is summarized
in the corollary below. It is obtained by directly particularizing Theorems
\ref{thm_representation_new} and \ref{thm_integral_equation} under the HARA utility.

\begin{corollary}
\label{Corollary_HARA}Under the HARA utility (\ref{hara_utility}) with $w>0$,
the investor-specific price of risk $\theta_{v}^{u}$ in
(\ref{thetau_representation_v}) satisfies the following $d-$dimensional
equation
\begin{equation}
\theta^{u}\left(  v,Y_{v},\lambda_{v}^{\ast};T\right)  =\frac{\sigma
(v,Y_{v})^{+}\sigma(v,Y_{v})-I_{d}}{E[\tilde{\mathcal{Q}}_{v,T}(\lambda
_{v}^{\ast})|Y_{v},\lambda_{v}^{\ast}]}\times(E[\mathcal{\tilde{H}}_{v,T}%
^{r}(\lambda_{v}^{\ast})|Y_{v},\lambda_{v}^{\ast}]+E[\mathcal{\tilde{H}}%
_{v,T}^{\theta}(\lambda_{v}^{\ast})|Y_{v},\lambda_{v}^{\ast}]+\left(
\lambda_{v}^{\ast}\right)  ^{\frac{1}{\gamma}}E\left[  \zeta_{v,T}(\lambda
_{v}^{\ast})|Y_{v},\lambda_{v}^{\ast}\right]  ). \label{integral_hara}%
\end{equation}
Here, $\mathcal{\tilde{H}}_{v,T}^{r}(\lambda_{v}^{\ast})$, $\mathcal{\tilde
{H}}_{v,T}^{\theta}(\lambda_{v}^{\ast})$, and $\tilde{\mathcal{Q}}%
_{v,T}(\lambda_{v}^{\ast})$ are defined as (\ref{crra_Hr}), (\ref{crra_Ht}),
and (\ref{crra_Q}), except for replacing $\xi_{v,s}^{\mathcal{S}}$ and
$H_{v,s}^{\theta}$ by the $\lambda_{v}^{\ast}$--dependent version $\xi
_{v,s}^{\mathcal{S}}(\lambda_{v}^{\ast})$ and $H_{v,s}^{\theta}(\lambda
_{v}^{\ast})$, which evolve according to SDEs
(\ref{thm1_SDE_xi_incomp_explicit}) and (\ref{thm1_SDE_Htheta}). Besides,
$\zeta_{v,T}(\lambda_{v}^{\ast})$ is a $d$--dimensional column vector given
by
\begin{subequations}
\begin{equation}
\zeta_{v,T}(\lambda_{v}^{\ast})=\zeta_{v,T}^{r}(\lambda_{v}^{\ast}%
)+\zeta_{v,T}^{\theta}(\lambda_{v}^{\ast}), \label{zeta_hara}%
\end{equation}
where%
\begin{align}
\zeta_{v,T}^{r}(\lambda_{v}^{\ast})  &  :=\bar{x}\xi_{v,T}^{\mathcal{S}%
}(\lambda_{v}^{\ast})H_{v,T}^{r}+\bar{c}\int_{v}^{T}\xi_{v,s}^{\mathcal{S}%
}(\lambda_{v}^{\ast})H_{v,s}^{r}ds,\label{zeta_r}\\
\zeta_{v,T}^{\theta}(\lambda_{v}^{\ast})  &  :=\bar{x}\xi_{v,T}^{\mathcal{S}%
}(\lambda_{v}^{\ast})H_{v,T}^{\theta}(\lambda_{v}^{\ast})+\bar{c}\int_{v}%
^{T}\xi_{v,s}^{\mathcal{S}}(\lambda_{v}^{\ast})H_{v,s}^{\theta}(\lambda
_{v}^{\ast})ds, \label{zeta_theta}%
\end{align}
with $\bar{x}$ and $\bar{c}$ being the minimum allowable levels for terminal
wealth and intermediate consumption under the HARA utility (\ref{hara_utility}%
). The optimal policy under HARA utility follows by $\pi_{t}=\pi^{mv}%
(t,X_{t},Y_{t})+\pi^{r}(t,X_{t},Y_{t})+\pi^{\theta}(t,X_{t},Y_{t})$, where
\end{subequations}
\begin{subequations}
\begin{equation}
\pi^{mv}(t,X_{t},Y_{t})=-\frac{1}{X_{t}}(\sigma(t,Y_{t})^{+})^{\top}\theta
^{h}(t,Y_{t})\left(  \lambda_{t}^{\ast}\right)  ^{-\frac{1}{\gamma}}%
E_{t}[\tilde{\mathcal{Q}}_{t,T}\left(  \lambda_{t}^{\ast}\right)  ],
\label{pimv_hara}%
\end{equation}
and the hedge components given by:
\begin{align}
\pi^{r}(t,X_{t},Y_{t})  &  =-\frac{1}{X_{t}}(\sigma(t,Y_{t})^{+})^{\top
}\left(  \left(  \lambda_{t}^{\ast}\right)  ^{-\frac{1}{\gamma}}%
E_{t}[\mathcal{\tilde{H}}_{t,T}^{r}\left(  \lambda_{t}^{\ast}\right)
]+E_{t}\left[  \zeta_{t,T}^{r}(\lambda_{t}^{\ast})\right]  \right)
,\label{pir_hara}\\
\pi^{\theta}(t,X_{t},Y_{t})  &  =-\frac{1}{X_{t}}(\sigma(t,Y_{t})^{+})^{\top
}\left(  \left(  \lambda_{t}^{\ast}\right)  ^{-\frac{1}{\gamma}}%
E_{t}[\mathcal{\tilde{H}}_{t,T}^{\theta}\left(  \lambda_{t}^{\ast}\right)
]+E_{t}[\zeta_{t,T}^{\theta}(\lambda_{t}^{\ast})]\right)  .
\label{pitheta_hara}%
\end{align}
The multiplier $\lambda_{t}^{\ast}$ is characterized as the unique solution
for the wealth constraint:
\end{subequations}
\begin{equation}
\left(  \lambda_{t}^{\ast}\right)  ^{-\frac{1}{\gamma}}E_{t}[\tilde
{\mathcal{G}}_{t,T}\left(  \lambda_{t}^{\ast}\right)  ]+\overline{x}%
E_{t}\left[  \xi_{t,T}^{\mathcal{S}}\left(  \lambda_{t}^{\ast}\right)
\right]  +\overline{c}E_{t}\left[  \int_{t}^{T}\xi_{t,s}^{\mathcal{S}}\left(
\lambda_{t}^{\ast}\right)  ds\right]  =X_{t}, \label{constr_hara}%
\end{equation}
where $\tilde{\mathcal{G}}_{t,T}\left(  \lambda_{t}^{\ast}\right)
=-\gamma\tilde{\mathcal{Q}}_{t,T}(\lambda_{t}^{\ast})$, defined by
(\ref{crra_G}) with $\xi_{t,s}^{\mathcal{S}}$ replaced by the $\lambda
_{t}^{\ast}$--dependent version $\xi_{t,s}^{\mathcal{S}}(\lambda_{t}^{\ast})$.
For the case of $w=0$ in utility (\ref{hara_utility}), the above
representation still holds except for dropping the terms related to
$\overline{c}$ in (\ref{zeta_r}), (\ref{zeta_theta}), and (\ref{constr_hara}).
\end{corollary}

Corollary \ref{Corollary_HARA} shows how investor's wealth level gets involved
in the optimal policy when switching from the wealth-indpendent CRRA utility
to the HARA utility. First, by (\ref{integral_hara}), we can see that the
multiplier $\lambda_{v}^{\ast}$ directly appears in the investor-specific
price of risk under the HARA utility via the term $\left(  \lambda_{v}^{\ast
}\right)  ^{\frac{1}{\gamma}}E\left[  \zeta_{v,T}(\lambda_{v}^{\ast}%
)|Y_{v},\lambda_{v}^{\ast}\right]  $, which is fully absent under the CRRA
utility by (\ref{integral_crra}). It introduces the dependence of $\theta
^{u}\left(  v,Y_{v},\lambda_{v}^{\ast};T\right)  $ on the multiplier
$\lambda_{v}^{\ast}$ under the HARA utility. As the multiplier $\lambda
_{v}^{\ast}$ is determined by the wealth level $X_{v}$ via (\ref{XtGt}), the
investor-specific price of risk is wealth-dependent under the HARA utility. By
the dynamics in (\ref{thm1_SDE_xi_incomp_explicit}) and (\ref{thm1_SDE_Htheta}%
), the components $\xi_{v,s}^{\mathcal{S}}(\lambda_{v}^{\ast})$ and
$H_{v,s}^{\theta}(\lambda_{v}^{\ast})$, also depend on the multiplier
$\lambda_{v}^{\ast}$. It further introduces $\lambda_{v}^{\ast}$ into the
functions $\mathcal{\tilde{H}}_{v,T}^{r}\left(  \lambda_{v}^{\ast}\right)  $,
$\mathcal{\tilde{H}}_{v,T}^{\theta}\left(  \lambda_{v}^{\ast}\right)  $,
$\tilde{\mathcal{G}}_{v,T}\left(  \lambda_{v}^{\ast}\right)  $, and
$\tilde{\mathcal{Q}}_{v,T}(\lambda_{v}^{\ast})$ by (\ref{crra_Ht}) --
(\ref{crra_G}) and (\ref{crra_Q}).

The above comparison shows that the term $\left(  \lambda_{v}^{\ast}\right)
^{\frac{1}{\gamma}}E\left[  \zeta_{v,T}(\lambda_{v}^{\ast})|Y_{v},\lambda
_{v}^{\ast}\right]  $ in (\ref{integral_hara}) plays an essential role in
differentiating the optimal policy under the HARA utility from that under the
CRRA utility. By definitions (\ref{zeta_hara}) -- (\ref{zeta_theta}), we see
that $\zeta_{v,T}(\lambda_{v}^{\ast})$ is essentially introduced by the
coefficients $\bar{c}$ and $\bar{x}$, which denote the lower bounds on
intermediate consumption and terminal wealth in (\ref{hara_utility}) for HARA
investors. When $\bar{x}=\bar{c}=0$, the term $\zeta_{v,T}(\lambda_{v}^{\ast
})$ vanishes in (\ref{integral_hara}) as the HARA utility reduces to the CRRA
utility, under which the optimal policy becomes independent of wealth. Our
decomposition suggests that the wealth-dependent property of the HARA utility
is indeed introduced by the wealth constraints $\bar{c}$ and $\bar{x}$.

\bigskip

\section{Proofs of Proposition \ref{prop:hara_opt} and Corollary
\ref{corollary:Heston} \label{Appendix_proofs}}

\subsection{Proof of Proposition \ref{prop:hara_opt}%
\label{Appendix_proof_HARA_prop}}

As a supplement to Proposition \ref{prop:hara_opt}, we first characterize the
investor-specific price of risk $\theta_{v}^{u}$ under HARA utility with
nonrandom interest rate $r_{t}$. By Proposition \ref{prop:hara_opt}, it
coincides with the counterpart under CRRA utility, and thus can be simplified
as $\theta_{v}^{u}=\theta^{u}(v,Y_{v};T)$. It satisfies the following
$d$-dimensional equation:
\begin{equation}
\theta^{u}\left(  v,Y_{v};T\right)  =\frac{\sigma(v,Y_{v})^{+}\sigma
(v,Y_{v})-I_{d}}{E[\tilde{\mathcal{Q}}_{v,T}|Y_{v}]}\times E[\mathcal{\tilde
{H}}_{v,T}^{\theta}|Y_{v}], \label{CRRA_HARA_theta_u_equation_det_r}%
\end{equation}
where $\mathcal{\tilde{H}}_{v,T}^{\theta}$ and $\tilde{\mathcal{Q}}_{v,T}$ are
given by (\ref{crra_Ht}) and (\ref{crra_Q}). In the following, we prove
Proposition \ref{prop:hara_opt} together with the representation
(\ref{CRRA_HARA_theta_u_equation_det_r}).

Next, as a technical preparation, we prove the following lemma.

\begin{lemma}
\label{lemma_xiH}With deterministic interest rate $r_{s}$, the following
relationship holds for general utility functions:%
\begin{equation}
E_{t}[\xi_{t,s}^{\mathcal{S}}\left(  \lambda_{t}^{\ast}\right)  H_{t,s}%
^{\theta}\left(  \lambda_{t}^{\ast}\right)  ]\equiv0_{d},\text{ for any }s\geq
t. \label{conclusion}%
\end{equation}
Here, $\xi_{t,s}^{\mathcal{S}}\left(  \lambda_{t}^{\ast}\right)  $ is the
relative state price density, and $H_{t,s}^{\theta}\left(  \lambda_{t}^{\ast
}\right)  $ is the Malliavin term related to the uncertainty in the total
price of risk, with dynamics explicitly given in
(\ref{thm1_SDE_xi_incomp_explicit}) and (\ref{thm1_SDE_Htheta}).
\end{lemma}

\begin{proof}
By (\ref{def_Xit_incomp}), we get $\xi_{t}^{\mathcal{S}}=\exp(-\int_{0}%
^{t}r_{v}dv-\int_{0}^{t}(\theta_{v}^{\mathcal{S}})^{\top}dW_{v}-\frac{1}%
{2}\int_{0}^{t}(\theta_{v}^{\mathcal{S}})^{\top}\theta_{v}^{\mathcal{S}}dv)$
for the state price density in incomplete markets. We can decompose it into
two parts related to the interest rate and total price of risk, i.e., $\xi
_{t}^{\mathcal{S}}=B_{t}\eta_{t},$ where%
\begin{equation}
B_{t}=\exp\left(  -\int_{0}^{t}r_{v}dv\right)  \ \text{and}\ \eta_{t}%
=\exp\left(  -\int_{0}^{t}\left(  \theta_{v}^{\mathcal{S}}\right)  ^{\top
}dW_{vs}-\frac{1}{2}\int_{0}^{t}\left(  \theta_{v}^{\mathcal{S}}\right)
^{\top}\theta_{v}^{\mathcal{S}}dv\right)  . \label{Bt_etat}%
\end{equation}
With a deterministic interest rate $r_{s}$, the discount term $B_{t}$ is also
deterministic. A straightforward application of Ito's formula leads to the SDE
of $\eta_{t}$ as
\begin{equation}
d\eta_{t}=-\eta_{t}\left(  \theta_{t}^{\mathcal{S}}\right)  ^{\top}dW_{t}.
\label{SDE_eta}%
\end{equation}
The martingale property of $\eta_{t}$ leads to%
\begin{equation*}
E_{t}\left[  \eta_{s}\right]  =\eta_{t},\text{ for any\ }s\geq t.
\end{equation*}
Next, we prove%
\begin{equation}
E_{t}[\eta_{s}H_{t,s}^{\theta}\left(  \lambda_{t}^{\ast}\right)  ]\equiv0_{d},
\label{EetaHtheta}%
\end{equation}
using standard Ito calculus. By the dynamics of $\eta_{s}$ in (\ref{SDE_eta})
and that of $H_{t,s}^{\theta}(\lambda_{t}^{\ast})$ in (\ref{thm1_SDE_Htheta}),
i.e., $dH_{t,s}^{\theta}(\lambda_{t}^{\ast})=M_{t,s}\left(  \lambda_{t}^{\ast
}\right)  [\theta_{s}^{\mathcal{S}}(\lambda_{t}^{\ast})ds+dW_{s}]$ with%
\[
M_{t,s}\left(  \lambda_{t}^{\ast}\right)  =\left(  \mathcal{D}_{t}%
Y_{s}\right)  \nabla\theta^{h}(s,Y_{s})+\mathcal{D}_{t}\theta_{s}^{u}%
(\lambda_{t}^{\ast}),
\]
a straightforward application of the Ito's formula leads to
\[
d(\eta_{s}H_{t,s}^{\theta}\left(  \lambda_{t}^{\ast}\right)  )=(-\eta
_{s}\theta_{s}^{\mathcal{S}}\left(  \lambda_{t}^{\ast}\right)  ^{\top}%
H_{t,s}^{\theta}\left(  \lambda_{t}^{\ast}\right)  +\eta_{s}M_{t,s}\left(
\lambda_{t}^{\ast}\right)  )dW_{s}.
\]
Then, (\ref{EetaHtheta}) follows by $E_{t}[\eta_{s}H_{t,s}^{\theta}\left(
\lambda_{t}^{\ast}\right)  ]=\eta_{t}H_{t,t}^{\theta}\left(  \lambda_{t}%
^{\ast}\right)  \equiv0_{d}$ according to initial condition $H_{t,t}^{\theta
}\left(  \lambda_{t}^{\ast}\right)  \equiv0_{d}$. Finally, relationship
(\ref{conclusion}) comes from $E_{t}[\xi_{t,s}^{\mathcal{S}}\left(
\lambda_{t}^{\ast}\right)  H_{t,s}^{\theta}\left(  \lambda_{t}^{\ast}\right)
]=E_{t}[\xi_{s}^{\mathcal{S}}H_{t,s}^{\theta}]/\xi_{t}^{\mathcal{S}}%
=B_{s}E_{t}[\eta_{s}H_{t,s}^{\theta}\left(  \lambda_{t}^{\ast}\right)
]/\xi_{t}^{\mathcal{S}}=0_{d},$ where the second equality follows from the
deterministic nature of the discount term $B_{s}$ as well as (\ref{EetaHtheta}).
\end{proof}
Now, we are ready to prove Proposition \ref{prop:hara_opt} for the optimal
policy under HARA utility with a deterministic interest rate. Without loss of
generality, we assume $w>0$ in utility (\ref{hara_utility}), as the case of
$w=0$ follows in a similar fashion.\footnote{For the proof under the case of
$w=0,$ we just need to drop all the terms related to $\bar{c}$.}

\begin{proof}
\textit{Part 1:} First, we show that with a deterministic interest rate, the
investor-specific price of risk $\theta^{u}\left(  v,Y_{v},\lambda_{v}^{\ast
};T\right)  $ under HARA utility coincides with its counterpart under CRRA
utility given in Corollary \ref{Corollary_CRRA}, and is thus independent of
multiplier $\lambda_{v}^{\ast}$. To begin with, we employ the dual problem
technique introduced in \cite{HePearson91} to characterize the
investor-specific price of risk $\theta^{u}$ by
\begin{equation}
\inf_{\theta^{u}\in\text{Ker}(\sigma)}E\left[  \int_{0}^{T}\tilde{u}%
(v,\lambda_{v}^{\ast})dv+\tilde{U}(T,\lambda_{T}^{\ast})\right]  ,
\label{dual_HP_1_main}%
\end{equation}
where $\tilde{u}(t,y):=\sup_{x\geq0}(u(t,x)-yx)$ and $\tilde{U}(t,y):=\sup
_{x\geq0}\left(  U(t,x)-yx\right)  $; the constraint $\theta^{u}\in
$Ker$(\sigma)$ abbreviates for $\theta_{v}^{u}\in$Ker$(\sigma(v,Y_{v}))$ for
$0\leq v\leq T$, with Ker$(\sigma(v,Y_{v})):=\{w\in R^{d}:\sigma
(v,Y_{v})w\equiv0_{m}\}$ denoting the kernel of $\sigma(v,Y_{v})$. Using the
explicit forms of functions $I^{u}(t,y)$ and $I^{U}(t,y)$ under HARA utility,
we can explicitly specify the dual problem (\ref{dual_HP_1_main}) conditional
on information up to time $v$ as
\begin{equation}
\inf_{\theta^{u}\in\text{Ker}(\sigma)}E_{v}\left[  (1-w)^{\frac{1}{\gamma}%
}e^{-\frac{\rho T}{\gamma}}\left(  \lambda_{T}^{\ast}\right)  ^{1-\frac
{1}{\gamma}}+w^{\frac{1}{\gamma}}\int_{v}^{T}e^{-\frac{\rho s}{\gamma}}\left(
\lambda_{s}^{\ast}\right)  ^{1-\frac{1}{\gamma}}ds+\frac{\gamma-1}{\gamma
}A_{v,T}\right]  , \label{dual_HARA}%
\end{equation}
where $A_{v,T}=\bar{x}\lambda_{T}^{\ast}+\bar{c}\int_{v}^{T}\lambda_{s}^{\ast
}ds.$ Here, $\bar{x}$ and $\bar{c}$ are the minimum requirements for terminal
wealth and intermediate consumption. On the other hand, under CRRA utility,
the dual problem (\ref{dual_HP_1_main}) specifies to
\begin{equation}
\inf_{\theta^{u}\in\text{Ker}(\sigma)}E_{v}\left[  (1-w)^{\frac{1}{\gamma}%
}e^{-\frac{\rho T}{\gamma}}\left(  \lambda_{T}^{\ast}\right)  ^{1-\frac
{1}{\gamma}}+w^{\frac{1}{\gamma}}\int_{v}^{T}e^{-\frac{\rho s}{\gamma}}\left(
\lambda_{s}^{\ast}\right)  ^{1-\frac{1}{\gamma}}ds\right]
,\ \label{dual_CRRA_v_main}%
\end{equation}
conditioning on information up to time $v$. Comparing (\ref{dual_HARA}) and
(\ref{dual_CRRA_v_main}), we see that the term $A_{v,T}$ in (\ref{dual_HARA})
distinguishes the dual problem under HARA utility from that under CRRA
utility.
With a deterministic interest rate, we then verify that $E_{v}\left[
A_{v,T}\right]  $ does not depend on the control process $\theta_{v}^{u}$ for
$v\in\lbrack v,T]$ and thus can be dropped from the dual problem
(\ref{dual_HARA}) to simplify it as the CRRA counterpart
(\ref{dual_CRRA_v_main}). To see this, we use the relationship $\lambda
_{s}^{\ast}=\lambda_{0}^{\ast}\xi_{s}^{\mathcal{S}}=\lambda_{v}^{\ast}%
\xi_{v,s}^{\mathcal{S}}$ to derive that
\begin{equation}
E_{v}\left[  A_{v,T}\right]  =\bar{x}E_{v}\left[  \lambda_{T}^{\ast}\right]
+\bar{c}\int_{v}^{T}E_{v}\left[  \lambda_{s}^{\ast}\right]  \ ds=\lambda
_{v}^{\ast}\left[  \bar{x}E_{v}[\xi_{v,T}^{\mathcal{S}}]+\bar{c}\int_{v}%
^{T}E_{v}[\xi_{v,s}^{\mathcal{S}}]ds\right]  . \label{EvA}%
\end{equation}
We express the conditional expectation $E_{v}[\xi_{v,s}^{\mathcal{S}}]$ as
$E_{v}[\xi_{v,s}^{\mathcal{S}}]=E_{v}\left[  B_{v,s}\eta_{v,s}\right]  $,
where $B_{v,s}:=B_{s}/B_{v}$ and $\eta_{v,s}:=\eta_{s}/\eta_{v}$ following
(\ref{Bt_etat}). A straightforward application of Ito's formula leads to the
SDE of $\eta_{v,s}$ as $d\eta_{v,s}=-\eta_{v,s}\left(  \theta_{s}%
^{\mathcal{S}}\right)  ^{\top}dW_{s}.$ As we assume a deterministic interest
rate, $B_{v,s}$ is also deterministic. Thus, we have
\begin{equation}
E_{v}[\xi_{v,s}^{\mathcal{S}}]=E_{v}\left[  B_{v,s}\eta_{v,s}\right]
=B_{v,s}E_{v}\left[  \eta_{v,s}\right]  =B_{v,s}, \label{E_xi_hara}%
\end{equation}
where the last equality follows from the martingale property of $\eta_{v,s}$
as a process in $s$ and the fact that $\eta_{v,v}=1$. Plugging $E_{v}\left[
\xi_{v,s}^{\mathcal{S}}\right]  $ into (\ref{EvA}), we obtain that
$E_{v}\left[  A_{v,T}\right]  =\lambda_{v}^{\ast}[\bar{x}B_{v,T}+\bar{c}%
\int_{v}^{T}B_{v,s}\ ds],$ which obviously does not depend on the control
process $\theta_{v}^{u}$. Thus, we can drop the term $A_{v,T}$ from
(\ref{dual_HARA}).
By the above arguments, we show that with a deterministic
interest rate, the investor-specific price of risk $\theta_{v}^{u}$ under HARA
utility is uniquely characterized as the control process for the dual problem
(\ref{dual_CRRA_v_main}) with the underlying Markov process $\left(  Y_{s}%
,\xi_{v,s}^{\mathcal{S}}\right)  $ for $v\leq s\leq T$. Thus, we can verify
that the dual problems, as well as the underlying Markov process, are actually
the same under the HARA and CRRA utilities. They directly lead to the same
optimal control process $\theta_{v}^{u}$. It proves that with a deterministic
interest rate, the investor-specific price of risk $\theta^{u}\left(
v,Y_{v},\lambda_{v}^{\ast};T\right)  $ under HARA utility coincides with its
counterpart under CRRA utility, and is thus independent of the multiplier
$\lambda_{v}^{\ast}$. So, we can express it as $\theta_{v}^{u}=\theta
^{u}\left(  v,Y_{v};T\right)  $ with the same function $\theta^{u}\left(
v,Y_{v};T\right)  $ that satisfies equation (\ref{integral_hara}) under CRRA
utility. Consequently, the quantities $\mathcal{\tilde{H}}_{v,T}^{r}%
(\lambda_{v}^{\ast})$, $\mathcal{\tilde{H}}_{v,T}^{\theta}(\lambda_{v}^{\ast
})$, $\tilde{\mathcal{Q}}_{v,T}(\lambda_{v}^{\ast})$, and $\tilde{\mathcal{G}%
}_{v,T}(\lambda_{v}^{\ast})$ for HARA utility are also independent of the
multiplier $\lambda_{v}^{\ast}$, and coincide with their counterparts under
CRRA utility, which are given in (\ref{crra_Hr}), (\ref{crra_Ht}),
(\ref{crra_Q}), and (\ref{crra_G}).
Next, we establish equation
(\ref{CRRA_HARA_theta_u_equation_det_r}) that governs the investor-specific
price of risk $\theta^{u}\left(  v,Y_{v};T\right)  $. It follows from equation
(\ref{integral_hara}) that, whether the interest rate is deterministic or not,
$\theta^{u}\left(  v,Y_{v},\lambda_{v}^{\ast};T\right)  $ under HARA utility
is characterized by
\begin{equation}
\theta^{u}\left(  v,Y_{v},\lambda_{v}^{\ast};T\right)  =\frac{\sigma
(v,Y_{v})^{+}\sigma(v,Y_{v})-I_{d}}{E[\tilde{\mathcal{Q}}_{v,T}(\lambda
_{v}^{\ast})|Y_{v},\lambda_{v}^{\ast}]}\times(E[\mathcal{\tilde{H}}_{v,T}^{r}(\lambda_{v}^{\ast
})|Y_{v},\lambda_{v}^{\ast}]+E[\mathcal{\tilde{H}}_{v,T}^{\theta}(\lambda_{v}^{\ast}%
)|Y_{v},\lambda_{v}^{\ast}]+\left(  \lambda_{v}^{\ast}\right)  ^{\frac{1}{\gamma}}E\left[
\zeta_{v,T}(\lambda_{v}^{\ast})|Y_{v},\lambda_{v}^{\ast}\right]  ), \label{theta_u_HARA_eqn}%
\end{equation}
where $\zeta_{v,T}(\lambda_{v}^{\ast})=\zeta_{v,T}^{r}(\lambda_{v}^{\ast
})+\zeta_{v,T}^{\theta}(\lambda_{v}^{\ast})$ according to (\ref{zeta_hara}).
With a deterministic interest rate, we have $H_{v,s}^{r}\equiv0_{d}$ due to
(\ref{thm1_SDE_Hr}) and $\nabla r(s,Y_{s})\equiv0_{n}$. Thus, it follows from
(\ref{crra_Hr}) and (\ref{zeta_r}) that $\mathcal{\tilde{H}}_{v,T}^{r}%
(\lambda_{v}^{\ast})\equiv0_{d}\text{ and }\zeta_{v,T}^{r}(\lambda_{v}^{\ast
})\equiv0_{d}.$ Also recall that $\theta^{u}\left(  v,Y_{v},\lambda_{v}^{\ast
};T\right)  $ is independent of $\lambda_{v}^{\ast}$ and thus simplifies to
$\theta^{u}\left(  v,Y_{v};T\right)  .$ Then, we simplify equation
(\ref{theta_u_HARA_eqn}) to
\begin{equation}
\theta^{u}\left(  v,Y_{v};T\right)  =\frac{\sigma(v,Y_{v})^{+}\sigma
(v,Y_{v})-I_{d}}{E[\tilde{\mathcal{Q}}_{v,T}(\lambda_{v}^{\ast})|Y_{v},\lambda_{v}^{\ast}]}%
\times(E[\mathcal{\tilde{H}}_{v,T}^{\theta}(\lambda_{v}^{\ast})|Y_{v},\lambda_{v}^{\ast}]+\left(
\lambda_{v}^{\ast}\right)  ^{\frac{1}{\gamma}}E[  \zeta^{\theta}_{v,T}(\lambda
_{v}^{\ast})|Y_{v},\lambda_{v}^{\ast}]  ), \label{thetau_ap}%
\end{equation}
where $\zeta_{v,T}^{\theta}(\lambda_{v}^{\ast})$ is defined by
(\ref{zeta_theta}) as $\zeta_{v,T}^{\theta}(\lambda_{v}^{\ast})=\bar{x}%
\xi_{v,T}^{\mathcal{S}}(\lambda_{v}^{\ast})H_{v,T}^{\theta}(\lambda_{v}^{\ast
})+\bar{c}\int_{v}^{T}\xi_{v,s}^{\mathcal{S}}(\lambda_{v}^{\ast}%
)H_{v,s}^{\theta}(\lambda_{v}^{\ast})ds.$ By Lemma \ref{lemma_xiH}, its
expectation is always zero under a deterministic interest rate, i.e.,
\begin{equation}
E_{v}[\zeta_{v,T}^{\theta}(\lambda_{v}^{\ast})]=\bar{x}E_{v}[\xi
_{v,T}^{\mathcal{S}}(\lambda_{v}^{\ast})H_{v,T}^{\theta}(\lambda_{v}^{\ast
})]+\bar{c}\int_{v}^{T}E_{v}[\xi_{v,s}^{\mathcal{S}}(\lambda_{v}^{\ast
})H_{v,s}^{\theta}(\lambda_{v}^{\ast})]ds\equiv0_{d}. \label{zeta_zero}%
\end{equation}
Thus, the last term $\left(  \lambda_{v}^{\ast}\right)  ^{\frac{1}{\gamma}%
}E[\zeta_{v,T}^{\theta}(\lambda_{v}^{\ast})|Y_{v},\lambda_{v}^{\ast}]$ vanishes in
(\ref{thetau_ap}), and the equation further simplifies to
\[
\theta^{u}\left(  v,Y_{v};T\right)  =\frac{\sigma(v,Y_{v})^{+}\sigma
(v,Y_{v})-I_{d}}{E[\tilde{\mathcal{Q}}_{v,T}(\lambda_{v}^{\ast})|Y_{v},\lambda_{v}^{\ast}]}\times
E[\mathcal{\tilde{H}}_{v,T}^{\theta}(\lambda_{v}^{\ast})|Y_{v},\lambda_{v}^{\ast}].
\]
By examining the definitions of $\mathcal{\tilde{H}}_{v,T}^{\theta}%
(\lambda_{v}^{\ast})$ and $\tilde{\mathcal{Q}}_{v,T}(\lambda_{v}^{\ast})$ as
well as the SDEs of $\xi_{v,s}^{\mathcal{S}}(\lambda_{v}^{\ast})$ and
$H_{v,s}^{\theta}(\lambda_{v}^{\ast})$ in (\ref{thm1_SDE_xi_incomp_explicit})
and (\ref{thm1_SDE_Htheta}), we confirm that $\mathcal{\tilde{H}}%
_{v,T}^{\theta}(\lambda_{v}^{\ast})$ and $\tilde{\mathcal{Q}}_{v,T}%
(\lambda_{v}^{\ast})$ reduce to $\mathcal{\tilde{H}}_{v,T}^{\theta}$ and
$\tilde{\mathcal{Q}}_{v,T}$ given in (\ref{crra_Ht}) and (\ref{crra_Q}).
Hence, the multiplier $\lambda_{v}^{\ast}$ does not show up in either the
above equation system or its solution $\theta^{u}\left(  v,Y_{v};T\right)  $. It proves equation (\ref{CRRA_HARA_theta_u_equation_det_r}) for $\theta^u(v,Y_v;T)$.
\textit{Part 2:} Next, we look into the optimal policy under HARA utility with
a deterministic interest rate. Under this circumstance, we have $H_{t,s}%
^{r}\equiv0_{d}$ by (\ref{thm1_SDE_Hr}). Thus, it follows from (\ref{crra_Hr}%
), (\ref{zeta_r}), and (\ref{pir_hara}) that the interest hedge component
$\pi_{H}^{r}(t,X_{t},Y_{t})=0_{m}$, i.e., there is no need to hedge interest
rate uncertainty. So, we only need to focus on the mean-variance and price of
risk hedge components. First, we solve for the multiplier $\lambda_{t}^{\ast}$
from the wealth equation (\ref{constr_hara}), i.e., $\left(  \lambda_{t}%
^{\ast}\right)  ^{-\frac{1}{\gamma}}E_{t}[\tilde{\mathcal{G}}_{t,T}%
]+\overline{x}E_{t}[\xi_{t,T}^{\mathcal{S}}]+\overline{c}E_{t}\big[\int%
_{t}^{T}\xi_{t,s}^{\mathcal{S}}ds\big]=X_{t}.$ Here, we drop the dependence on
$\lambda_{t}^{\ast}$ from $\tilde{\mathcal{G}}_{t,T}$ and $\xi_{t,s}%
^{\mathcal{S}}$. It is because we have shown in Part 1 that the
investor-specific price of risk $\theta^{u}\left(  t,Y_{t};T\right)  $ does
not depend on $\lambda_{t}^{\ast},$ and neither do $\tilde{\mathcal{G}}_{t,T}$
and $\xi_{t,s}^{\mathcal{S}}$ according to (\ref{crra_G}) and
(\ref{thm1_SDE_xi_incomp_explicit}). By (\ref{E_xi_hara}), we have $E_{t}%
[\xi_{t,s}^{\mathcal{S}}]=B_{t,s}$. Plugging it to the above equation, we
solve $\left(  \lambda_{t}^{\ast}\right)  ^{-\frac{1}{\gamma}}$ as
\begin{equation}
\left(  \lambda_{t}^{\ast}\right)  ^{-\frac{1}{\gamma}}=\frac{\bar{X}_{t}%
}{E_{t}[\tilde{\mathcal{G}}_{t,T}]}, \label{lambda_start_power}%
\end{equation}
where $\bar{X}_{t}$ is defined in (\ref{xtH_1}), i.e., $\bar{X}_{t}%
=X_{t}-\overline{x}B_{t,T}-\overline{c}\int_{t}^{T}B_{t,s}ds$; $\tilde
{\mathcal{G}}_{t,T}$ is defined in (\ref{crra_G}). Plugging
(\ref{lambda_start_power}) into the mean-variance component in
(\ref{pimv_hara}) and invoking the relationship $\tilde{\mathcal{G}}%
_{t,T}=-\gamma\tilde{\mathcal{Q}}_{t,T}$, we can derive%
\begin{align}
\pi_{H}^{mv}(t,X_{t},Y_{t})  &  =-\frac{1}{X_{t}}(\sigma(t,Y_{t})^{+})^{\top
}\theta^{h}(t,Y_{t})\left(  \lambda_{t}^{\ast}\right)  ^{-\frac{1}{\gamma}%
}E_{t}[\tilde{\mathcal{Q}}_{t,T}]\nonumber\\
&  =-\frac{1}{X_{t}}(\sigma(t,Y_{t})^{+})^{\top}\theta^{h}(t,Y_{t})\frac
{\bar{X}_{t}}{E_{t}[\tilde{\mathcal{G}}_{t,T}]}E_{t}[\tilde{\mathcal{Q}}%
_{t,T}]\nonumber\\
&  =\frac{\bar{X}_{t}}{\gamma X_{t}}(\sigma(t,Y_{t})^{+})^{\top}\theta
^{h}(t,Y_{t}). \label{pimv_determin_r}%
\end{align}
Next, in (\ref{pitheta_hara}), we have
\[
\pi_{H}^{\theta}(t,X_{t},Y_{t})=-\frac{1}{X_{t}}(\sigma(t,Y_{t})^{+})^{\top
}\left(  \left(  \lambda_{t}^{\ast}\right)  ^{-\frac{1}{\gamma}}%
E_{t}[\mathcal{\tilde{H}}_{t,T}^{\theta}]+E_{t}[\zeta_{t,T}^{\theta}]\right)
\break=-\frac{1}{X_{t}}(\sigma(t,Y_{t})^{+})^{\top}\left(  \lambda_{t}^{\ast
}\right)  ^{-\frac{1}{\gamma}}E_{t}[\mathcal{\tilde{H}}_{t,T}^{\theta}]
\]
for the price of risk hedge component, where the second equality follows from
(\ref{zeta_zero}). Plugging (\ref{lambda_start_power}) into the right-hand
side, we obtain%
\begin{equation}
\pi_{H}^{\theta}(t,X_{t},Y_{t})=-(\sigma(t,Y_{t})^{+})^{\top}\frac{\bar{X}%
_{t}}{X_{t}}\frac{E_{t}[\mathcal{\tilde{H}}_{t,T}^{\theta}]}{E_{t}%
[\tilde{\mathcal{G}}_{t,T}]}=(\sigma(t,Y_{t})^{+})^{\top}\frac{\bar{X}_{t}%
}{X_{t}}\frac{E_{t}[\mathcal{\tilde{H}}_{t,T}^{\theta}]}{\gamma E_{t}%
[\tilde{\mathcal{Q}}_{t,T}]}, \label{pitheta_determin_r}%
\end{equation}
where the last equality follows from the relation $\tilde{\mathcal{Q}}%
_{t,T}=-\tilde{\mathcal{G}}_{t,T}/\gamma$ by (\ref{crra_G}) and (\ref{crra_Q}%
). Finally, we can obtain relationships (\ref{H_C_ratio}) by comparing the
optimal components in (\ref{pimv_determin_r}) and (\ref{pitheta_determin_r})
with their counterparts under CRRA utility.
\end{proof}

\subsection{Proof of Corollary \ref{corollary:Heston}
\label{Appendix_proof_Corollary_Heston}}

\begin{proof}
The optimal policy (\ref{HN_hara}) directly follows from the HARA policy
decomposition (\ref{H_C_ratio}) established in Proposition \ref{prop:hara_opt}
as well as the simple fact that the interest rate hedge component is zero for
both CRRA and HARA utilities as the interest rate is constant in the Heston-SV
model. In particular, $\bar{X}_{t}$ in (\ref{Xbar_Heston}) follows from
(\ref{xtH_2}); the decomposition of the policy under CRRA utility in
(\ref{piHN_crra}) follows from Theorems \ref{thm_representation_new} and
\ref{thm_integral_equation} as well as the results in \cite{liu2007portfolio}.
The form of (\ref{thetau_explicit}) for the investor-specific price of risk is
derived as follows. By the definition (\ref{def_thetau}) and the orthogonal
condition (\ref{ortho_cond}), we have
$
\sigma(t,Y_{t})\theta^{u}_{t} \equiv0_{m}.$
Under Heston-SV model with HARA utility, it implies $\sigma(t,V_{t}%
)\theta^{u}(t,V_{t};T)\equiv0$ with $\sigma(t,V_{t})=(\sqrt{V_{t}},0)$.
Combining (\ref{CRRA_HARA_theta_u_equation_det_r}) and (\ref{pi_CRRA}) as well
as plugging in $\pi_{C}^{\theta}(t,V_{t})=-\rho\sigma\delta\phi(T-t)$, we
obtain $\theta_{2}^{u}(t,V_{t};T)$ in closed-form as (\ref{thetau_explicit}).
\end{proof}

\end{document}